\documentclass{amsproc}

\usepackage[british]{babel}
\addto\captionsbritish{}

\usepackage{amsmath, amssymb, amsthm, mathtools}

\newtheorem{theorem}{Theorem}[section]

\newtheorem{lemma}[theorem]{Lemma}

\newtheorem*{theorem*}{Theorem}
\newtheorem*{proposition*}{Proposition}
\newtheorem*{lemma*}{Lemma}
\newtheorem*{corollary*}{Corollary}
\newtheorem{defi}{Definition}
\theoremstyle{definition}
\newtheorem*{definition}{Definition}
\newtheorem*{definitions}{Definitions}
\newtheorem*{data}{Data}

\usepackage{url, verbatim} 
\usepackage{enumitem}  

\usepackage[dvipsnames]{xcolor}

\usepackage{graphicx}
\usepackage[bf, small]{caption} 
\usepackage{varioref}
\usepackage{subcaption}
\usepackage{tikz}
\usetikzlibrary{calc}

\newcommand{\tikzmath}[2][]
     {\vcenter{\hbox{\begin{tikzpicture}[#1]#2
                     \end{tikzpicture}}}
     }

\def\acts{\hspace{.1cm}{\setlength{\unitlength}{.30mm}\linethickness{.09mm}
                        \begin{picture}(8,8)(0,0)\qbezier(7,6)(4.5,8.3)(2,7)\qbezier(2,7)(-1.5,4)(2,1)\qbezier(2,1)(4.5,-.3)(7,2)
                                                 \qbezier(7,6)(6.1,7.5)(6.8,9)\qbezier(7,6)(5,6.1)(4.2,4.4)
                        \end{picture}\hspace{.1cm}}}

\def \A {\mathcal{A}}
\def \B {\mathcal{B}}
\def \C {\mathcal{C}}
\DeclareMathOperator* \anacont {anal. cont.}

\title{Three-tier CFTs from Frobenius algebras}
       \author{Andr{\'e} Henriques \\ \it Notes compiled by J. Lamers}
      \address{Mathematisch Instituut\\
               Universiteit Utrecht, Postbus 80.010\\
               3508 TA Utrecht, The Netherlands\\}
        \email{a.g.henriques@uu.nl}

\begin{document}

	\maketitle

	\begin{abstract}
	These are lecture notes of a course given at the Summer School on Topology and Field Theories held at the Centre for Mathematics of the University of Notre Dame, Indiana, from May 29 to June 2, 2012.
	
	The idea of extending quantum field theories to manifolds of lower dimension was first proposed by Dan Freed in the nineties. In the case of conformal field theory (\textsc{cft}), we are talking of an extension of the Atiyah-Segal axioms, where one replaces the bordism category of Riemann surfaces by a suitable bordism bicategory, whose objects are points, whose morphism are 1-manifolds, and whose 2-morphisms are pieces of Riemann surface.
	
	There is a beautiful classification of full (rational) \textsc{cft} due to Fuchs, Runkel and Schweigert, which roughly says the following. Fix a chiral algebra $A$ (= vertex algebra). Then the set of full \textsc{cft}s whose left and right chiral algebras agree with $A$ is classified by Frobenius algebras internal to $\text{Rep}\,(A)$. A famous example to which one can successfully apply this is the case where the chiral algebra $A$ is affine $\mathfrak{su}(2)$ at level $k$, for some $k\in\mathbb{N}$. In that case, the Frobenius algebras in $\text{Rep}\,(A)$ are classified by $A_n$, $D_n$, $E_6$, $E_7$, $E_8$, and so are the corresponding \textsc{cft}s.
	
	Recently, Kapustin and Saulina gave a conceptual interpretation of the FRS classification in terms of 3-dimensional Chern-Simons theory with defects. Those defects are also given by Frobenius algebra object in $\text{Rep}\,(A)$.
	Inspired by the proposal of Kapustin and Saulina, we will (partially) construct the three-tier \textsc{cft} associated to a Frobenius algebra object.
	\end{abstract}

\tableofcontents

\section{Introduction}

In these notes we define, and partially construct, extended conformal field theories starting from a so-called chiral conformal field theory, and a Frobenius algebra object.

The idea of \textsl{extended} field theory, which goes back to the work of Freed in the ninetees~\cite{Freed1993}, started in the context of topological field theory. There, it is an extension of Atiyah's definition of topological quantum field theory (\textsc{tqft})~\cite{Atiyah1989} where, instead of just assigning vector spaces to $(d-1)$-dimensional manifolds and linear maps to $d$-dimensional cobordisms, one also assigns data to manifolds of lower dimension, all the way down to points. Thus, the extended field theory consists of $d+1$~tiers.

Extended \textsl{conformal} field theories (\textsc{cft}s) were first proposed by Stolz and Teichner~\cite{ST2004}, in the context of their project of constructing elliptic cohomology, and then also mentioned in a review paper by Segal~\cite{Segal2007}. However, they did not provide any constructions of extended \textsc{cft}s. We will show that this can be done, at least to a great extent.

\subsection{Outline}

Let us briefly outline the content of these notes. In  Section~\ref{s:extended CFT} we introduce (full) \textsc{cft}\footnote{In this paper, ``\textsc{cft}'' will always refer to two-dimensional conformal field theory.} in the formalism of Graeme Segal, and define extended \textsc{cft}. The source and target bicategories of extended \textsc{cft} are discussed in some detail.

Section~\ref{s:chiral vs full} contains a discussion of chiral \textsc{cft}s. 
We introduce the two important ingredients of our construction: conformal nets, and Frobenius algebra objects.
We also recall some aspects of the construction of Fuchs, Runkel and Schweigert, which constructs a (non-extended) full \textsc{cft} from a chiral \textsc{cft} and a Frobenius algebra object in the associated category.

In Section~\ref{s:construction} we describe work in progress: the construction of an extended \textsc{cft}
from a conformal net, and a Frobenius algebra object in the representation category of the conformal net.
We finish by describing the main unsolved problem, namely the construction of the bimodule map that corresponds to a surface with four cusps (the `ninja star' in Figure~\ref{fig:ninja star}). 
If this could be done, this would complete the construction of the full \textsc{cft}.

\begin{figure}[h]
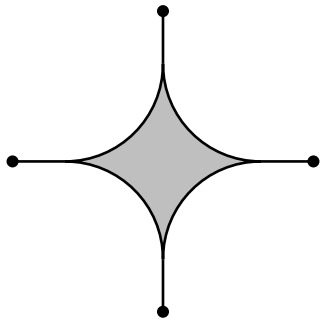

	\begin{center}
\tikz{\draw[line width=1] (1.3,0) -- (2,0)(0,1.3) -- (0,2)(-1.3,0) -- (-2,0)(0,-1.3) -- (0,-2);
\fill (2,0) circle (.08)(-2,0) circle (.08)(0,2) circle (.08)(0,-2) circle (.08);
\filldraw[line width=1, fill=gray!50]
(1.3,0) to[in=-90, out=180] (0,1.3) to[out=-90, in=0] (-1.3,0) to[in=90, out=0] (0,-1.3)
to[out=90, in=180] (1.3,0)
;}
	\end{center}\caption{A `ninja star' is a 2-surface with four cusps. Our main open problem is to construct the corresponding map of bimodules.}\label{fig:ninja star}
\end{figure}

\section{Extended conformal field theory}\label{s:extended CFT}

The definition of extended \textsc{cft} is an extension of Segal's definition of \textsc{cft}.
We start at the beginning, and introduce \textsc{cft} in Segal's formalism.
We will also discuss the notion of conformal welding, which is a necessary ingredient of the definition.

\subsection{Segal's definition of conformal field theory}\label{s:Segal CFT}

There are several (non-equivalent) ways to define conformal field theory. Although Segal's definition~\cite{Segal1988,Segal2004} is not the most mainstream one, it is the one that has become popular amongst mathematicians.

\begin{definition}[Segal]
A full\footnote{
There also exists another notion, called chiral \textsc{cft}.
This will be discussed in Section~\ref{s:chiral vs full}.
Until then, all \textsc{cft}s will be full \textsc{cft}s.} \textit{conformal field theory} is a symmetric monoidal functor from the category of conformal cobordisms, which consists of
	\[ \begin{cases} \ \text{objects: one-dimensional, compact, oriented, smooth manifolds;} \\ \ \text{morphisms: cobordisms equipped with a complex structure;} \\ \ \text{monoidal structure: taking disjoint unions;} \end{cases} \]
to the category of Hilbert spaces, with
	\[ \begin{cases} \ \text{objects: Hilbert spaces;} \\ \ \text{morphisms: bounded linear maps;} \\ \ \text{monoidal structure: usual tensor product of Hilbert spaces.} \end{cases} \]
\end{definition}

Let us take a closer look at the category of conformal cobordisms.
Its objects consist of possibly empty disjoint unions of oriented circles (no parametrizations), always with a smooth structure.
A \textsc{cft} maps a circle to a Hilbert space, referred to as the `state space' by physicists, and a diffeomorphisms between circles to a unitary isomorphism.
It then maps the disjoint unions of $k$ circles to the tensor product of the Hilbert spaces associated to the individual circles.
Finally, the empty manifold, which is the unit object for the monoidal structure, is sent to the trivial Hilbert space~$\mathbb{C}$,
which is the unit for the tensor product in the category of Hilbert spaces. 

The morphisms are Riemann surfaces with boundary.
Both the smooth structure and the complex structure extend all the way to the boundary of the cobordisms. 
Alternatively, one could take the complex structure to only be defined on the interior, 
and require that the cobordism be locally isomorphic to the upper half plane.
The orientations of the one-manifolds have to be compatible with those of the cobordisms connecting them:
if $\Sigma$ is a cobordism from $S$ to $S'$, then by definition there exists an orientation preserving diffeomorphism from the boundary $\partial\Sigma$ of $\Sigma$ to the disjoint union $S\amalg \overline{S'}$ of the `ingoing' manifold~$S$ and the `outgoing' $S'$ with orientation reversed. 

A \textsc{cft} sends cobordisms to linear maps between Hilbert spaces, the `propagator' or `correlator'. In particular, a closed cobordism $\Sigma$ between two empty manifolds is mapped to a linear map $\mathbb{C} \longrightarrow \mathbb{C}$. The latter is completely determined by a single complex number $Z(\Sigma)$, the `partition function' at the Riemann surface.
A priori, the category of conformal cobordisms does not come with identity morphisms: those need to be added by hand, and one can then think of them as infinitesimally thin cobordisms.
One also needs to include diffeomorphisms between 1-manifolds as degenerate cases of conformal cobordisms.
More interesting is to go one step further and allow for cobordisms that are partially thin, and partially thick, such as the one in Figure \ref{fig:part thin}.
\begin{figure}[h]
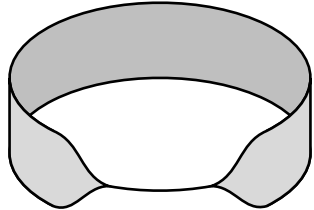

	\begin{center}
\tikz{
\coordinate (a) at ($(0,0)+(-45:2cm and 1cm)$);
\coordinate (b) at ($(0,0)+(225:2cm and 1cm)$);
\coordinate (c) at ($(0,1)+(-45:2cm and 1cm)$);
\coordinate (d) at ($(0,1)+(225:2cm and 1cm)$);
\coordinate (e) at ($(0,.5)+(-70:2cm and 1cm)$);
\coordinate (f) at ($(0,.5)+(-110:2cm and 1cm)$);
\fill[gray!50] (-2,0) arc (180:0:2cm and 1cm) -- +(0,1) arc (0:180:2cm and 1cm) -- cycle;
\draw[line width=1](0,0)+(-45:2cm and 1cm) arc (-45:180+45:2cm and 1cm)
(0,1)+(-45:2cm and 1cm) arc (-45:180+45:2cm and 1cm)
(0,.5)+(-70:2cm and 1cm) arc (-70:-110:2cm and 1cm);
\filldraw[line width=1, fill=gray!30]
(2,0) arc (0:-45:2cm and 1cm) to[out=195, in=0] (e) to[out=20, in=203] (c) arc (-45:0:2cm and 1cm) -- cycle
(-2,0) arc (180:225:2cm and 1cm) to[out=-20, in=180] (f) to[out=160, in=-23] (d) arc (225:180:2cm and 1cm) -- cycle;
}
\end{center}\caption{A partially thin annulus.}\label{fig:part thin}
\end{figure}

Finally, the operator a \textsc{cft} associates to a morphism in the bordism category should depend continuously on the morphism.
What this means is somewhat involved to explain, as it involves constructing a topology 
on the moduli space of all morphisms, encompassing honest bordisms and diffeomorphisms.
More precisely, the dependence should be real-analytic in the interior of the moduli space (honest bordisms), and continuous on the boundary (diffeomorphisms).

\subsubsection{Conformal welding}

\begin{figure}[h]
	\begin{center}
	\includegraphics[scale=.14]{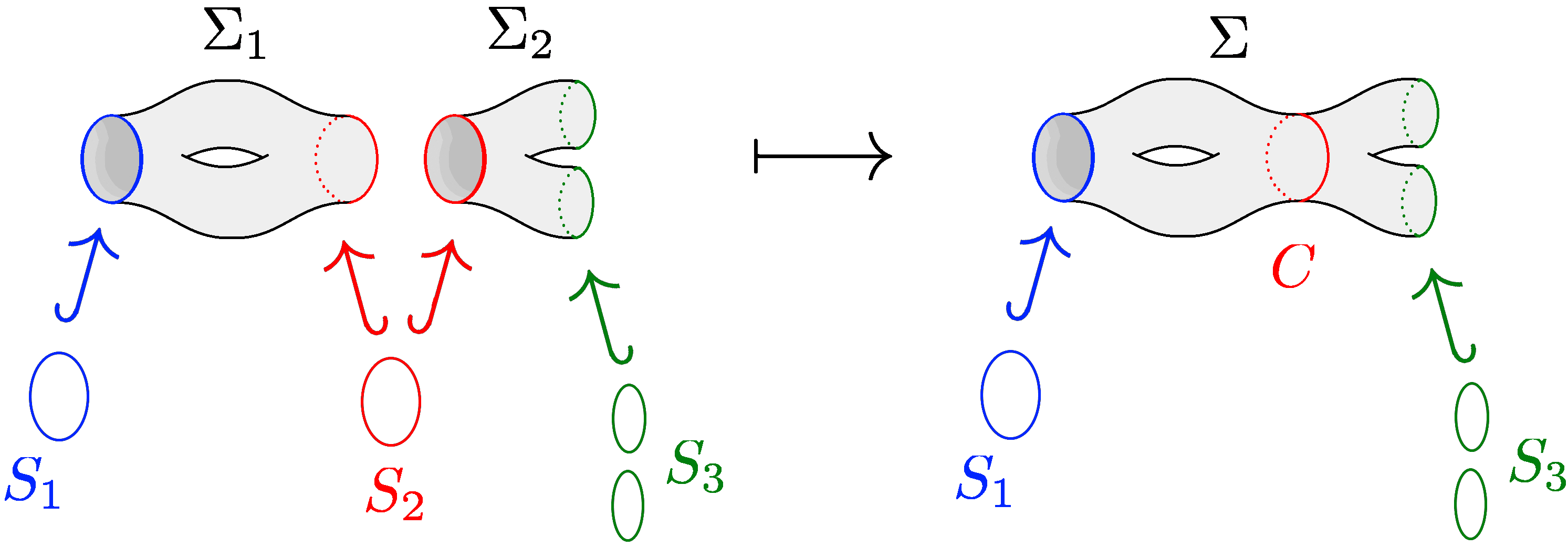}
	\end{center}\caption{Conformal welding of two composable cobordisms.}\label{fig:conf welding}
\end{figure}

The composition of conformal cobordisms is tricky and deserves special attention. 
We outline the procedure, which is called \textit{conformal welding}. 
Consider two composable cobordisms $\Sigma_1$ and $\Sigma_2$ as depicted in Figure~\ref{fig:conf welding}.
As topological spaces, $\Sigma_1$ and $\Sigma_2$ can be glued in the obvious way.
However, a priori, the composition $\Sigma$ is only equipped with a smooth structure \textsl{away} from the curve~$C$ along which $\Sigma_1$ and $\Sigma_2$ have been glued, and similarly for the complex structure.
These issues are resolved by the following theorem~\cite{Segal,RS2006}.

\begin{theorem}\label{thm:conf welding}
In the above situation, there exists a unique complex structure on the interior of the topological manifold~$\Sigma$ which is compatible with the given complex structures on $\Sigma_1$ and $\Sigma_2$. Moreover, the embedding of $C$ into $\Sigma$ is smooth.
\end{theorem} 
Note that the embedding $C \hookrightarrow \Sigma$ will typically not be analytic;
this already signals that the proof of the theorem will have to be rather involved.
A closely related result, which is needed in the proof that conformal welding is well defined, is~\cite{Bell1990}:

\begin{lemma}\label{lem:conf welding}
Let $D \subset \mathbb{C}$ be a connected, simply connected open subset of the complex plane,
and let us assume that the boundary of $D$ is smooth.
Let $D_0 \subset \mathbb{C}$ be the standard disc centered at the origin.
Then the map $D \longrightarrow D_0$ provided by the Riemann mapping theorem is smooth all the way to the boundary.
\end{lemma}

\noindent Since the problem in Theorem \ref{thm:conf welding} is local, one can reduce the general problem of conformal welding to the simpler situation of glueing two discs along a smooth identification of their boundaries.
Moreover, using Lemma~\ref{lem:conf welding},
one can further reduce the problem to that of glueing two standard discs along a smooth identification of their boundaries.
Theorem \ref{thm:conf welding} is therefore equivalent to the following special case of the theorem:
given two standard discs $D_0$ and $D_0'$ in $\mathbb{C}$, and a diffeomorphism~$\varphi$ between their boundaries, the resulting glued surface is a copy of the Riemann sphere $\mathbb{CP}^1$, along with a smoothly embedded curve in it, as shown in Figure~\ref{fig:glueing discs}.

\begin{figure}[h]
	\begin{center}
	\includegraphics[scale=.12]{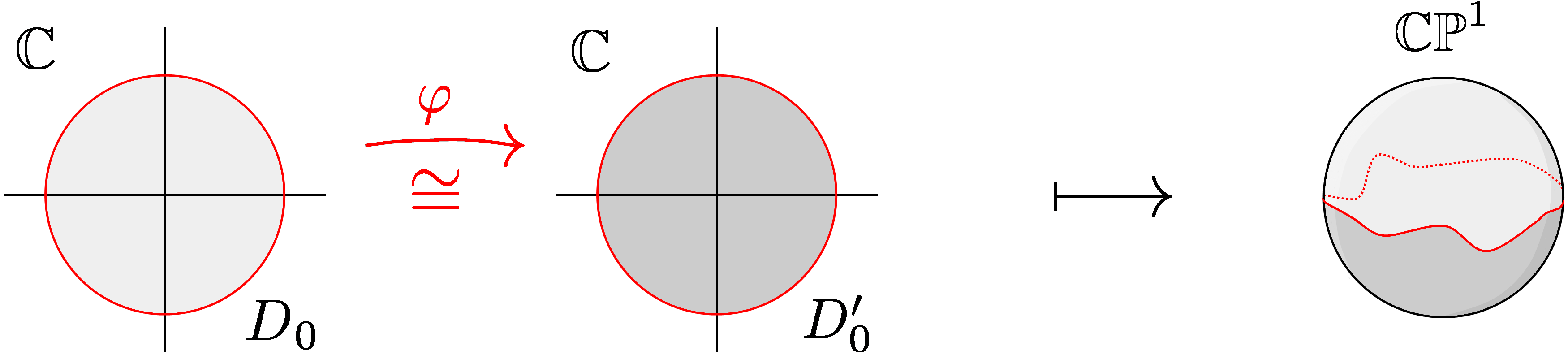}
	\end{center}\caption{Glueing two standard discs along their boundaries results in the Riemann sphere with a smoothly embedded curve.}\label{fig:glueing discs}
\end{figure}

In order to get an extended \textsc{cft}, both the source and target categories in Segal's definition of a \textsc{cft} are replaced by appropriate bicategories. An extended \textsc{cft} is then simply a symmetric monoidal functor between these bicategories. We first discuss the source bicategory.

\subsection{The source bicategory: conformal surfaces with cusps}\label{s:source cat}

Geometrically, an \textit{extended cobordism} is a cobordism, say $d$-dimensional, whose boundary comes in two pieces where each piece is viewed as $(d-1)$-dimensional cobordisms, and so on. In the case of \textsc{cft}s one is interested in~$d=2$, resulting in three tiers: zero-manifolds, one-manifolds, and two-manifolds.

The source category is therefore not a category but rather a bicategory, see \cite{Ben67} for background and definitions.
Before describing this bicategory in more detail, let us give the geometrical picture. Starting in dimension zero, we first have zero-dimensional, oriented manifolds: these are disjoint unions of points, each of which is labelled $+$ or~$-$ indicating the orientation. Moving up one dimension, we have cobordisms between the zero-dimensional manifolds.

\begin{figure}[h]
	\begin{center}
	\begin{subfigure}[b]{0.45\textwidth}
		\centering\includegraphics[scale=.10]{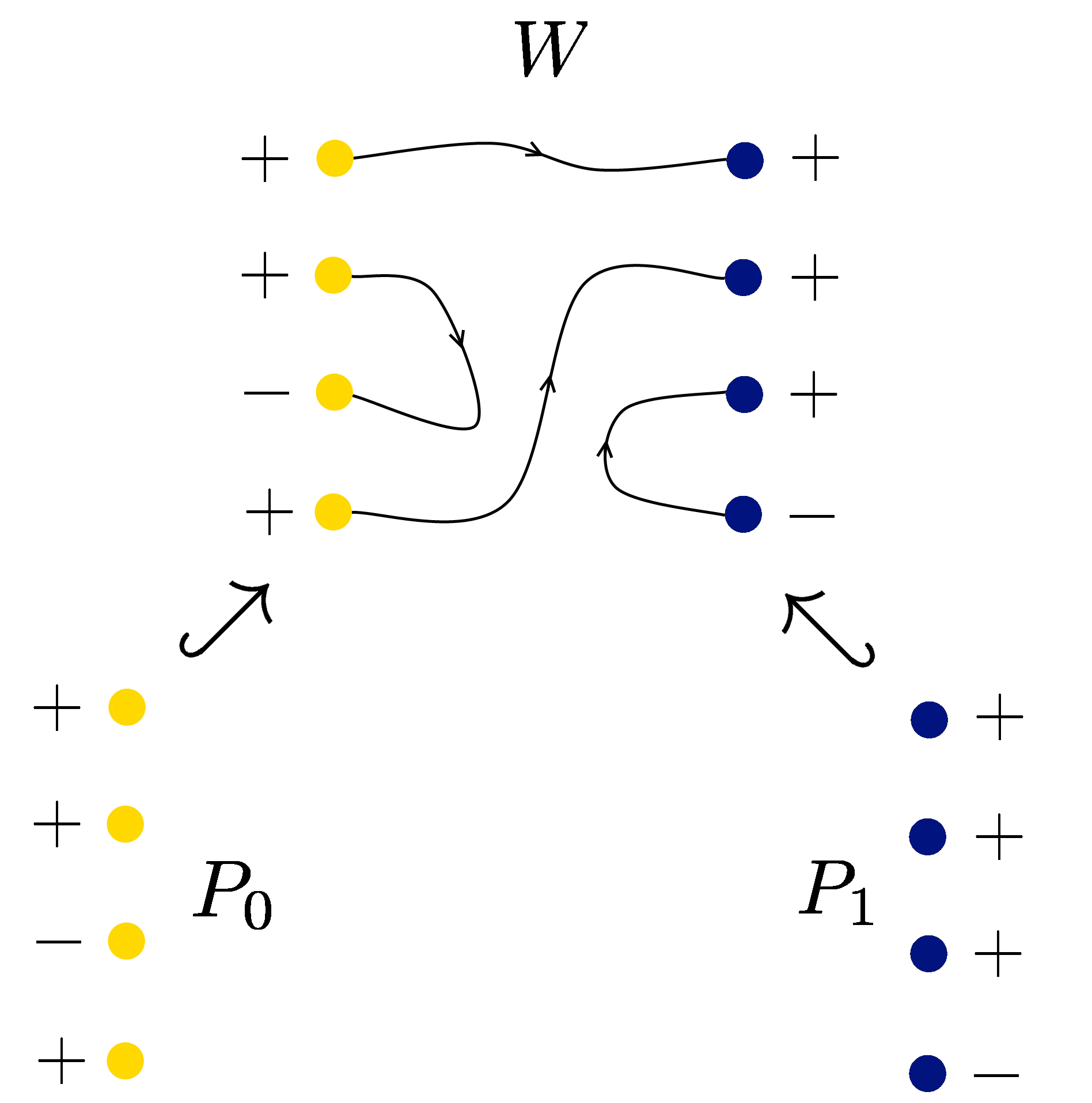}
		\caption{An example of a one-dimensional cobordism $W\colon P_0\longrightarrow P_1$.}\label{fig:1d cobordims position}
	\end{subfigure} \quad 
	\begin{subfigure}[b]{0.42\textwidth}
		\centering\includegraphics[scale=.10]{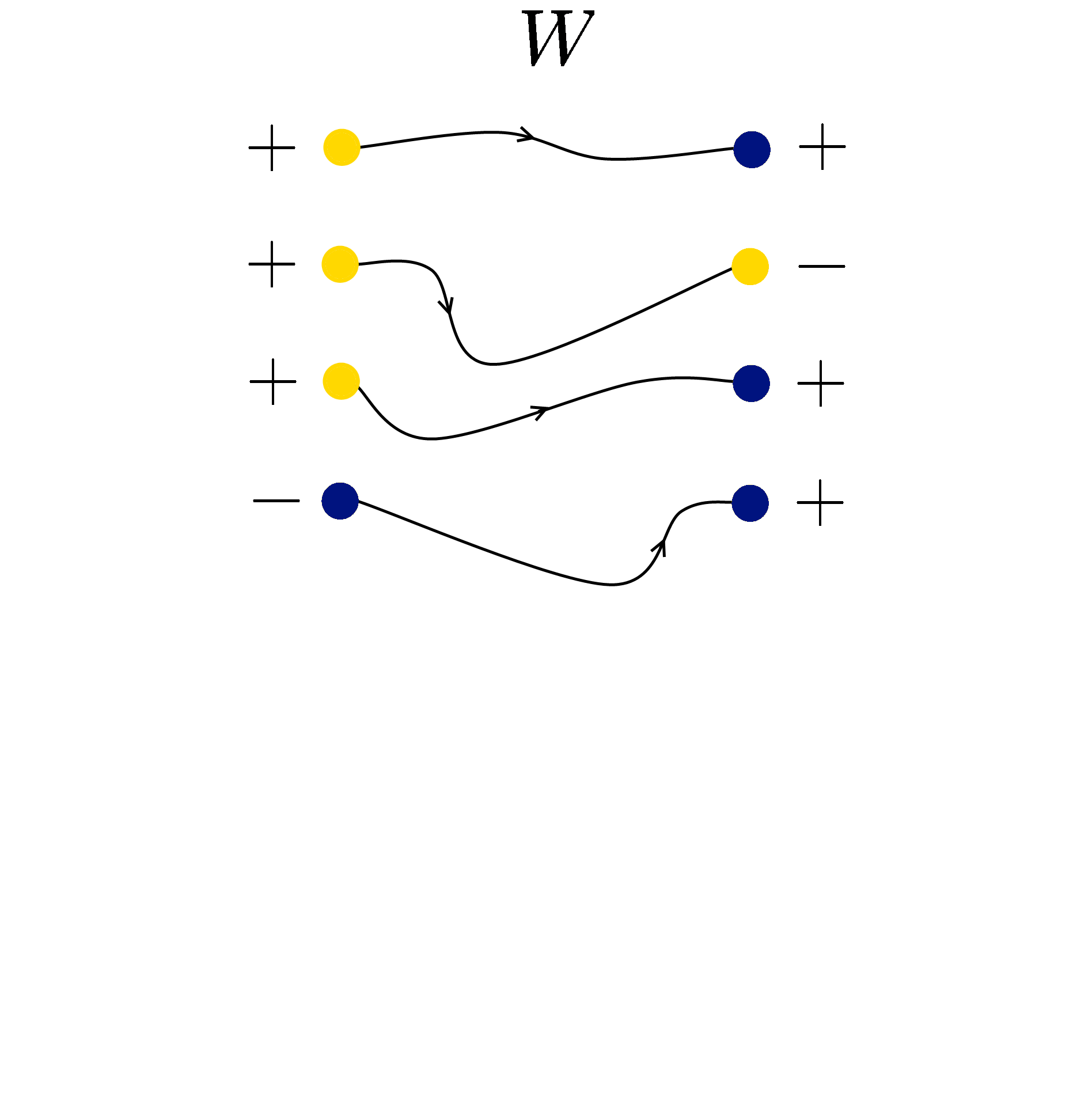}
		\caption{The same cobordism, depicted in a more convenient way.}\label{fig:1d cobordims color}
	\end{subfigure}
	\end{center}
	\caption{The figure on the left shows a one-dimensional cobordism. Notice that the orientation of the points is related to the orientation of the cobordisms connecting them, as required by the definition of a cobordism. The figure on the right shows the same cobordism displayed in way that will be more suitable for illustrations of two-dimensional cobordisms:
the incoming and outgoing zero-manifold are distinguished by their color, and the inclusion of the boundary manifolds is understood.}
\end{figure}

An example of such a one-dimensional cobordism shown in Figure~\ref{fig:1d cobordims position}, all `incoming' zero-manifolds are on the left, and all `outgoing' on the right. To facilitate drawing the more complicated examples below it is convenient to employ a different convention, and use colors to represent whether a zero-manifold is incoming or outgoing, respectively. With this convention, the example from Figure~\ref{fig:1d cobordims position} can also be represented as in Figure~\ref{fig:1d cobordims color}.

Circles, which form the objects of the source category of a non-extended \textsc{cft}, fit in the formalism of extended \textsc{cft} as closed cobordisms between empty zero-manifolds. They can be obtained from intervals by glueing. Going up one more dimension, cobordisms between such closed cobordisms are the conformal cobordisms that we encountered before: these are Riemann surfaces with boundary. However, now we also have two-dimensional cobordisms with cusps such as the examples in Figure~\ref{fig:2d cobordims a}.

\begin{figure}[h]
	\begin{center}
	\includegraphics[scale=.12]{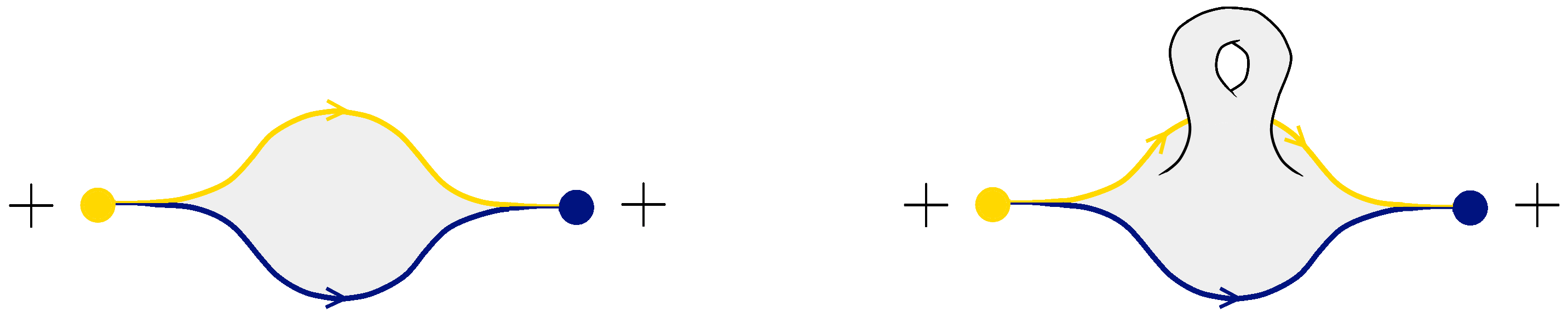}
	\end{center}\caption{Two examples of two-dimensional cobordisms with cusps. The cobordism on the right has a nontrivial topology. Like in Figure~\ref{fig:1d cobordims color}, the colors indicate which one-dimensional boundary is incoming and which is outgoing.}\label{fig:2d cobordims a}
\end{figure}

The source category is the bicategory of conformal surfaces with cusps, which is defined as follows. It has
\begin{itemize}
	\item objects: zero-dimensional, oriented manifolds $P$.
	\item 1-morphisms: one-dimensional cobordisms $P_0 \hookrightarrow W \hookleftarrow P_1$ with smooth structure, and collars $P_0 \times [\,0,\varepsilon) \longrightarrow W$ and $P_1 \times (-\varepsilon,0\,] \longrightarrow W$ parametrizing the ends.
	\item 2-morphisms: two-dimensional cobordisms $W_0 \hookrightarrow \Sigma \hookleftarrow W_1$ of cobordisms, with conformal structure in the interior $\Sigma \setminus (W_0 \cup W_1)$, and such that the diagrams
\[
\tikzmath{
\node(a) at (0,1.5) {$P_0 \times [\,0,\varepsilon)$};
\node(b) at (2,1.5) {$W_0$};
\node(c) at (0,0) {$W_1$};
\node(d) at (2,0) {$\Sigma$};
\draw[->] (a) -- (b);\draw[->] (a) -- (c);\draw[->] (b) -- (d);\draw[->] (c) -- (d);
} 
\qquad \text{and} \qquad 
\tikzmath{
\node(a) at (0,1.5) {$P_1 \times (-\varepsilon,0\,]$};
\node(b) at (2,1.5) {$W_0$};
\node(c) at (0,0) {$W_1$};
\node(d) at (2,0) {$\Sigma$};
\draw[->] (a) -- (b);\draw[->] (a) -- (c);\draw[->] (b) -- (d);\draw[->] (c) -- (d);
} 
\]
commute, after maybe shrinking $\varepsilon$.
Furthermore, $\Sigma$ should be locally isomorphic to one of the local models specified in Section~\ref{s:local models} below.
\end{itemize}
The two diagrams above say that the parametrizations of the one-dimensional cobordisms bounding the surface have to agree on neighbourhoods of their ends.
In particular, this forces the two-dimensional cobordism $\Sigma$ to be in fact one-dimensional near the zero-manifolds $P_0$ and $P_1$.
Figure~\ref{fig:2d cobordims b} shows an example of a 2-morphism in the category of conformal surfaces with cusps. Taking disjoint unions endows the category of conformal surfaces with cusps with a symmetric monoidal structure.

\begin{figure}[h]
	\begin{center}
	\includegraphics[scale=.14]{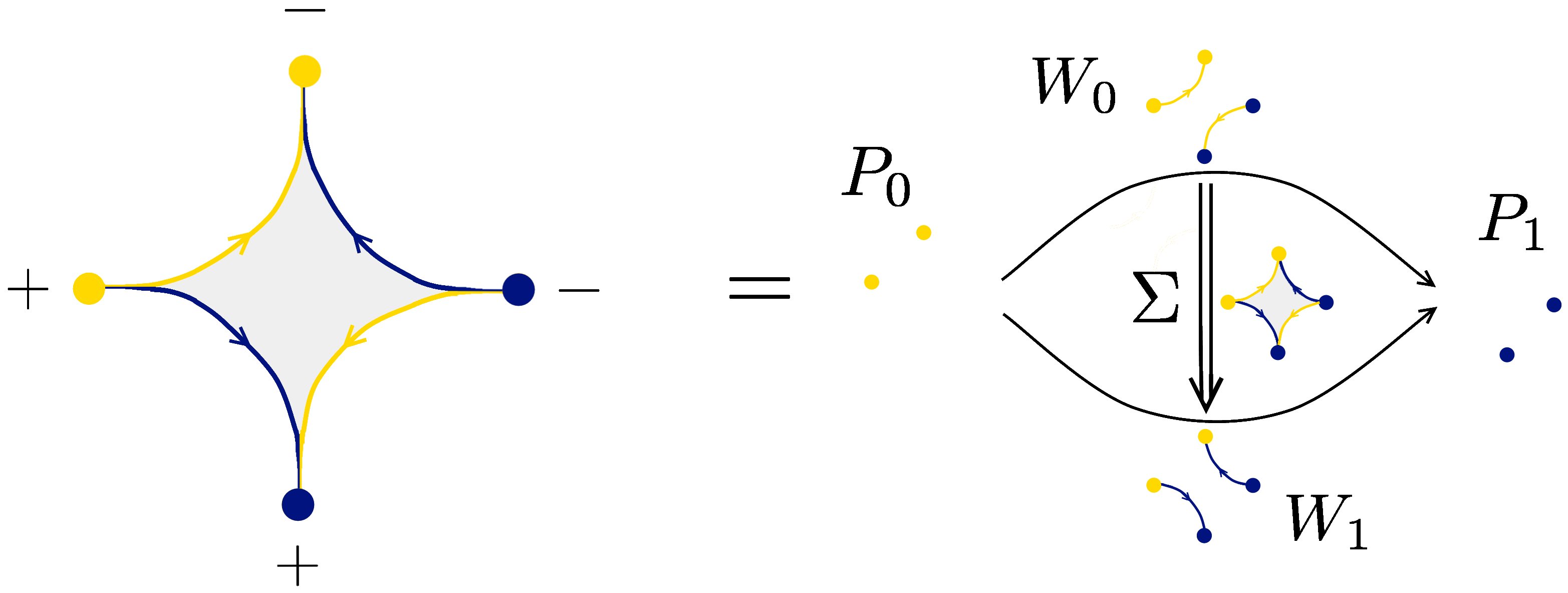}
	\end{center}\caption{An example of a conformal surface with cusps, and the corresponding 2-morphism $\Sigma\colon W_0 \longrightarrow W_1$.}\label{fig:2d cobordims b}
\end{figure}

\subsubsection{Local models}\label{s:local models}

The various manifolds comprising the bicategory of conformal surfaces admit local models.
Being a local model means that any point of such a manifold has a neighbourhood that looks the same as some open subset of the corresponding local model. For example, the local model of an object~$P$ is simply a point with a choice of orientation, and for a 1-morphism~$S$ it is the unit interval~$[0,1]$. Unlike for the case of 0-morphisms (objects) and 1-morphisms, where one can show that, locally, they must look like one of the local models, the case of 2-morphisms is different. For 2-morphisms, giving the list of allowed local models is part of the definition of what things we allow as 2-morphisms.

We can describe the local models for our 2-morphisms $\Sigma$ as follows. Let $f,g \in C^\infty\big([0,1],\mathbb{R}\big)$ be smooth functions on the unit interval, such that  $f\leq g$, and such that $f$ and $g$ are equal on neighbourhoods of $0$ and~$1$. Then a local model of $\Sigma$ is
\begin{equation}\label{eqn:local model for Sigma}
	\Sigma = \{ \, x + i \,y \mid f(x) \leq y \leq g(x) \, \} \ .
\end{equation}
In particular, since we require $f$ and $g$ to agree near the ends, the tips of $\Sigma$ are really one-dimensional cusps as depicted in Figure~\ref{fig:Sigma}. There are many different local models for the 2-morphisms, with different choices for $f$ and $g$ yielding varying degrees of `sharpness' for the cusps.

\begin{figure}[h]
	\begin{center}
	\includegraphics[scale=.12]{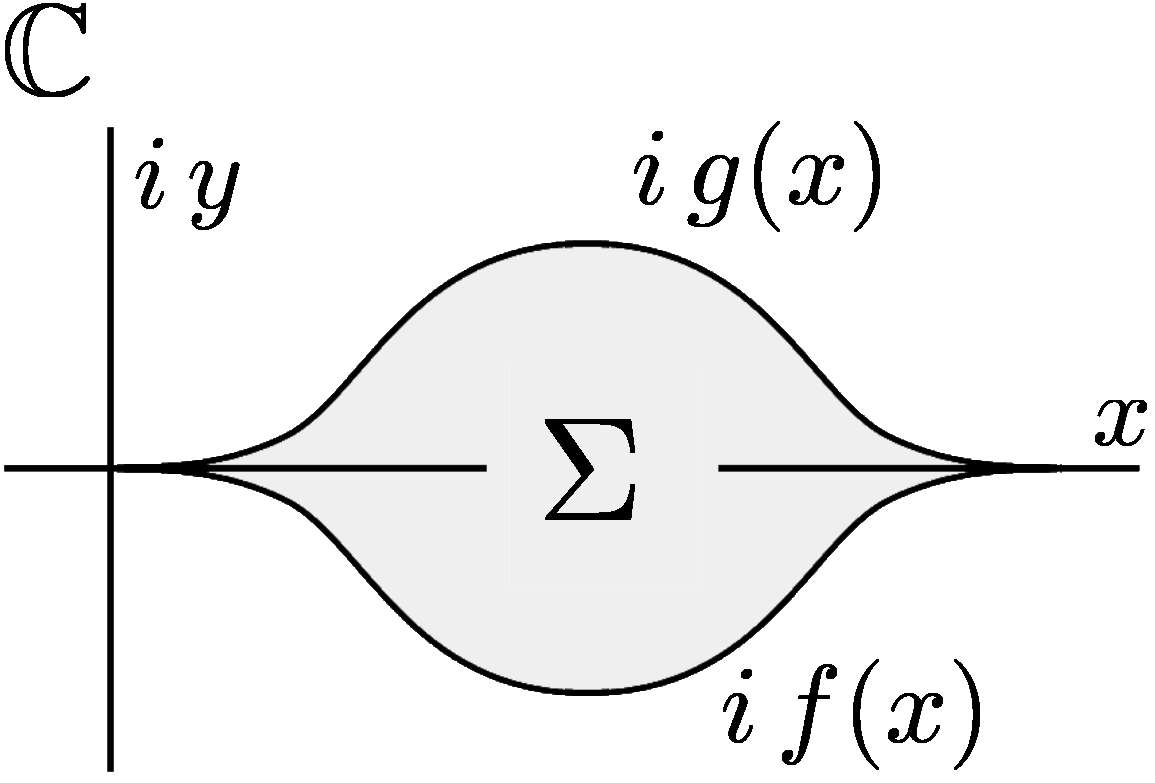}
	\end{center}\caption{A local model as described by equation~\eqref{eqn:local model for Sigma}. Different choices for $f,g\in C^\infty\big([0,1],\mathbb{R}\big)$ agreeing near the endpoints give rise to different degrees of sharpness for the cusps.}\label{fig:Sigma}
\end{figure}

Since the 1-morphisms correspond to collared one-manifolds, they can be composed by glueing. To see that the glued surfaces with cusps are again of the prescribed form, notice that our problem is local. Thus we may assume without loss of generality that the surfaces we want to glue are given by the local models. It is clear from Figure~\ref{fig:horizontal composition} that the horizontal composition is again of the form~\eqref{eqn:local model for Sigma}.

\begin{figure}[h]
	\begin{center}
	\includegraphics[scale=.10]{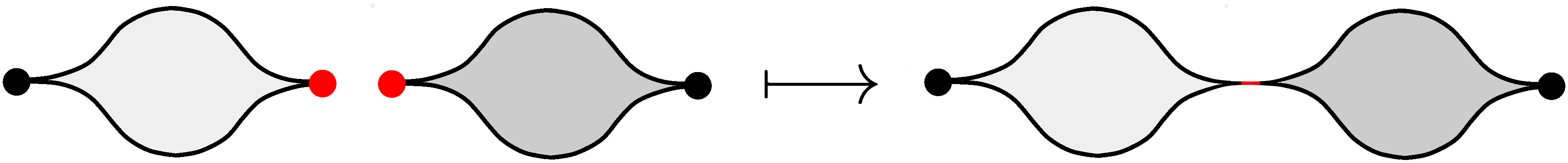}
	\end{center}\caption{Horizontal composition of two-dimensional local models.}\label{fig:horizontal composition}
\end{figure}

\begin{figure}[h]
	\begin{center}
	\includegraphics[scale=.10]{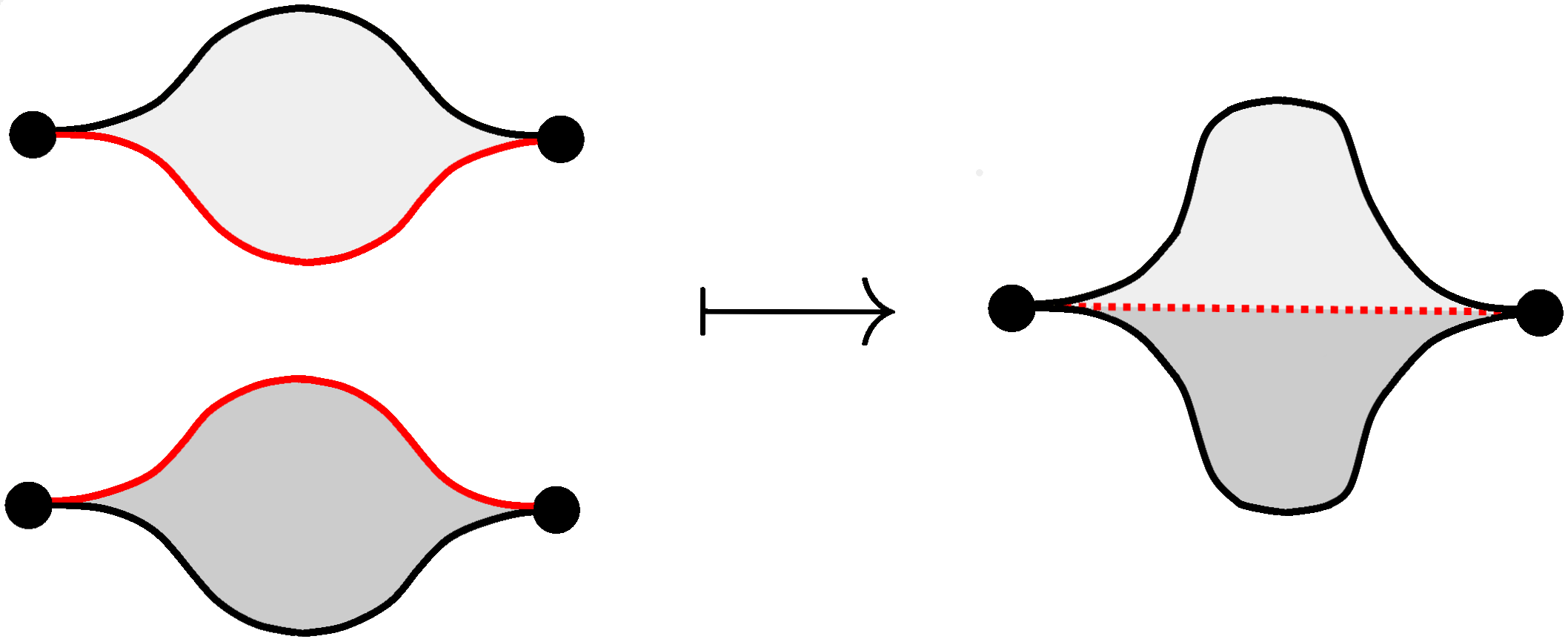}
	\end{center}\caption{Vertical composition of two local models for conformal surfaces with cusps. It requires some work to show that the result is a local model too.}\label{fig:vertical composition}
\end{figure}

The vertical composition of two 2-morphisms looks as shown in Figure~\ref{fig:vertical composition}.
We can use Theorem~\ref{thm:conf welding}, underlying conformal welding, to get a complex structure on the interior of the glued surface.
But a priori, it is not clear that the result is again one of our local models.
To show that, we will use Lemma~\ref{lem:conf welding}. First we get rid of the corners by embedding the two surfaces that we want to glue into discs with a smooth boundary, as shown in Figure~\ref{fig:vertical composition step 1}. 
Next, we extend the diffeomorphism of the boundaries of the surfaces that we want to identify to a diffeomorphism between the boundary circles, and glue. The result is depicted in Figure~\ref{fig:vertical composition step 2}. By Lemma~\ref{lem:conf welding}, everything is smoothly embedded in~$\mathbb{CP}^1$. This shows that the glued surface is again one of our allowed local models, and so it is again a 2-morphism.

\begin{figure}[h]
	\begin{center}
	\includegraphics[scale=.10]{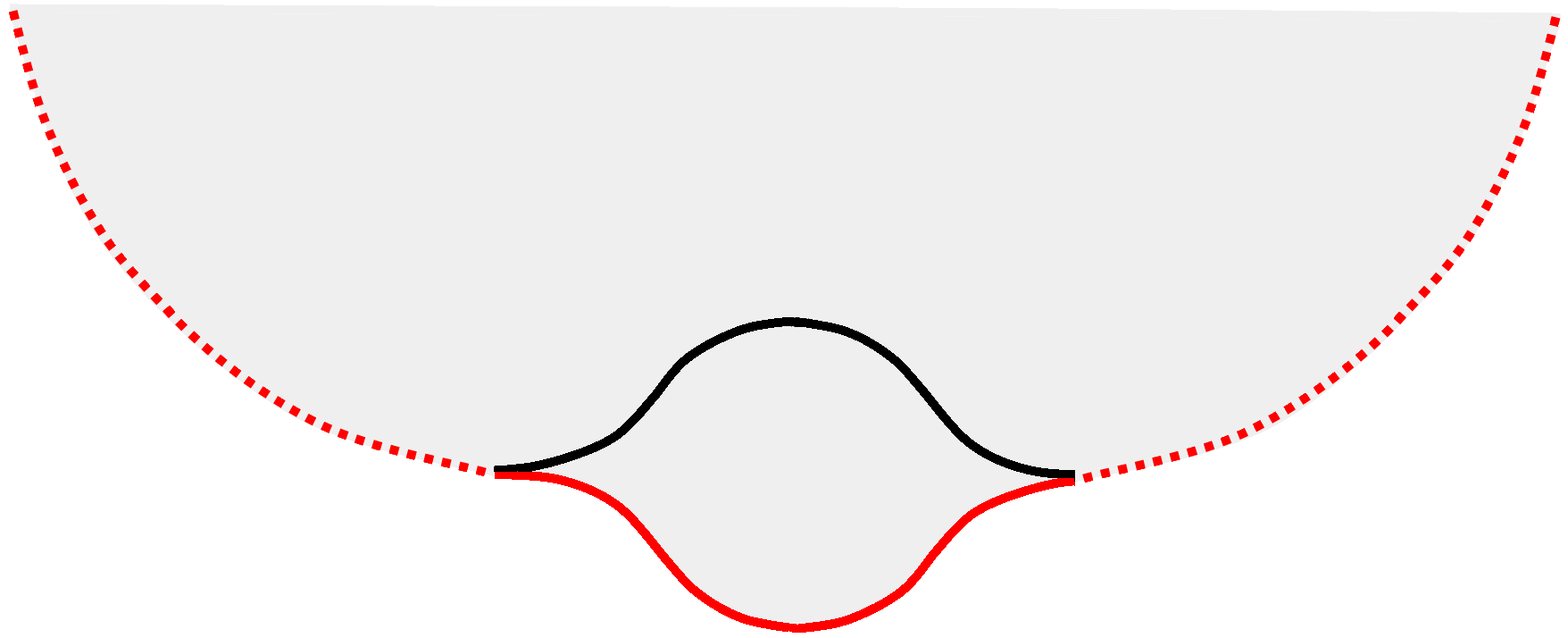}
	\end{center}\caption{A local model $\Sigma$ can be embedded into some disc whith smooth boundary in the complex plane. }\label{fig:vertical composition step 1}
\end{figure}

\begin{figure}[h]
	\begin{center}
	\includegraphics[scale=.10]{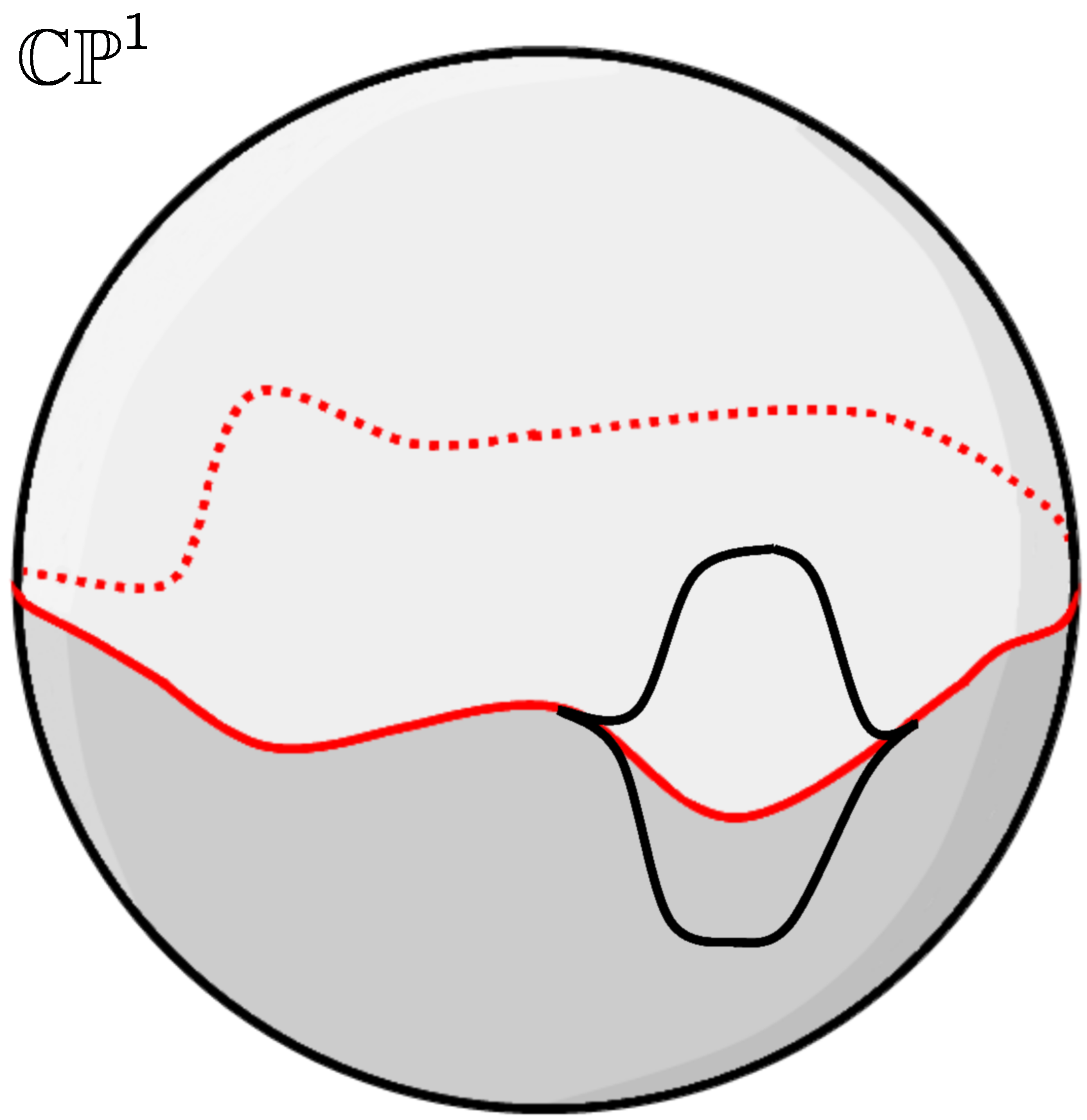}
	\end{center}\caption{Glueing two local models $\Sigma_1$ and $\Sigma_2$ that are embedded into discs  like in Figure~\ref{fig:vertical composition step 1} yields a copy of the vertical composition of $\Sigma_1$ and $\Sigma_2$ inside the Riemann sphere.}\label{fig:vertical composition step 2}
\end{figure}

\subsection{The target bicategory: von Neumann algebras}\label{s:target cat}

Since an extended \textsc{cft} should encompass the notion of \textsc{cft}, it should certainly map (a union of) circles to some Hilbert space, and a cobordism connecting such circles to a linear map between Hilbert spaces, as before.
We have to decide what we want to assign to a point: these should be some kind of algebras.
If we want to stay in a Hilbert space setting, then there are not many options for the kind of algebras to consider.
It turns out that the appropriate choice is given by von Neumann algebras.
A one-dimensional cobordism is then mapped to a bimodule between von Neumann algebras, and a surface such as~\eqref{eqn:local model for Sigma} corresponds to a linear map between bimodules. In short, the target bicategory is defined as follows:
\begin{itemize}
	\item objects: von Neumann algebras;
	\item 1-morphisms: bimodules (that is, Hilbert spaces with a left action of the first von Neumann algebra, and a commuting right action of the second von Neumann algebra);
	\item 2-morphisms: bounded linear maps that are compatible with the bimodule structure.
\end{itemize}

Before we give a definition of these notions, we recollect some facts from the theory of operator algebras.
Given a Hilbert space $H$, denote the algebra of bounded operators on~$H$ by $B(H)$. Recall that an operator $a \in B(H)$ is \textit{trace class} if it is compact and the trace-norm $\|a\|_1 \coloneqq \sum_k \sqrt{\mu_k}$ is finite; here the $\mu_k$ are the eigenvalues of the positive operator~$a^* a$. This ensures that the trace of $a$ is well defined. Write $B_1(H)$ for the trace-class operators in~$B(H)$. The pairing
\begin{align*}
	B(H) \times B_1(H) \longrightarrow \mathbb{C} \ , \quad (a,b) \longmapsto \text{tr}\, (a\,b)
\end{align*}
induces a topology on $B(H)$ that is called the \textit{ultraweak} topology. Thus, a (generalized) sequence $\{a_i\}$ in $B(H)$ converges ultraweakly to $a \in B(H)$ if and only if for all $b \in B_1(H)$ we have that $\text{tr}(a_ib) \longrightarrow \text{tr}(ab)$ in $\mathbb{C}$.

\begin{definition}
A \textit{von Neumann algebra} is a topological $*$-algebra\footnote{Notice that the multiplication
map $B(H) \times B(H) \longrightarrow B(H)$ is not continuous, so the term ``topological $*$-algebra'' should be taken with a grain of salt.} $A$ over $\mathbb{C}$ that can be embedded in some $B(H)$ as a ultraweakly closed $*$-subalgebra.
\end{definition}

By the von Neumann bicommutant theorem, $A$ is ultraweakly closed if and only if it is its own bicommutant.

\begin{definitions}
A \textit{module} over a von Neumann algebra~$A$ is a Hilbert space~$H$ together with a continuous $*$-homomorphism $A \longrightarrow B(H)$.

Similarly, if $A$ and $B$ are von Neumann algebras, an \textit{$A$-$B$-bimodule} is a Hilbert space~$H$ equipped with two continuous $*$-homomorphisms $A \longrightarrow B(H)$ and $B^\text{op} \longrightarrow B(H)$ whose images commute. We write ${}_A M_B$ to indicate that $M$ is an $A$-$B$-bimodule.
\end{definitions}

\noindent Here we have written $B^\text{op}$ for the von Neumann algebra obtained from $B$ by reversing the order in the multiplication: if $m \colon B \otimes B \longrightarrow B$ is the original multiplication on $B$ then the opposite multiplication is given by $m^\text{op}(a,b)=b\,a$.

It is more work to define the composition of bimodules in the bicategory of von Neumann algebras than it is to do so in the bicategory of rings.
Recall the way in which rings and bimodules form a bicategory.
Given two rings $R$ and $S$, let $\text{Hom}(R,S)$ be the category of $R$-$S$-bimodules.
The morphisms in $\text{Hom}(R,S)$ are then the 2-morphisms of our bicategory.
If $R$, $S$, $T$ are rings, the
horizontal composition of two bimodules ${}_R M_S$ and ${}_S N_T$ is given by the tensor product:
	\[  {}_A M_B \circ {}_B N_C \coloneqq {}_A M \underset{B}{\otimes} N_C \ . \]
Here the tensor product is taken over $B$, so that $(m \cdot b) \otimes n = m \otimes (b \cdot n)$. The $A$-$A$-bimodule $A$ is then the unit object for 
this kind of composition.

If we want to do something similar with von Neumann algebras, the first obstacle is the definition of the unit object: a von Neumann algebra is not a Hilbert space, so it cannot serve as a bimodule over itself. However, there is a canonical way to turn a von Neumann algebra $A$ into a Hilbert space, called $L^2 A$.

\subsubsection{The $L^2$-space of a von Neumann algebra}

The definition of $L^2 A$ requires more prerequisites from the theory of operator algebras. We outline its construction.

For $A$ a von Neumann algebra, let
\begin{align*}
	L^1 A & \coloneqq \{ \ \varphi \colon A\longrightarrow \mathbb{C} \mid \text{continuous} \ \} \ , \\
	L^1_+ A & \coloneqq \{ \ \varphi \in L^1 A \mid \varphi(a^* a) \geq 0 \ \text{for all} \ a \in A \ \} \ .
\end{align*}
Elements of $L^1_+ A$ are called \textit{states}\footnote{Often, one also puts the condition that $\varphi(1)=1$.} on $A$. The Gelfand-Naimark-Segal construction says that for each state $\varphi \in L^1_+ A$ there exists a cyclic representation $\pi_\varphi$ of $A$ on some Hilbert space~$H_\varphi$ with cyclic vector $\Omega_\varphi$. Thus, the image $\pi_\varphi(A)\,\Omega_\varphi$ of the action of $A$ on~$\Omega_\varphi$ is dense in~$H_\varphi$.

If the state is faithful (i.e. if $\varphi(a^* a)>0$ for $a\not=0$), then the antilinear operator $\pi_\varphi(a)\, \Omega_\varphi \longmapsto \pi_\varphi(a)^* \,\Omega_\varphi$ defined on $\pi_\varphi(A)\,\Omega_\varphi$ can be extended to an operator $S_\varphi$ on  the closure $H_\varphi$ of $\pi_\varphi(A)\,\Omega_\varphi$. From this operator we can further construct the positive operator $\Delta_\varphi \coloneqq |S_\varphi|^2 = S_\varphi^* S_\varphi$. 
Since $\Delta_\varphi$ is a positive operator, $\Delta_\varphi^{it}= \exp(i\,t\log\Delta_\varphi)$ is well defined for all $t\in \mathbb{R}$.

By a theorem that is due to Tomita and Takesaki, for each $a\in A$, the assignment $t \longmapsto \Delta_\varphi^{-it} \, a \, \Delta_\varphi^{it}$ defines a one-parameter family of elements in~$A$. This is called the \textit{modular group} of $A$ associated with~$\varphi$.

Next, consider the algebra $\text{Mat}_2(A)$ of $2\times 2$ matrices with coefficients in $A$ and let $\varphi\oplus\psi \in L^1_+\big(\text{Mat}_2(A)\big)$. Via the above construction, $\varphi\oplus\psi$ yields a modular group in $\text{Mat}_2(A)$. Applying this modular group to the element
	\[ \left(\begin{array}{cc} 0 & 1 \\ 0 & 0 \end{array}\right) \in \text{Mat}_2(A) \]
we get elements that are of the form
	\[  \Delta_{\varphi\oplus\psi}^{-it} \left(\begin{array}{cc} 0 & 1 \\ 0 & 0 \end{array}\right) \Delta_{\varphi\oplus\psi}^{it} = \left(\begin{array}{cc} 0 & \cdots \\ 0 & 0 \end{array}\right) \ \in\mathrm{Mat}_2(A)\,. \]
The \textit{non-commutative Radon-Nikodym derivative} $[D\varphi:D\psi]_t$ is then defined via
	\[ \left(\begin{array}{cc} 0 & [D\varphi : D\psi]_t \\ 0 & 0 \end{array}\right) \coloneqq \Delta_{\varphi\oplus\psi}^{-it} \left(\begin{array}{cc} 0 & 1 \\ 0 & 0 \end{array}\right) \Delta_{\varphi\oplus\psi}^{it} \ . \]

Now consider the free vector space on symbols $\sqrt{\varphi}$ with $\varphi \in L^1_+ A$. The above construction allows us to define a (semi-definite) inner product on this vector space via the formula
	\[ \langle \, \sqrt{\varphi} , \sqrt{\psi} \, \rangle \coloneqq \anacont_{t\longrightarrow i/2} \, \varphi\big([D\varphi : D\psi]_t \big) \ . \]

After all these preliminaries, we are finally in a position to define $L^2 A$: it is the Hilbert space obtained as the completion of the above free vector space with respect to this inner product.
For each von Neumann algebra $A$, the Hilbert space $L^2 A$ is an $A$-$A$-bimodule, ${}_A L^2 A_A$, and this is the unit morphism in the bicategory of von Neumann algebras.

The Hilbert space $L^2 A$ is also equipped with a \emph{positive cone} $L_+^2A\subset L^2A$, given by
\[
L_+^2A:=\{\sqrt{\varphi}\,|\,\varphi\in L_+^1A\}\,,
\]
and an antilinear involution $J:L^2 A\to L^2 A$, called the \emph{modular conjugation}.
The modular conjugation is given by
$J(\sum_i c_i\sqrt{\varphi_i}):=\sum_i \bar c_i\sqrt{\varphi_i}$.

\subsubsection{Connes fusion}

The second difficulty towards defining the bicategory of von Neumann algebras is that the ordinary tensor product does not work:
it would have ${}_A A_A$ as its unit, not ${}_AL^2(A)_A$. The appropriate tensor product of von Neumann bimodules, known as \textit{Connes fusion} and denoted by $\boxtimes$, is tailor-made so that ${}_A L^2 A_A$ is a unit for that operation. We have
\begin{equation}\label{eqn:Connes fusion}
	M \underset{A}{\boxtimes} N \coloneqq \text{completion of } M \underset{A}{\otimes} \text{Hom}_A(L^2 A, N)\,.
\end{equation}
This is actually forced on us if we want $L^2 A$ to be the unit.
If we accept for a moment that $L^2 A$ is a unit, then
given an $A$-linear map $\varphi \colon L^2 A \longrightarrow N$ and an element $m \in M$, there is an easy way of producing an element of $M \boxtimes_A N$:
take the image of $m \in M \cong M \boxtimes_A L^2 A$ under the map $1\boxtimes \varphi:M \boxtimes_A L^2 A\to M \boxtimes_A L^2 N$.

The completion is taken with respect to an inner product on the right-hand side of~\eqref{eqn:Connes fusion}.
Let us work backwards to figure out the correct formula for the inner product.
The inner product of two elements $n\otimes \phi$ and $m\otimes \psi$ of $M \boxtimes_A N$ can be described as the composition
\[
\tikzmath{
\node(a) at (.2,0) {$\mathbb C $};
\node(b) at (2.5,0) {$M\cong M\boxtimes_AL^2A$};
\node(c) at (6,0) {$M\boxtimes_A N$};
\node(d) at (9.5,0) {$M\boxtimes_AL^2A\cong M$};
\node(e) at (11.8,0) {$\mathbb C$};
\draw[->] (a) --node[above]{$\scriptstyle n$} (b);\draw[->] (b) --node[above]{$1\boxtimes\phi$} (c);\draw[->] (c) --node[above]{$1\boxtimes\psi^*$} (d);\draw[->] (d) --node[above]{$\scriptstyle m^*$} (e);
\draw[dashed, ->] ($(b.north)+(.3,0)$) to[out=65, in=115, looseness=.8]node[above]{$\scriptstyle 1\,\boxtimes\,(\psi^*\circ\,\phi)$} ($(d.north)+(-.3,0)$);
} 
\]
where $\phi,\psi\in\text{Hom}_A(L^2 A,N)$, $\psi^*$ is the adjoint of $\psi$, and we view $m,n\in M$ as maps $\mathbb{C}\longrightarrow M$. Notice that the map $\psi^*\circ\phi\colon {}_A L^2 A\longrightarrow {}_A L^2 A$ commutes with the left action of $A$ on $L^2 A$. Now, one of the properties of $L^2 A$ is that endomorphisms of $L^2 A$ which are equivariant for the left $A$-action are given by right multiplication $\rho_a$ by some $a \in A$. Therefore we have that $\psi^*\circ\phi = \rho_a$ for some $a=a_{\psi^*\circ\phi} \in A$. The inner product $\langle m \otimes \phi, n\otimes\psi\rangle$ on $M \otimes \text{Hom}_A(L^2 A,M)$ is now given by the composition
\[
\tikzmath{
\node(a) at (0,0) {$\mathbb C $};
\node(b) at (2,0) {$M$};
\node(d) at (5.5,0) {$M$};
\node(e) at (7.5,0) {$\mathbb C$};
\draw[->] (a) --node[above]{$\scriptstyle n$} (b);\draw[->] (b) --node[above]{$\rho_a$} (d);
\draw[->] (d) --node[above]{$\scriptstyle m^*$} (e);
} \,.
\]
A more symmetric way to write the Connes fusion product is
	\[ M \underset{A}{\boxtimes} N \cong \text{Hom}_A(L^2 A, M) \underset{A}{\otimes} L^2 A \underset{A}{\otimes} \text{Hom}_A(L^2 A, N)\,. \]
The evaluation map $\text{Hom}_A(L^2 A, M) \otimes_A L^2 A \longrightarrow M$ relates this description to the previous asymmetric definition: after completion those two descriptions become isomorphic to each other.

\section{Conformal nets and Frobenius algebra objects}\label{s:chiral vs full}

Before we push on, let us pause a moment to sketch the big picture. Actually, there are a couple of different things called `\textsc{cft}'; in particular, physics distinguishes between \textsl{chiral} and \textsl{full} \textsc{cft}. It is quite common to use `\textsc{cft}' to refer to one of these things, but it may not always be clear from the context to which one. What we have been calling \textsc{cft} above are really \textit{full} \textsc{cft}s. We abbreviate `chiral \textsc{cft}' to `$\chi$\textsc{cft}' so that we can continue to use `\textsc{cft}' without further specification exclusively for `full \textsc{cft}'.

Chiral \textsc{cft}s can be seen as an intermediate step towards full \textsc{cft}. The distinction between chiral and full \textsc{cft} has its origin in physics. Now comes the mathematics to make things more complicated: there exist different mathematical formalisms to talk about $\chi$\textsc{cft}, and to talk about full \textsc{cft}. We have already discussed Segal's formalism for full \textsc{cft}. Chiral \textsc{cft} can be described in the formalism of Segal as well, and there are also approaches using \textsl{vertex operator algebras} or \textsl{conformal nets}. Shortly we will present $\chi$\textsc{cft} in Segal's formalism, and in Sections~\ref{s:loop group nets} and~\ref{s:conformal net} we will also look at the approach via conformal nets.

Recall that the \textit{loop group} of a compact Lie group~$G$ is defined as the group of maps from the unit circle into $G$:
\begin{equation}\label{eqn:loop group}
	L\,G \coloneqq \text{Map}_{C^\infty}(S^1,G) \ .
\end{equation}
Loop groups are relevant for us because they lead to vertex operator algebras or conformal nets, and so they provide examples of $\chi$\textsc{cft}s. In order to construct a full \textsc{cft} out of a $\chi$\textsc{cft}, one needs additional data: a \textsl{Frobenius algebra object} in the monoidal category associated to the $\chi$\textsc{cft}. In Section~\ref{s:Frobenius} we will define Frobenius algebra objects, and~in Section \ref{s:FRS} we will illustrate how such an object helps to construct a full \textsc{cft} out of a $\chi$\textsc{cft}. 

To summarize, the situation can be represented as follows:
\[
\tikzmath{
\node(a) at (0,0) {loop groups};
\node(b) at (5,0) {chiral \textsc{cft}};
\node(c) at (10,0) {full \textsc{cft}};
\draw[->] (a) --node[above, scale=.85]{\parbox{2cm}{provide examples of}} (b);
\draw[->] (b) --node[above, scale=.85]{\parbox{2.5cm}{+ Frobenius\\ algebra object}} (c);
} 
\]
We will show that with the same input, a chiral \textsc{cft} and a Frobenius algebra object, one can actually do better and construct an \textsl{extended} \textsc{cft}:
\begin{equation}\label{eqn:diagram with extended CFT}
\tikzmath{
\node(a) at (6,2.5) {extended \textsc{cft}};
\node(b) at (0,0) {chiral \textsc{cft}};
\node(c) at (6,0) {full \textsc{cft}};
\draw[->] (a) --node[right, scale=.85]{forget} (c);
\draw[->] (b) --node[above, scale=.85]{\parbox{2.5cm}{+ Frobenius\\ algebra object}} (c);
\draw[->, dotted] (b) to[bend left=20] (a);
} 
\end{equation}

In Section~\ref{s:construction} we will (partially) construct examples of extended \textsc{cft}s. We should put this task into perspective: already for (non-extended) Segal \textsc{cft}s, the interesting examples --- most notably those coming from loop groups --- have \textsl{not} completely been constructed. In the spirit of the cobordism hypothesis \cite{Lurie(On-the-classification-of-topological-field-theories)}, one could even hope that it is easier to construct extended \textsc{cft}s than full \textsc{cft}s.

\subsection{Chiral conformal field theory}\label{s:chiral CFT}

In this section, we use Segal's formalism to describe (non-extended) $\chi$\textsc{cft}. Recall from Section~\ref{s:Segal CFT} that a non-extended \textsc{cft} assigns Hilbert spaces to closed one-dimensional manifolds, and maps between Hilbert spaces to conformal cobordisms. Chiral \textsc{cft}s have the same source category as full \textsc{cft}s, but there is an intermediate layer on the side of the target. 

A $\chi$\textsc{cft} first assigns to every closed one-dimensional manifold a $\mathbb C$-linear category~$\C$.
To each object $\lambda \in \C$, it further assigns a Hilbert space $H_\lambda$. Likewise, a cobordism (always with complex structure) is mapped to a functor $f \colon \C_\text{in} \longrightarrow \C_\text{out}$, and for each $\lambda \in \C_\text{in}$ we further get a map $H_\lambda \longrightarrow H_{f(\lambda)}$. 
This is only part of the data of a $\chi$\textsc{cft}, but we can already try to list certain axioms.
Most importantly, the map $H_\lambda \longrightarrow H_{f(\lambda)}$ must depend on the complex structure of the cobordism in a holomorphic fashion.
Here is what this means, roughly.
If we fix two one-manifolds $W_\text{in}$ and $W_\text{out}$, then the (infinite-dimensional) moduli space of Riemann surfaces with these boundaries has its own complex structure: the functions from this moduli space, mapping points to operators, are required to be holomorphic.\footnote{More precisely, these functions should be holomorphic in the interior of the moduli space, but only continuous on its boundary.}

Note however that the Hilbert space $H_{f(\lambda)}$ depends on the choice of cobordism.
So, as a prerequisite for the above condition to make sense, $f \colon \C_\text{in} \longrightarrow \C_\text{out}$ should depend holomorphically on the choice of cobordism.
That is, the Hilbert spaces $H_{f(\lambda)}$ should form a holomorphic bundle over the moduli space of complex cobordisms.
On top of that, there is also a unitary projectively flat connection on that same bundle:
a path between two cobordisms (in the moduli space of cobordisms between $W_\text{in}$ and $W_\text{out}$) then induces a natural isomorphism between the corresponding functors,
and if two paths are homotopic, then the two natural isomorphisms are equal up to a phase.
So, overall, a $\chi$\textsc{cft} is a rather involved kind of structure.

To get a feeling about what these categories associated to one-manifolds are, we look at the examples of $\chi$\textsc{cft}s which are provided by loop groups. Let $G$ be a Lie group. To a one-manifold $W$ in the source category we assign the category of representations of $\text{Map}_{C^\infty}(W,G)$ (compare with~\eqref{eqn:loop group}). 
Moreover, each object $\lambda$ in that category has an underlying Hilbert space: those are the $H_\lambda$.

We should point out that $\chi$\textsc{cft} in the above formalism are difficult to construct, and despite a lot of hard work (\cite{TUY89, Zhu1996, Huang1997, Posthuma-PhD}, and of course \cite{Segal1988}) the $\chi$\textsc{cft}s corresponding to loop groups have been constructed to a great extent, but not completely.

We finish our short discussion of chiral conformal field theories by emphasizing the most important structure that a such a theory encodes: a braided monoidal category.
Let $\C$ be the category that the $\chi$\textsc{cft} assigns to the standard circle.
Then the pair of pants equips $\C$ with a monoidal structure\footnote{This is not completely obvious since, a priori, the pair of pants needs a conformal structure before 
we know which functor $\C\times\C\to \C$ it induces. However, because the pair of pants has genus zero, there is nevertheless a way of getting a canonical functor $\C\times\C\to \C$.}, and the diffeomorphism that switches the two pant legs (followed by a path inside the space of complex structures on the pair of pants) further equips it with the structure of a braided monoidal category.

\subsection{Conformal nets}

Our construction, as indicated in~\eqref{eqn:diagram with extended CFT}, is not based on the above formalism. Rather, it uses the formalism of conformal nets. To get acquainted with conformal nets we will start by giving the data of a conformal net and look at an example. In Section~\ref{s:conformal net} we will give the complete abstract definition of conformal nets, including the axioms for the above data.

\begin{data}
A \textit{conformal net} $\A$ is a monoidal functor from the category\footnote{Notice that the source category is not quite monoidal: we cannot take disjoint unions of embeddings that are orientation preserving and embeddings that are orientation reversing.}
	\[ \begin{cases} \ \text{objects: compact, oriented, one-dimensional manifolds with boundary;} \\ \ \text{morphisms: embeddings that either preserve the orientation on all}
	\\ \ \qquad\qquad \text{connected components or reverse the orientation everywhere} \end{cases} \]
to the category
	\[ \begin{cases} \ \text{objects: von Neumann algebras;} \\ \ \text{morphisms: injective homomorphisms and antihomomorphisms.} \end{cases} \]
We require that an embedding $W_1 \hookrightarrow W_2$ is sent to an injective homomorphism 
$\A(W_1) \rightarrow \A(W_2)$ if it preserves orientation, and to an injective homomorphism $\A(W_1) \rightarrow \A(W_2)^\text{op}$ if it reverses orientation,
where $\A(W_2)^\text{op}$ is the opposite of the von Neumann algebra $\A(W_2)$.
(Note: an antihomomorphism $A\to B$ is a homomorphism $A\to B^\text{op}$).
\end{data}

\subsubsection{Conformal nets associated to loop groups}\label{s:loop group nets}

An important class of examples of conformal nets is given by \textit{loop group nets}. Let $G$ be a simply connected compact Lie group equipped with a `level'. If the group is simple, then a level is just a positive integer $k \in \mathbb{Z}_{\geq 1}$; in general, a \textit{level} is a biinvariant metric on~$G$ such that the square lengths of closed geodesics are in~$2\mathbb{Z}$. To a one-manifold~$W$ we want to assign an algebra. As an intermediate step towards this algebra, we define the group
\begin{equation}\label{eqn:L_W G}
	L_W G \coloneqq \text{Map}_*(W,G) \subset \text{Map}_{C^\infty}(W,G)
\end{equation}
of all smooth maps $W \longrightarrow G$ that send the boundary $\partial W$ to the unit $e \in G$
and all of whose derivatives are zero at the boundary. Thus, if $W=S$ is a circle, then $L_W G$ is a version of the free loops on~$G$, while if $W=I$ is an interval, it is a version of the based loops on~$G$. The group structure is given by pointwise multiplication in~$G$.

Like the loop group, this group has a central extension by $S^1$.
That central extension is easiest to describe at the level of the Lie algebra $\mathfrak{g}$ of~$G$, where it becomes a central extension by $\mathbb R$.
The Lie algebra $L_W\mathfrak{g}$ of the loop group $L_W G$ 
consists of smooth maps $W\to \mathfrak g$ all of whose derivatives are zero at the boundary. It
has a central extension defined by the cocycle\footnote{Recall that the \textit{central extension} $\hat{\mathfrak{g}}$ of a Lie algebra~$\mathfrak{g}$ is given by the vector space $\mathfrak{g} \oplus \mathbb{C}K$ with bracket \[ [X + \lambda K, Y + \mu K] = [X,Y] + c(X,Y)\, K \ , \qquad X,Y \in \mathfrak{g} \ , \quad \lambda , \mu \in \mathbb{C} . \] Here the map $c \colon \mathfrak{g} \otimes \mathfrak{g} \longrightarrow \mathbb{C}$ is a \textit{Lie algebra 2-cocycle}: it is antisymmetric and satisfies the cocycle condition $c(X,[Y,Z]) + c(Y,[Z,X]) + c(Z,[X,Y]) = 0$. This ensures that the new bracket is antisymmetric and satisfies the Jacobi identity.}
\begin{equation}\label{eqn:cocycle}
	c(f,g) = \int_W \langle \, f, \text{d}\,g \, \rangle_k \ , \qquad f, g \in L_W\mathfrak{g} \ .
\end{equation}
Here, the pairing is given by the metric and depends on the choice of level for~$G$. The corresponding central extension of $L_W G$ is the one that we are after.

The value $\A(W)=\A_{LG,k}(W)$ of the conformal net $\A_{LG,k}$ on the 1-manifold $W$ is then defined as the completion of the group algebra of $L_W G$, with multiplication twisted by the cocycle~\eqref{eqn:cocycle}. This is similar to the group algebra of the central extension,
but the central $S^1$ is identified with the $S^1$ in the scalars.
More precisely, we start by forming the free vector space $\mathbb{C}[L_W G]$; since $L_W G$ is a group, this free vector space has the structure of an algebra. The group cocycle $c \colon L_W G \times L_W G \longrightarrow \mathbb{C}^*$ corresponding to~\eqref{eqn:cocycle} allows us to modify the multiplication to $g \cdot_c h \coloneqq c(g,h) \, g \, h$. The associativity is maintained due to the cocycle condition $c(gh,k) \, c(g,h) = c(g,hk) \, c(h,k)$. Finally, the resulting twisted group algebra is not complete, and so we take some completion to make it into a von Neumann algebra.

Loop group nets are made so that they remember all the relevant information about the corresponding loop group. In particular, there is a notion of representation of a conformal net, and the representations of the loop group net agree with the `positive energy' representation of $L\,G$.\footnote{Unfortunately, the fact that representations of $\A_{LG,k}$ are the same as positive energy representation of $L\,G$ is not known in general, even though this is widely expected to be the case. It is known for $G=SU(n)$ due to results of Wassermann \cite{Wassermann1998} and partially known for $G=\mathit{Spin}(2n)$ due to Toledano-Laredo \cite{ToledanoLaredo1997}.}

\begin{definition}
A \textit{representation} of a conformal net $\A$ is a Hilbert space~$H$ equipped with compatible actions of $\A(I)$ for every proper subinterval $I \subsetneq S^1$ of the unit circle.
The category of representations of a conformal net $\A$ is denoted $\text{Rep}\,(\A)$.
\end{definition}

Note that although $S^1$ itself has a von Neumann algebra $\A(S^1)$ associated to it, there are examples of conformal nets where $\A(S^1)$ does not act on a representation~$H$. For this reason, one requires actions of the algebras associated to all manifolds $I \subsetneq S^1$ that are \textsl{strictly} contained in $S$.
Those must be compatible in the sense that the inclusions $I_1 \hookrightarrow I_2\subset S^1$ determine the restrictions of the actions. Often, a representation is equivalent to having a single action of the algebra $\A(S^1)$. We should expect this to hold for loop group nets in particular, although we do not know how this can be proven, except for $G=SU(n)$.

The category $\text{Rep}\,(\A)$ is monoidal with respect to a product $\boxtimes$ called ``fusion product'', which we now describe.
Consider two representations $H$ and $K$ of $\A$, as shown in Figure~\ref{fig:conformal net rep}. The two half-circles $I$ and $J$, with orientations induced by their inclusion in $S^1$, act as $\A(I) \acts H$ and $\A(J) \acts K$. Let $\varphi \colon I \longrightarrow J$ be the diffeomorphism that sends the `north pole' to the `north pole' and the `south pole' to the `south pole' (see again Figure~\ref{fig:conformal net rep}). Since $\varphi$ reverses the orientation, it provides an isomorphism $\A(I) \cong \A(J)^\text{op}$ and
therefore a right action of $\A(J)$ on $H$. The fusion product of $H$ and $K$ is then defined to be the Connes fusion $H \boxtimes_{\A(J)} K$.
The residual actions of $S^1 \setminus I$ and $S^1 \setminus J$ can then be used to make this product into a new representation of $\A$.
Moreover, one can show that, up to natural isomorphism, the functor $H,K\mapsto H \boxtimes_{\A(J)} K$ is independent of the choice of half-circles~$I$ and~$J$.
This also shows that the category $\text{Rep}\,(\A)$ is braided monoidal.
The construction of these natural isomorphisms is spelled out in Section \ref{s:conformal net}.

\begin{figure}[h]
	\begin{center}
	\includegraphics[scale=.10]{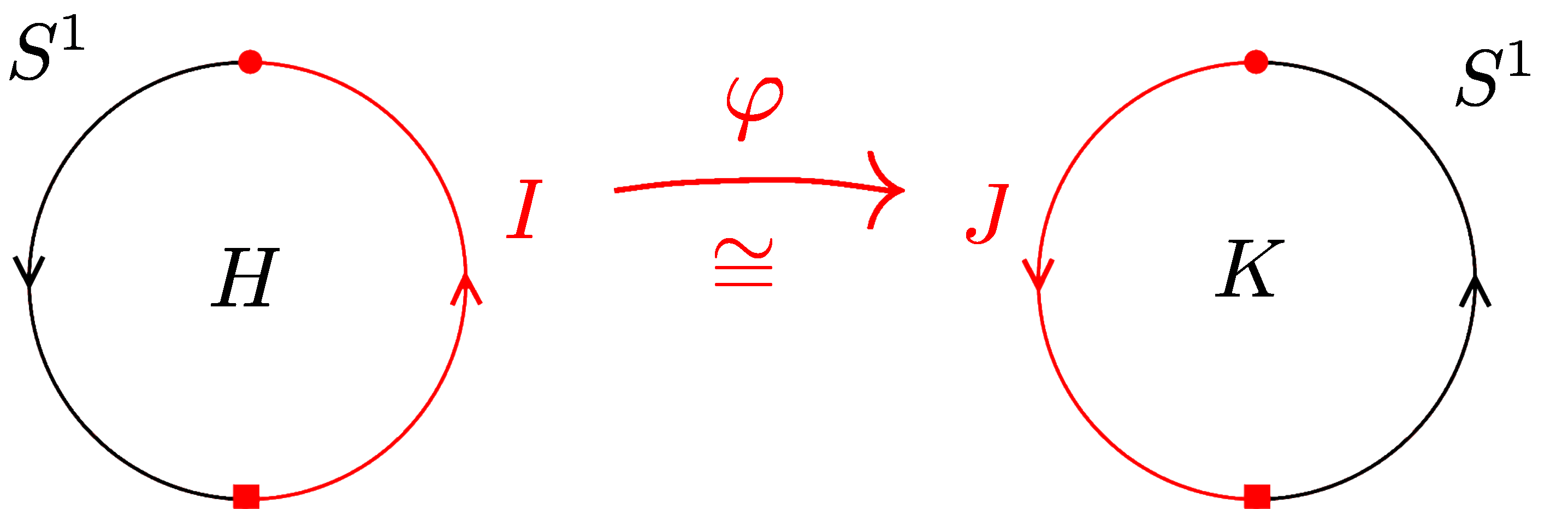}
	\end{center}\caption{The diffeomorphism used to define the fusion product $H \boxtimes_{\A(J)} K$ between two representations $H$ and $K$ of a conformal net. One can think of $H$ as having actions of all the algebras corresponding to submanifolds of the circle enclosing it, and of $K$ as having actions of algebras living on the other circle.}\label{fig:conformal net rep}
\end{figure}

In Section~\ref{s:representations} we will describe a coordinate-independent approach to the representation theory of conformal nets,
and to the fusion product of representations.

\subsubsection{The loop group of $SU(2)$}\label{s:loop group example}

To get a feeling of what representations of conformal nets are for the case of loop groups, we consider the simplest non-trivial case: the loop group $L \, SU(2)$ of $SU(2)$. Recall that $SU(2)$ has one irreducible representation $V_n$ of dimension~$n+1$ for each $n \in \mathbb{N}$. $V_0$ is the trivial representation, $V_1$ is the fundamental representation, and so on. For $m\leq n$ the tensor product of two irreducible representations is given by the Clebsch-Gordan decomposition
\begin{equation}\label{eqn:Clebsch-Gordan}
	V_m \otimes V_n \cong V_{n-m} \oplus V_{n-m+2} \oplus \cdots \oplus V_{n+m} \ , \qquad m\leq n \ .
\end{equation}
This formula is already determined by the simpler relation 
	\[ V_1 \otimes V_n \cong V_{n-1} \oplus V_{n+1} \ ,  \qquad n\geq 1 \ . \]

Likewise, the representation theory of the loop group $L\, SU(2)_k$ of $SU(2)$ at level~$k$ has irreducible representation $V_0, \cdots\mspace{-1mu},V_k$, where now each of the $V_n$ is an infinite-dimensional Hilbert space.
As before, $V_0$ is the monoidal unit and the fusion rules are entirely determined by the action of $V_1\boxtimes-$:
\begin{equation}\label{eqn:Clebsch-Gordan++}
V_1 \boxtimes V_n \cong \begin{cases} V_1 & \text{if } n = 0 \ ; \\ V_{n-1} \oplus V_{n+1} & \text{if } 1 \leq n \leq k - 1 \ ; \\ V_{k-1} & \text{if } n = k \ . \end{cases}
\end{equation}
It is a nice exercise to use the above formulas to find the analogue of~\eqref{eqn:Clebsch-Gordan} for $L\,SU(2)_k$.

We repeat that the current state of knowledge about general loop group nets is somewhat incomplete, even though it is quite clear what \textsl{should} be the case. The cases $G=SU(n)$ is the only one
where everything is known.

\subsection{Conformal nets revisited}\label{s:conformal net}

It is convenient to change the source category of conformal nets a bit: henceforth we restrict ourselves to \textsl{contractible} compact one-manifolds, i.e., to intervals. The circles can be recovered by gluing two intervals together.

\begin{defi}\label{def:conf nets}\rm
A \textit{conformal net}\footnote{Note that this definition differs from the definitions in the literature; see e.g.~\cite{GF1993, KL2004, Longo2008}.
Our definition is somewhat more general, it allows for more examples.} $\A$ is a continuous functor from the category
	\[ \begin{cases} \ \text{objects: contractible compact oriented one-manifolds;} \\ \ \text{morphisms: embeddings} \end{cases} \]
to the category
	\[ \begin{cases} \ \text{objects: von Neumann algebras;} \\ \ \text{morphisms: injective homomorphisms and antihomomorphisms.} \end{cases} \]
An embedding $I \hookrightarrow J$ is sent to a homomorphism $\A(I) \hookrightarrow \A(J)$ if it preserves the orientation, and to a homomorphism $\A(I) \hookrightarrow \A(J)^\text{op}$ if it reverses the orientation.
The hom-sets of the source category carry the $C^\infty$ topology, and there is also a topology on the hom-sets of the target category. It is with respect to these topologies that $\A$, mapping $I \hookrightarrow J$ to an (anti)homomorphism, has to be continuous. 
Moreover, a conformal net $\A$ is subject to the following axioms:
\begin{enumerate}
	\item[i)] The algebras $\A\big([0,1]\big)$ and $\A\big([1,2]\big)$ commute in, and generate a dense subalgebra of $\A\big([0,2]\big)$;
	\item[ii)] Denoting the algebraic tensor product by $\otimes_\text{alg}$
and the so-called \textsl{spatial tensor product} of von Neumann algebras by $\bar \otimes$, there exists an extension
that makes the diagram commute:
\[
\tikzmath{
\node(a) at (0,1.5) {$\A\big([0,1]\big) \otimes_\text{alg} \A\big([2,3]\big)$};
\node(b) at (5,1.5) {$\A\big([0,3]\big)$};
\node(c) at (0,0) {$\A\big([0,1]\big) \, \bar\otimes \, \A\big([2,3]\big)$};
\draw[->] (a) -- (b);
\draw[->] (a) -- (c);
\draw[->, dotted] (c) -- (b);
} 
\]
	\item[iii)] The image of the map
\[
\big\{ \ \varphi \in \text{Diff}\big([0,3]\big) \ \ \colon \ \varphi|_{[0,1]}=\text{id} \, , \ \varphi|_{[2,3]}=\text{id} \ \big\} \longrightarrow \text{Aut}\big(\A([0,3])\big)
\]
		is contained in the set of inner automorphisms of $\A([0,3])$;
	\item[iv)] There exists a dotted map such that the diagram
\[
\tikzmath{
\node(a) at (-1.5,1.5) {$\A\big([0,1]\big)^\text{op} \otimes_\text{alg} \A\big([0,1]\big)$};
\node(b) at (5,1.5) {$\A\big([0,2]\big)^\text{op} \otimes_\text{alg} \A\big([0,2]\big)$};
\node(c) at (-1.5,0) {$\A\big([-1,0]\big) \otimes_\text{alg} \A\big([0,1]\big)$};
\node(d) at (2,0) {$\A\big([-1,1]\big)$};
\node(e) at (5,0) {$B\big(L^2 \A\big([0,2]\big)\big)$};
\draw[->] (a) -- (b);
\draw[->] (b) -- (e);
\draw[->] (a) --node[left]{$\scriptstyle (x\mapsto -x)\otimes (y\mapsto y)$} (c);
\draw[->] (c) -- (d);
\draw[->, dotted] (d) -- (e);
} 
\]
		commutes.
\end{enumerate}
\end{defi}

We pause to explain why $\text{Rep}\,(\A)$ is braided.
Let $I$ and $J$ be two halves of the standard circle, as in Figure~\ref{fig:conformal net rep}.
Given $H, K\in \text{Rep}\,(\A)$, we need to construct the braiding isomorphism $H\boxtimes_{\A(J)} K \to K\boxtimes_{\A(J)} H$ in two steps.
It is the composite of two `quarter-braiding' isomorphisms
\[
H\boxtimes_{\A(J)} K \,\longrightarrow\, H\boxtimes_{\A(J_1)} K \,\longrightarrow\, H\boxtimes_{\A(I)} K \cong K\boxtimes_{\A(J)} H
\]
where $I_1$ and $J_1$ form another decomposition of $S^1$ into half-circles (for example top and bottom halves).
We focus on the first isomorphism, between $H\boxtimes_{\A(J)} K$ and $H\boxtimes_{\A(J_1)} K$.
Let $\varphi$ be a diffeomorphism of the circle that sends $I$ to $I_1$ and whose support does not cover the whole of $S^1$.
Similarly, let $\psi$ be a diffeomorphism of the circle that sends $J$ to $J_1$ and whose support does not cover the whole of $S^1$.
By axiom (iii) of the definition of conformal nets, there exist unitaries $u$ and $v$ such that $\text{Ad}(u)=\A(\varphi)$ and $\text{Ad}(v)=\A(\psi)$.
Multiplication by $u$ on $H$ and by $v$ on $K$ induce a map $H\boxtimes_{\A(J)} K\to H\boxtimes_{\A(J_1)} K$.
Unfortunately this map does not have the right equivariance properties to be a morphism in $\text{Rep}\,(\A)$.
To fix that, we consider the diffeomorphism $\varphi|_{J}\cup\psi|_I$.
Once again, its support is not the whole circle, and so by axiom (iii) we can find a unitary operator $w$ that corresponds to it, in one of the algebras that act on $H\boxtimes_{\A(J_1)} K$.
The quarter-braiding isomorphism is the composite
\[
H\boxtimes_{\A(J)} K\,\xrightarrow{\,u\boxtimes v\,}\, H\boxtimes_{\A(J_1)} K \,\stackrel{w^*}\longrightarrow\, H\boxtimes_{\A(J_1)} K.
\]

There is a subtle point that we should mention: conformal nets have two roles in life.
Although the relation with Segal's definition of $\chi$\textsc{cft} may not be clear from the above definition, conformal nets, or rather, a subset of them, serve as a model for $\chi$\textsc{cft} (we should emphasize that, as far as the math is concerned, the relationship between conformal nets and other models of $\chi$\textsc{cft} is completely conjectural).
On the other hand, conformal nets serve as a model for three-dimensional \textsc{tqft}, such as Chern-Simons theory.
The conformal nets described in Definition~\ref{def:conf nets} correspond to 3d \textsc{tqft}s.\footnote{Actually, a conformal net $\A$ only corresponds to 
a genuine 3d \textsc{tqft} (i.e, defined on all 3-bordisms) if a certain numerical invariant, the $\mu$-index of $\A$, is finite.}
In order to have an associated $\chi$\textsc{cft}s, a conformal net needs to satisfy a further `positive energy' condition.
The latter says that, under the map in axiom~(iii), the flow of a positive vector field in $\text{Diff}([0,3])$ correspond to a one-parameter group of unitaries in
$\A([0,3])$ with \emph{positive} generator (this generator is only well defined up to an additive constant).
The loop group conformal nets satisfy both conditions\footnote{To be precise, the finite $\mu$-index condition has only been proven for $SU(n)$ \cite{Xu00}. It is expected to hold for all loop group nets.} and so they correspond to both a three-dimensional \textsc{tqft} and a two-dimensional $\chi$\textsc{cft}.

We should point out that, at least conjecturally, the $\chi$\textsc{cft} associated to a conformal net $\A$ (satisfying the positive energy condition)
maps the circle to the representation category $\text{Rep}\,(\A)$, so that the Frobenius algebra objects
that occur in Section \ref{s:FRS} and in Section \ref{s:construction} live in the same world.
Also, we will not need the positive energy condition for the construction of the zero- and one-dimensional parts of the extended \textsc{cft}.
That condition only gets used when constructing the operator associated to a bigon (such as the one in Figure \ref{fig:Sigma}).

\subsection{Frobenius algebra objects}\label{s:Frobenius}

In Section~\ref{s:chiral CFT} we have seen that a $\chi$\textsc{cft} assigns to the standard circle a category $\C$, and that
the pair of pants equips that category with a monoidal structure.
In our example of interest, this is the category
	\[ \C \coloneqq \text{Rep}\,(\A_{LG,k}) \,\stackrel{\scriptscriptstyle (?)}\cong\, \text{Rep}\,(L\,G_k) \]
of representations of the conformal net $\A_{LG,k}$ associated to the loop group of $G$ at level~$k$. 
This category is expected to be equivalent to the category of positive energy representations of $L\,G$ at level $k$.
We are interested in objects of $\C$ with a particular kind of extra structure,
which can be defined in any monoidal dagger category (a dagger category is a category $\mathcal C$ equipped with involutive antilinear maps $\dagger\!:\text{Hom}(X,Y)\to\text{Hom}(Y,X)$, $X,Y\in\mathcal C$ that assemble to a functor $\dagger:\mathcal C\to \mathcal C^\mathrm{op}$). Indeed, our category $\C$ consists of Hilbert spaces, so there is a notion of adjoints that turns it into a monoidal dagger category. Here, as before, the monoidal structure is given by the fusion product.

\begin{definition}
A special symmetric \textit{Frobenius algebra object} (we will simply call them Frobenius algebra objects) is an object~$Q \in \C$ together with maps
\begin{itemize}
	\item multiplication $m \colon Q \boxtimes Q \longrightarrow Q$\,,
	\item unit $m: 1 \longrightarrow Q$\,,\hfill (here $1$ stands for the unit object of $\C$)
	\item comultiplication $\Delta \colon Q \longrightarrow Q \boxtimes Q$\,, and
	\item counit $\varepsilon \colon Q \longrightarrow 1$\,,
\end{itemize}
subject to the axioms shown in Figure~\ref{fig:Frobenius alg obj axioms}.
\end{definition}

\begin{figure}[h]
	\begin{center}
	\includegraphics[scale=.10]{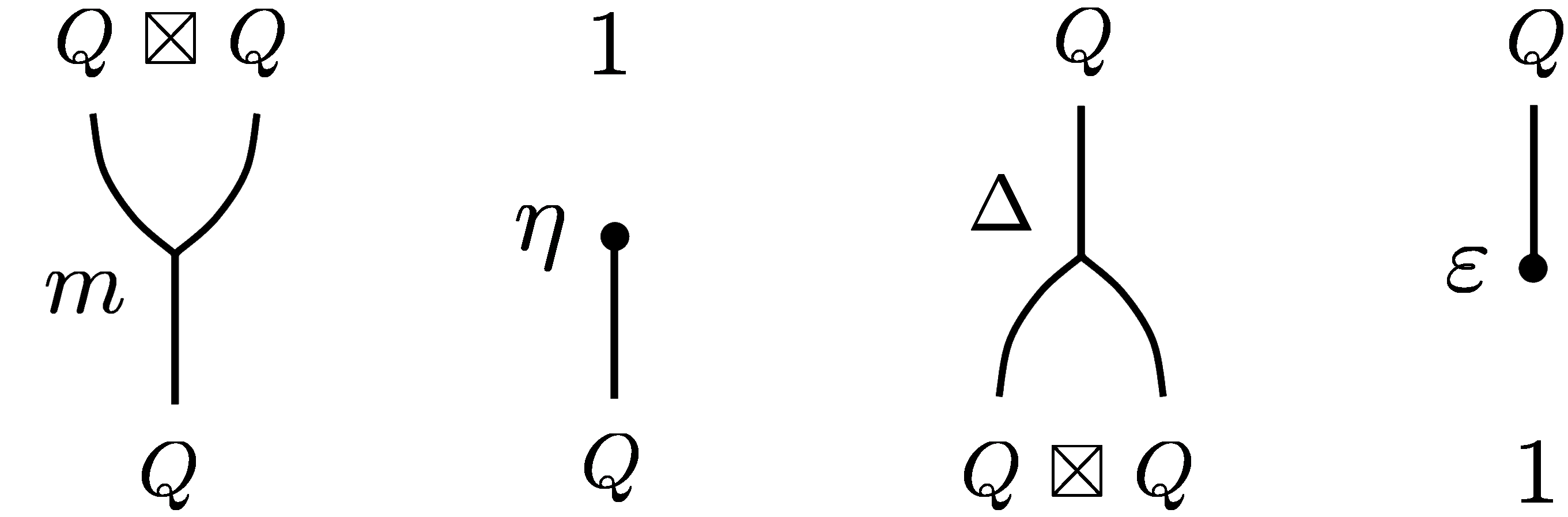}
	\end{center}\caption{The building blocks of the string diagrams used in the definition (see Figure ~\ref{fig:Frobenius alg obj axioms}) of a Frobenius algebra object: (co)multiplication and (co)unit. The diagrams are read from top to bottom. The precise shape of the strings is not important. Because of the distinctive shapes the labels are usually omitted.}\label{fig:Frobenius alg obj}
\end{figure}

\begin{figure}[h]
\centerline{\raisebox{.87cm}{i)}\,\; \includegraphics[scale=.05]{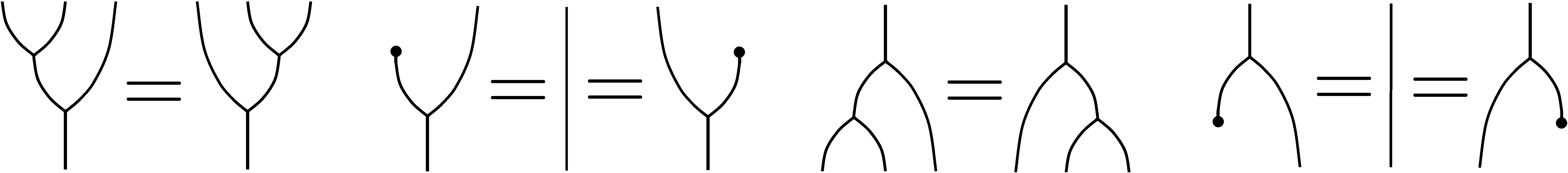}}\bigskip
\centerline{\raisebox{.87cm}{ii)}\,\, \includegraphics[scale=.05]{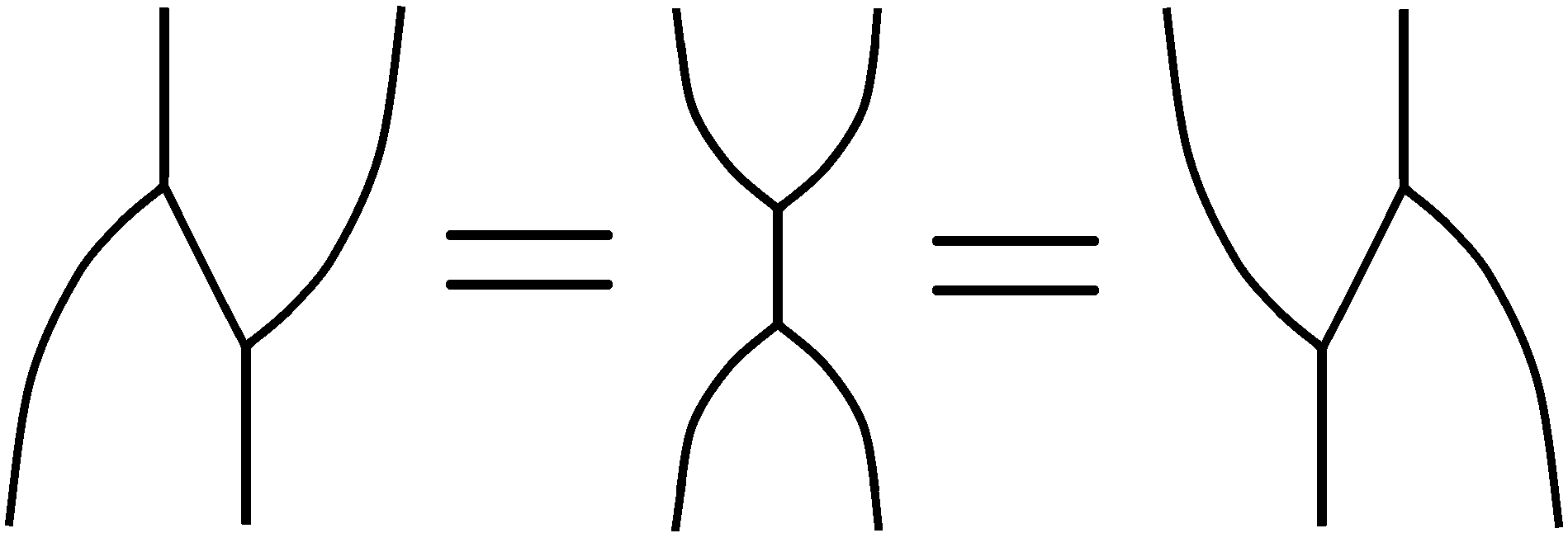} \quad \ \  \raisebox{.87cm}{iii)}\,\, \includegraphics[scale=.06]{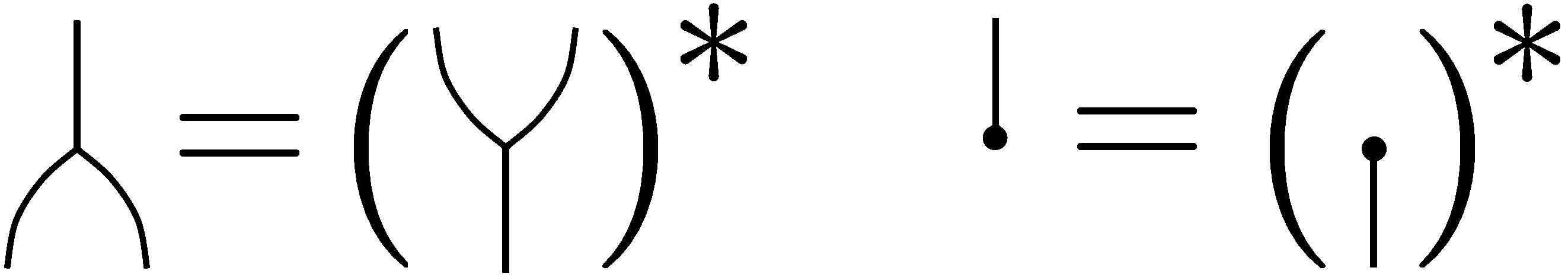}}\bigskip
\centerline{\raisebox{1.5cm}{iv)}\,\,\, \includegraphics[scale=.06]{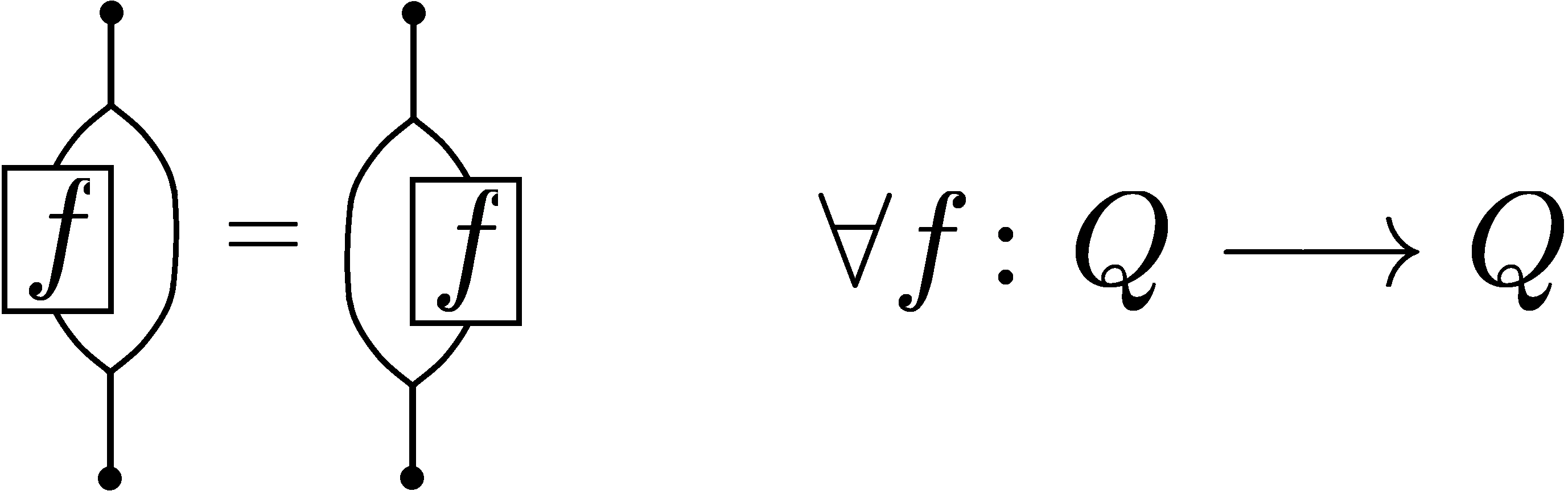} \quad\,\,\,\, \ \  \raisebox{1.5cm}{v)}\,\,\, \includegraphics[scale=.07]{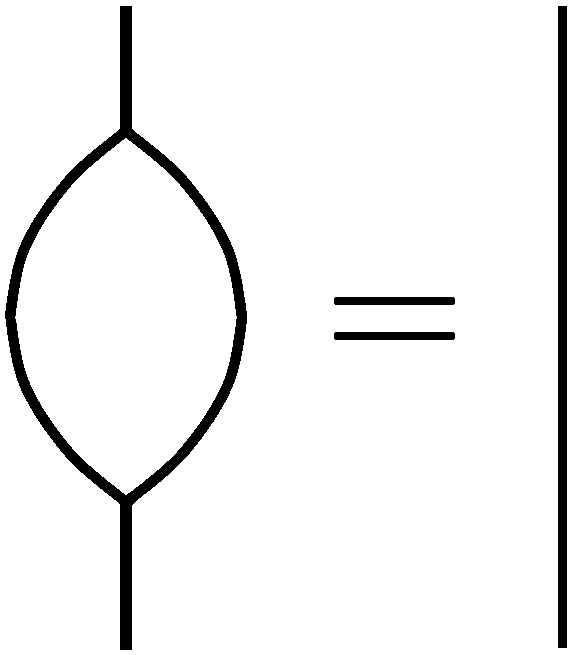}}
\caption{The axioms for a Frobenius algebra object.}\label{fig:Frobenius alg obj axioms}
\end{figure}

Axiom~(i) simply states that multiplication and comultiplication are associative and unital. Axiom~(ii) is called the \textit{Frobenius condition}.
The third axiom implies that the coalgebra structure on $Q$ is determined by its algebra structure, by taking adjoints.

Axiom~(iv) requires $Q$ to be \textit{symmetric}, 
and is equivalent to the condition
$\tikzmath[scale=.16]{\draw[line width=.8] (0,0) arc (-180:0:1) -- +(0,2.2) (0,0) arc (0:180:1) -- +(0,-2) (-1,1) -- +(0,1); \fill (-1,2) circle (.2); \draw [->, line width=.8] (.8,-1) -- +(-.01,0);} = 
\tikzmath[scale=.16]{\pgftransformxscale{-1} \draw[line width=.8] (0,0) arc (-180:0:1) -- +(0,2.2) (0,0) arc (0:180:1) -- +(0,-2) (-1,1) -- +(0,1); \fill (-1,2) circle (.2); \draw[->, line width=.8] (.8,-1) -- +(-.01,0);}$
from \cite{FRS2002}.

Finally, axiom (v) is called the \textit{special} property. The special property means that a Frobenius algebra object is very different from e.g.~cohomology rings of manifolds. In particular, it implies that the algebra $Q$ is \textit{semisimple}: any module over $Q$ is semisimple.

The definition of a Frobenius algebra object may look complicated, but a Frobenius algebra is just an algebra satisfying certain properties: everything is determined by the multiplication and unit maps. 


\subsubsection{Examples}

To get a feeling for what Frobenius algebras can look like, let us have a look at some examples.
Since every semisimple algebra is a direct sum of simple algebras, we restrict our attention to simple algebras.

A trivial example of a Frobenius algebra object is the unit object of the monoidal category.

For another example, consider an object $X \in \C$ and form the Connes fusion $Q = X \boxtimes X^\vee$ of the object with its dual. Then $Q$ is an algebra, indeed a Frobenius algebra. This is the correct generalization of matrix algebras to this context. For instance, taking $\C = \text{Rep}\big(LSU(2)_2\big)$ and $X=V_1$, then $X^\vee \cong V_1$ too and~\eqref{eqn:Clebsch-Gordan++} yields $Q = V_0 \oplus V_2$.\footnote{Another way to understand this example is as follows. The subcategory of $\C$ spanned by $V_0$ and~$V_2$ is equivalent to $\mathbb{Z}_{/2}$-graded vector spaces as a monoidal category. Inside there, we have the Clifford algebra $\langle \, e \mid e^2 = 1 , e \text{ is odd} \, \rangle.$} But $X \boxtimes X^\vee$ is Morita equivalent to the unit object (see Section \ref{s:defects}), so this is still not a very interesting example. However, it leads us to the next example.

Let $\C_k = \text{Rep}\big(LSU(2)_k\big)$ for arbitrary $k$. In $\C_k$, we have
\begin{equation}\label{eqn:Frob alg example}
	V_0 \boxtimes V_0 \cong V_0 \ , \quad V_0 \boxtimes V_k \cong V_k \ , \quad V_k \boxtimes V_0 \cong V_k \ , \quad V_k \boxtimes V_k \cong V_0 \ .
\end{equation}
These relations show that the full subcategory of $\mathcal C_k$ consisting of objects isomorphic to sums of $V_0$ and $V_k$ is again a monoidal category.
The monoidal structure is not fully determined by the relations \ref{eqn:Frob alg example}; rather, up to equivalence of
monoidal categories, there are two different monoidal categories that satisfy \ref{eqn:Frob alg example}.
One is the category of $\mathbb Z_{/2}$-graded vector spaces, and the other is a version of it where the associator is
twisted by a cocycle $c$ representing a non-trivial cohomology class $[c]\in H^3(\mathbb{Z}_{/2},S^1)\cong \mathbb{Z}_{/2}$.
It turns out that the subcategory of $\mathcal C_k$ is equivalent to the category of $\mathbb Z_{/2}$-graded vector spaces only if $k$ is even,
and $Q = V_0  \oplus V_k$ has a Frobenius algebra structure if and only if $[c] = 0$, if and only if $k$ is even.


\subsubsection{Classification}

There is a beautiful classification of all simple Frobenius algebra objects in $\C_k = \text{Rep}\big(LSU(2)_k\big)$ due to Ostrik~\cite{Ostrik2003} (inspired by the CIZ classification of modular invariants for $\hat{\mathfrak{sl}}(2)$ \textsc{cft}s~\cite{CIZ1987}) that goes as follows.

Up to Morita equivalence, Frobenius algebra objects in $\C_k$ fall in two infinite families corresponding to the following Dynkin diagrams:
\begin{align*}
	& \text{type $A_n$:} \qquad Q = V_0 \ , \quad k = n-1 \ , \quad \quad(n\ge 1) \\
	& \text{type $D_n$:} \qquad Q = V_0 \oplus V_k \ , \quad k = 2n-4 \ , \quad \quad(n\ge 4).\\
\intertext{In addition, there are three exceptional cases:}
	& \text{type $E_6$:} \qquad Q = V_0 \oplus V_6 \ , \quad k = 10 \ , \\
	& \text{type $E_7$:} \qquad Q = V_0 \oplus V_8 \oplus V_{16} \ , \quad k = 16 \ , \\
	& \text{type $E_8$:} \qquad Q = V_0 \oplus V_{10} \oplus V_{18} \oplus V_{28} \ , \quad k =28 \ .
\end{align*}
For each type we have listed a representative $Q$ of the Morita equivalence class.

\subsection{The FRS construction}\label{s:FRS}

In Section~\ref{s:loop group nets} we have mentioned that there is a construction that takes a chiral \textsc{cft} as input, along with a Frobenius algebra object in the category $\C$ provided by the $\chi$\textsc{cft}, and produces a full \textsc{cft} as output.

In the realm of algebraic quantum field theory, this result is due to Longo and Rehren~\cite{LR2004, KL2004a}. They start with a conformal net and a Frobenius algebra object, and construct a net of von Neumann algebras on $\mathbb{R}^2$ with its Minkowski signature. Such a net assign von Neumann algebras to open subsets of $\mathbb{R}^2$ in such a way that the algebras commute if the opens are causally separated.

Instead of elaborating on this construction we will discuss another approach, which is due to Fuchs, Runkel and Schweigert~\cite{FRS2002,FRS2004,FRS2004a,FRS2005,FRS2006}. This is a big body of work, and we will only outline some of its aspects.

\subsubsection{The partition function}

Let us at least describe how to take a $\chi$\textsc{cft} and a Frobenius algebra object and assign a number $Z(\Sigma)\in\mathbb{C}$ to a closed Riemann surface $\Sigma$ (c.f.~the discussion in Section~\ref{s:Segal CFT}).

A $\chi$\textsc{cft} assigns to $\Sigma$ a functor $\C_\text{in} \longrightarrow \C_\text{out}$. Recall that a non-extended $\chi$\textsc{cft} assigns to a closed surface $\Sigma$ a linear map $\mathbb{C}\longrightarrow\mathbb{C}$, which is completely determined by an element of $\mathbb{C}$. In the present context, something similar happens. The category $\mathsf{Vect}$ is the unit object of the target category $\mathsf{LinCat}$ of $\mathbb C$-linear categories. Indeed, $\mathsf{LinCat}$ is equipped with a tensor product operation, say $\otimes$, such that, for any linear category $\C$, $\mathsf{Vect} \otimes \C = \C$. Now a linear functor $f\colon \mathsf{Vect}\longrightarrow\mathsf{Vect}$ is completely determined by the image $V \coloneqq f(\mathbb{C})$, so for any $X\in\mathsf{Vect}$ we have $f(X)=X\otimes V$. The vector space $V$ associated to the functor $\mathcal C_{\text{in}}\longrightarrow \mathcal C_{\text{out}}$ is called the \textit{space of conformal blocks} associated to $\Sigma$ by the $\chi$\textsc{cft}. There is also a canonical element $\omega\in V$ provided by the structure of the $\chi$\textsc{cft}: it is the image
\begin{align*}
	\mathbb{C} = H_\lambda \longrightarrow H_{f(\lambda)} = V \ , \quad 	1  \longmapsto \ \omega \ ,
\end{align*}
for $\lambda=\mathbb{C} \in \mathsf{Vect}=\C_\text{in}$.

\begin{figure}[h]
	\begin{center}
	\includegraphics[scale=.11]{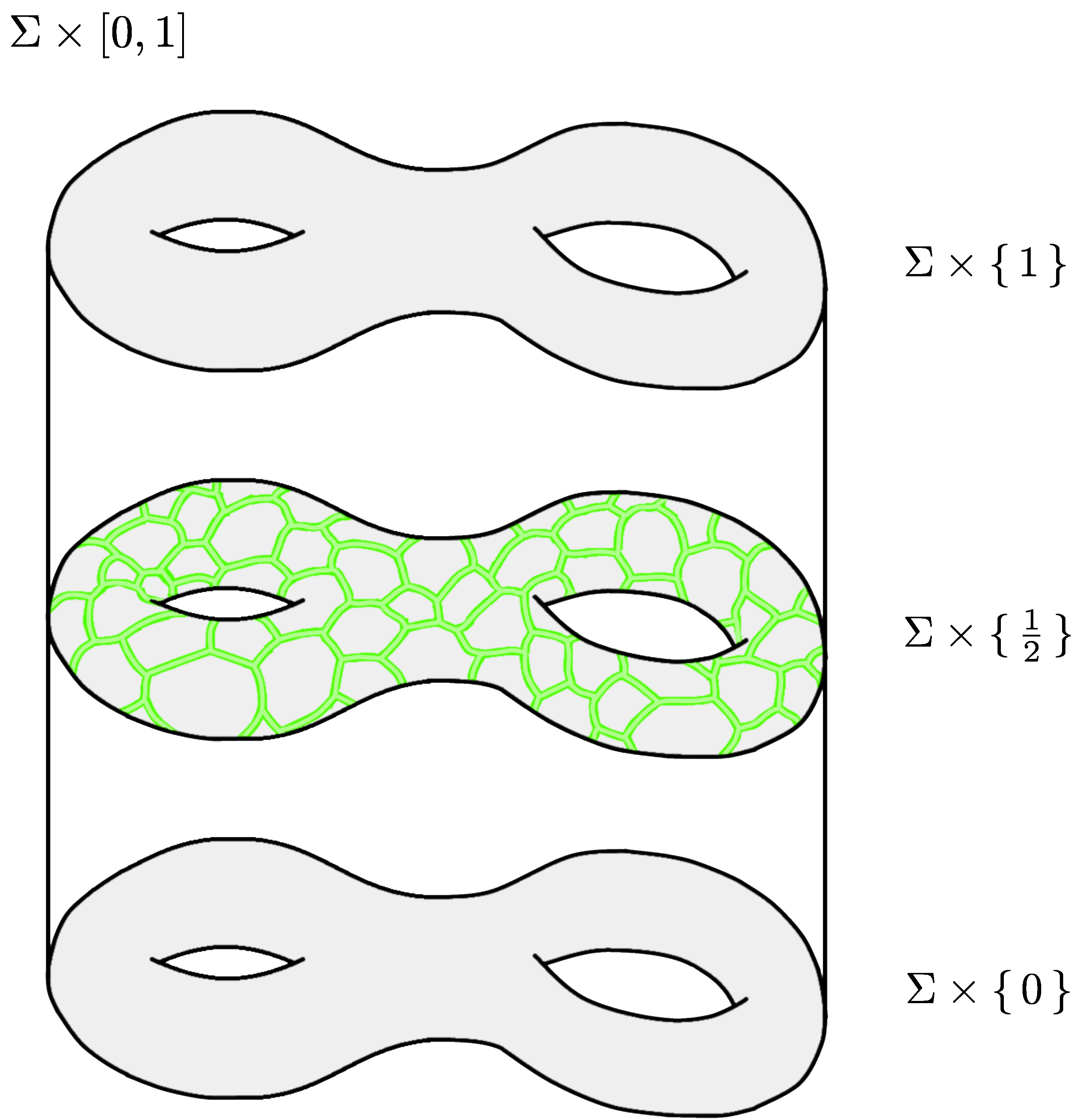}
	\end{center}\caption{The product of a closed Riemann surface $\Sigma$ with the unit interval. The middle slice is decorated by a ribbon graph with trivalent vertices.
Actually, the middle slice $\Sigma\times \{\frac12\}$ should not be there in the picture; it is only needed to explain where the ribbon graph sits.
	}\label{fig:FRS}
\end{figure}

To see where the Frobenius algebra object comes in, consider the 3-manifold $\Sigma \times [0,1]$
which is obtained by crossing $\Sigma$ with a unit interval.
Decorate the middle slice $\Sigma\times\{\frac 12\}$ with a \textit{ribbon graph}, whose edges are `thickened' to little two-dimensional ribbons, as shown in Figure~\ref{fig:FRS}. 
We only allow for trivalent vertices. Now give the ribbons an orientation, so as to get a directed ribbon graph.
This can always be done in such a way that each vertex has at least one incoming ribbon and at least one outgoing ribbon. This allows us to further color the graph with the Frobenius algebra object~$Q$ and its multiplication~$m$ and comultiplication~$\Delta$ according to the rules shown in Figure~\ref{fig:colored ribbon graph}.

\begin{figure}[h]
	\begin{center}
	\includegraphics[scale=.11]{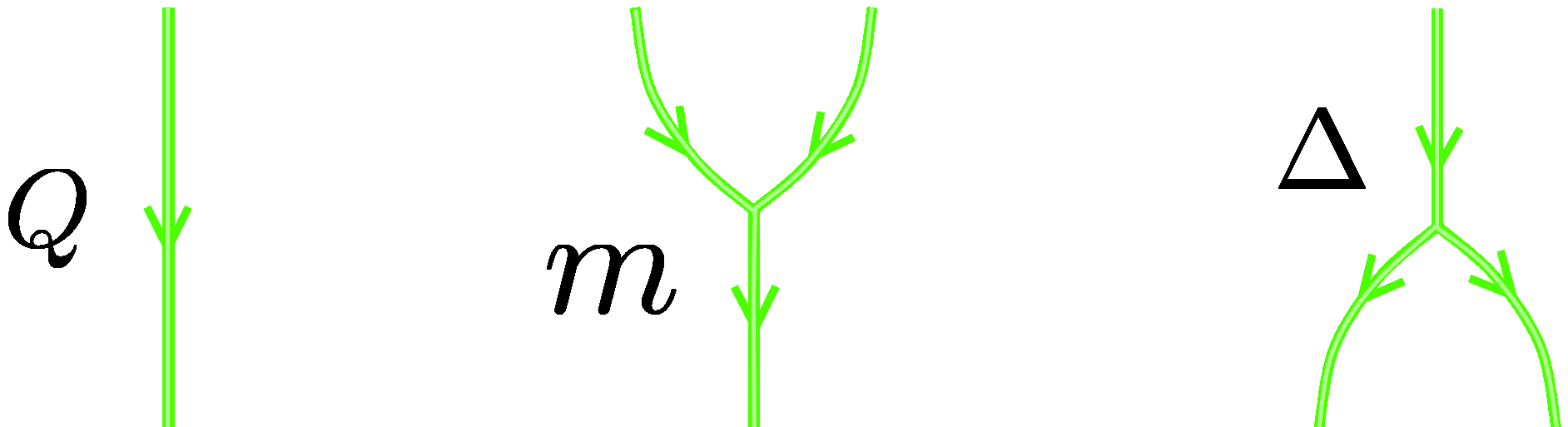}
	\end{center}\caption{A directed ribbon graph with trivalent vertices can always be colored by these rules (cf.~the string diagrams for the (co)multiplication of Frobenius algebra objects in Figure~\ref{fig:Frobenius alg obj}).}\label{fig:colored ribbon graph}
\end{figure}

In order to assign a number to $\Sigma$, FRS invoke the existence of a three-dimensional topological quantum field theory (\textsc{tqft}).\footnote{This 
not just a plain \textsc{tqft} that associates operators to 3-bordisms,
but these 3-bordisms can come equipped with suitably colored ribbon graphs.
Moreover, the ends of those ribbon graphs on the boundary of a bordism $\Sigma$ can give colored marked points on $\partial \Sigma$.}
The latter assigns to the three-manifold $\Sigma\times[0,1]$ with colored ribbon graph an element $c \in V\otimes\bar V$ of the conformal blocks of $\partial\big(\Sigma\times[0,1]\big)=\Sigma\sqcup\bar{\Sigma}$.
Here, $\bar\Sigma$ denotes the manifold $\Sigma$ with the reversed orientation.
The partition function is then given by
	\[ Z(\Sigma) =\langle \, c \, , \, \omega\otimes\bar \omega \, \rangle_{V\otimes\bar V} \in \mathbb{C} \ . \]
Using the axioms of a (special symmetric) Frobenius algebra object, one can then show that
this number does depend neither on the choice of ribbon graph nor on the orientation of the ribbons.

\subsubsection{The state space}\label{s:state space}

The next question concerns the state space of the \textsc{cft}: what is the Hilbert space associated to a circle? For that, one considers the same kind of picture as above, but with a hole, as shown in Figure~\ref{fig:FRS hole}.

\begin{figure}[h]
	\begin{center}
	\includegraphics[scale=.11]{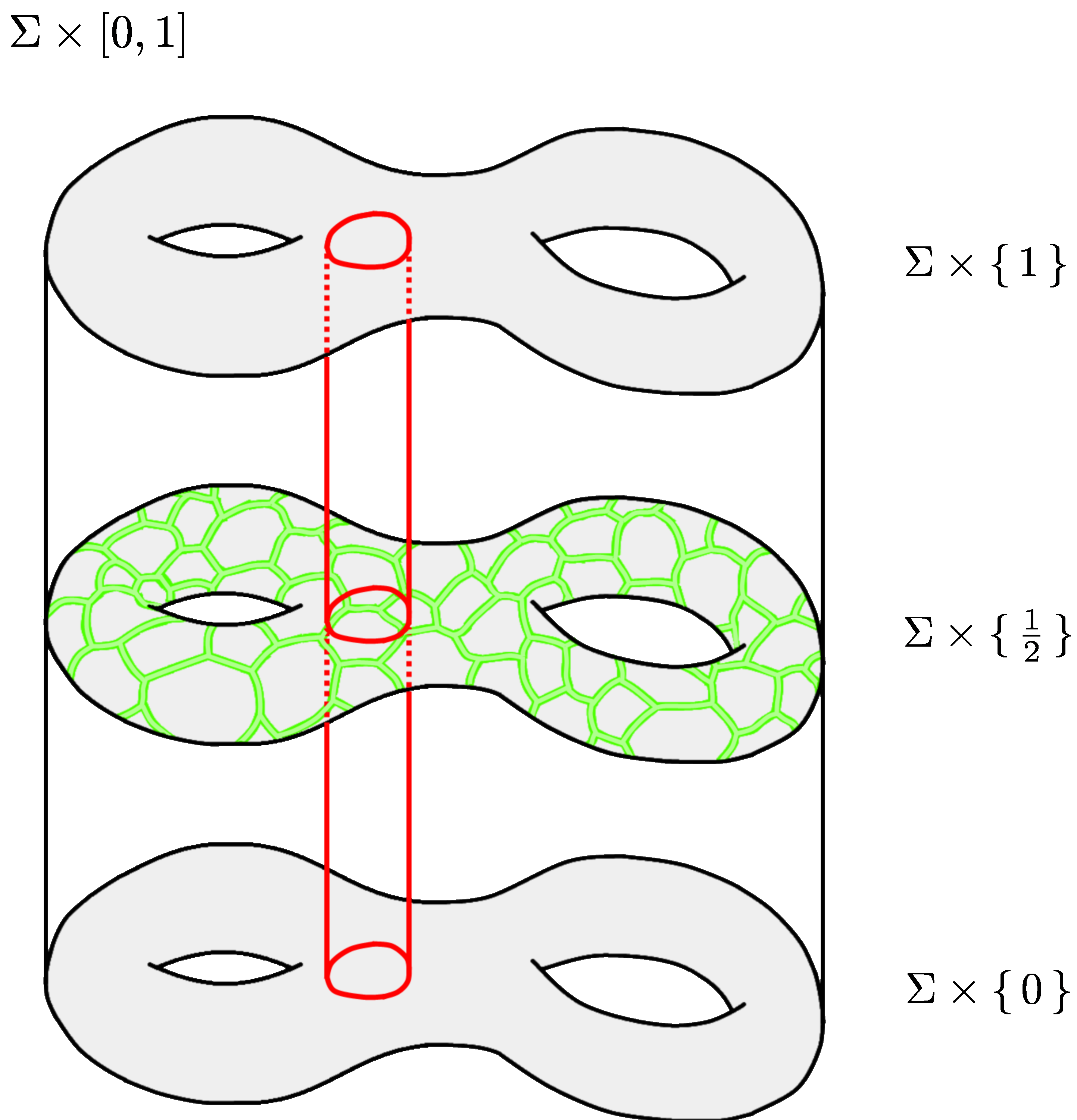}
	\end{center}\caption{The product of a closed Riemann surface $\Sigma$ with the unit interval, with a cylinder taken out.}\label{fig:FRS hole}
\end{figure}

We would like to add a single ribbon going down the middle of the hole, in such a way that it is compatible with the ribbon graph on the surface $\Sigma\times\{\frac 12\}$. Imagine a single strand coming from above, along with another one coming from below. The question is: what do we label these by? Since there are two strands (one from above, and one from below), we are looking for two objects of~$\C$.
More accurately, we are looking for one object in $\C$ and one in $\bar{\C}$, where the bar now stands for complex conjugation (the category $\bar \C$ has the same objects as $\C$, but $\text{Hom}_{\bar \C}(\lambda,\mu)$ is the complex conjugate of $\text{Hom}_{\C}(\lambda,\mu)$. Given an object $\lambda\in\C$, we shall denote by $\bar\lambda$ the corresponding object of $\bar\C$. The Hilbert space $H_{\bar\lambda}$ associated to $\bar\lambda\in\bar\C$ is then the complex conjugate of the Hilbert space $H_\lambda$ associated to $\lambda\in\C$).
Equivalently, we are looking for a single object of $\C \times \bar{\C}$. It turns out that this is not quite general enough. What we are really after is an object of $\C\otimes \bar{\C}$, that is, a formal direct sum of objects of $\C \times \bar{\C}$.

To see what the compatibility condition is, consider a juncture of the ribbon in the cylinder with a ribbon from $\Sigma\times\{\frac 12\}$, as depicted in Figure~\ref{fig:ribbon junction}. The ribbon graph can always be arranged so that there is such a junction.

\begin{figure}[h]
	\begin{center}
	\includegraphics[scale=.15]{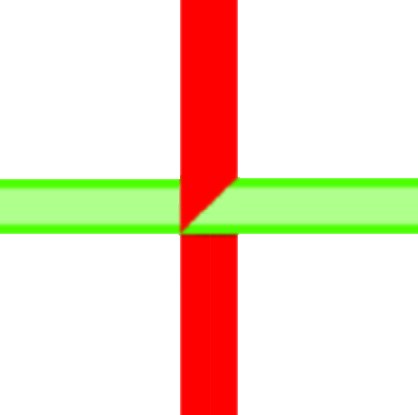}
	\end{center}\caption{A junction between the ribbon that is dangling down in the cylinder and a ribbon from the graph on $\Sigma \times \{ \, \frac 12 \, \}$.}\label{fig:ribbon junction}
\end{figure}

The value (object of $\C \otimes \bar{\C}$) that is assigned to the ribbon on the cylinder has to satisfy the compatibility requirements shown in Figure~\ref{fig:ribbon junction axioms}. 

\begin{figure}[h]
	\begin{center}
	\includegraphics[scale=.15]{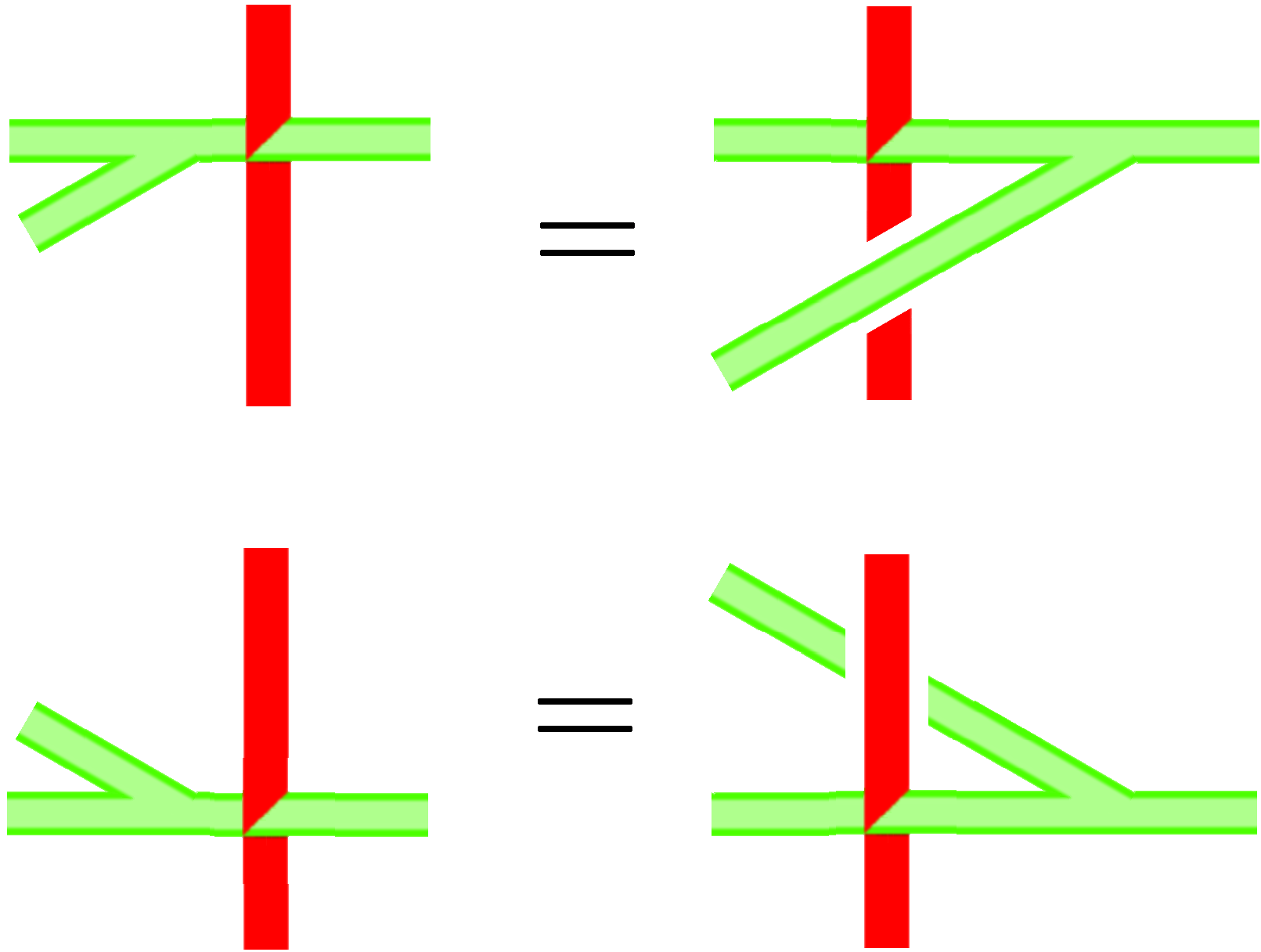}
	\end{center}\caption{Compatibility requirements between the vertical ribbon in the cylinder (cf.~Figure~\ref{fig:FRS hole}) and the ribbons of the graph on $\Sigma\times  \{ \, \frac 12 \, \}$.}\label{fig:ribbon junction axioms}
\end{figure}

It turns out that there is an object that is universal with respect to these properties: the \textit{full centre of $Q$}, given by\footnote{The `$Z$' in \eqref{eqn:full centre} stands for centre, and should not be mistaken for the partition function.}
\begin{equation}\label{eqn:full centre}
	Z_\text{full}(Q) =\bigoplus_{\mu,\lambda} \text{Hom}_{Q,Q}\big(\lambda \boxtimes^+ Q \boxtimes^- \mu^\vee,Q\big) \otimes \lambda \otimes \bar{\mu} \ \in \C \otimes \bar{\C}\ .
\end{equation}
This equation require some explanations.
First of all, $\mu^\vee$ is the dual of $\mu$, characterized by the existence of a non-zero map $1\to \mu\boxtimes \mu^\vee$.
The symbols $\boxtimes^\pm$ denote the tensor product from above and below (as opposed to from right or from the left), which can be defined because the monoidal category $\C$ is braided\footnote{For an object $X$ in a braided category, the space of possible ``multiplications by $X$'' is a circle, in which $X\boxtimes -$ and $-\boxtimes X$ are only two points.}. 
So $\lambda \boxtimes^+ Q \boxtimes^- \mu^\vee$ could more accurately be drawn as
$\tikzmath{\def\h{.3}
\node at (0,0){$\scriptstyle Q$};\node at (0,\h){$\scriptstyle \boxtimes$};\node at (0,-\h){$\scriptstyle \boxtimes$};\node at (0,2*\h){$\scriptstyle \lambda$};\node at (.05,-2*\h){$\scriptstyle \mu^\vee$};}
$.\\
That object is a $Q$-$Q$-bimodule: it comes with maps
$Q\boxtimes(\lambda \boxtimes^+ Q \boxtimes^- \mu^\vee)\to \lambda \boxtimes^+ Q \boxtimes^- \mu^\vee$ and 
$(\lambda \boxtimes^+ Q \boxtimes^- \mu^\vee)\boxtimes Q\to \lambda \boxtimes^+ Q \boxtimes^- \mu^\vee$ induced from the left and right actions of $Q$ on itself.
(The reader who finds $\boxtimes^\pm$ unpleasant can take $\lambda \boxtimes^+ Q \boxtimes^- \mu^\vee$ to simply mean $\lambda \boxtimes Q \boxtimes \mu^\vee$;
the braiding is then used to endow this object with the structure of a $Q$-$Q$-bimodule.)
Finally, $\text{Hom}_{Q,Q}$ in \eqref{eqn:full centre} refers to the space of bimodule homomorphisms.
Since this is a vector space, while $\lambda\in\C$ and $\bar\mu\in\bar \C$, the full centre of $Q$ is an object of $\C \otimes \bar{\C}$. 

The state space of the full \textsc{cft} associated to the $\chi$\textsc{cft} together with the Frobenius algebra object~$Q$ is then given by
\begin{equation}\label{eqn:H_full}
	H_\text{full} \coloneqq \bigoplus_{\mu,\lambda} \text{Hom}_{Q,Q}\big(\lambda \boxtimes^+ Q \boxtimes^- \mu^\vee,Q\big) \otimes H_\lambda \otimes \overline{H_\mu} \ .
\end{equation}
In the `Cardy case', where $Q=1$ is the unit object, this expression reduces to $H_\text{full} = \bigoplus_\lambda H_\lambda \otimes \overline{H_\lambda}$.

Equation~\eqref{eqn:H_full} is the result of FRS that we were after. The discussion in Section~\ref{s:FRS} mainly serves to provide some motivation for this result, as it is very important for the remainder. Indeed, below we will reproduce this result in the context of extended \textsc{cft}. We will define what an extended \textsc{cft} assigns to points and to intervals. Then we will take the two halves of a circle, and fuse them over the algebra associated to their boundary. Comparing the resulting Hilbert space with~\eqref{eqn:H_full} will provide a check of our formalism.

\subsubsection{Defects}\label{s:defects}

We mention one more feature of the FRS construction. Recall that two algebras $A$ and $B$ are said to be \textit{Morita equivalent} if there exist bimodules $_A X_B$ and $_B Y_A$ such that there are isomorphisms
	\[ {}_A X \otimes_B Y_A\cong {}_A A_A \quad\text{and}\quad {}_B Y \otimes_A X_B\cong {}_B B_B \ . \]

Now if we have two Frobenius algebra objects $Q$ and $Q'$ that are Morita equivalent (with the definition interpreted internally to the category~$\C$), the resulting full \textsc{cft} does not change. In particular, we get the same state space~\eqref{eqn:H_full}.

Kapustin and Saulina~\cite{KS2011} have a nice way of reinterpreting this fact. Recall the special property, axiom~(v) from Section~\ref{s:Frobenius}, of our Frobenius algebra. Figure~\ref{fig:ribbon special property} shows the special property in  terms of the directed ribbon graph on $\Sigma \times \{ \frac 12 \}$.

\begin{figure}[h]
	\begin{center}
	\includegraphics[scale=.11]{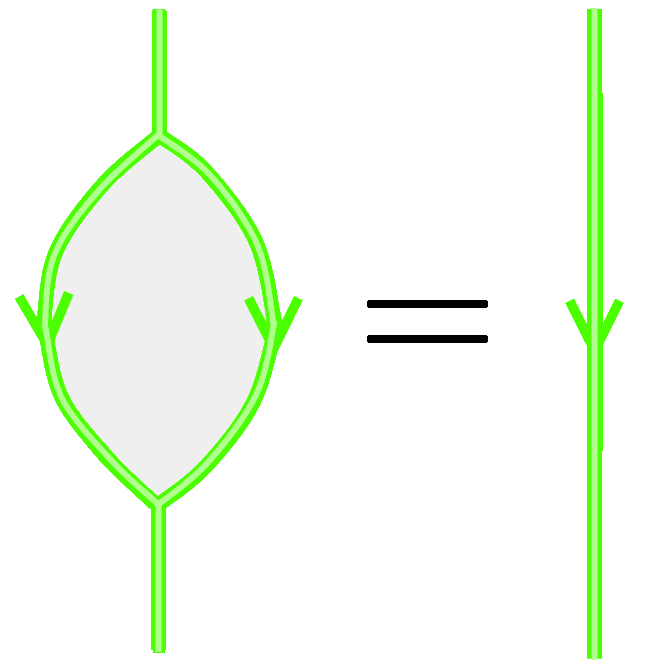}
	\end{center}\caption{The special property of Frobenius algebra objects allows us to fill in the holes in a ribbon graph.}\label{fig:ribbon special property}
\end{figure}

In other words: we can fill in the holes in the graph. If we do this everywhere, we get a three-manifold with an embedded surface, and the result looks like in Figure~\ref{fig:FRS filled}.

\begin{figure}[h]
	\begin{center}
	\includegraphics[scale=.11]{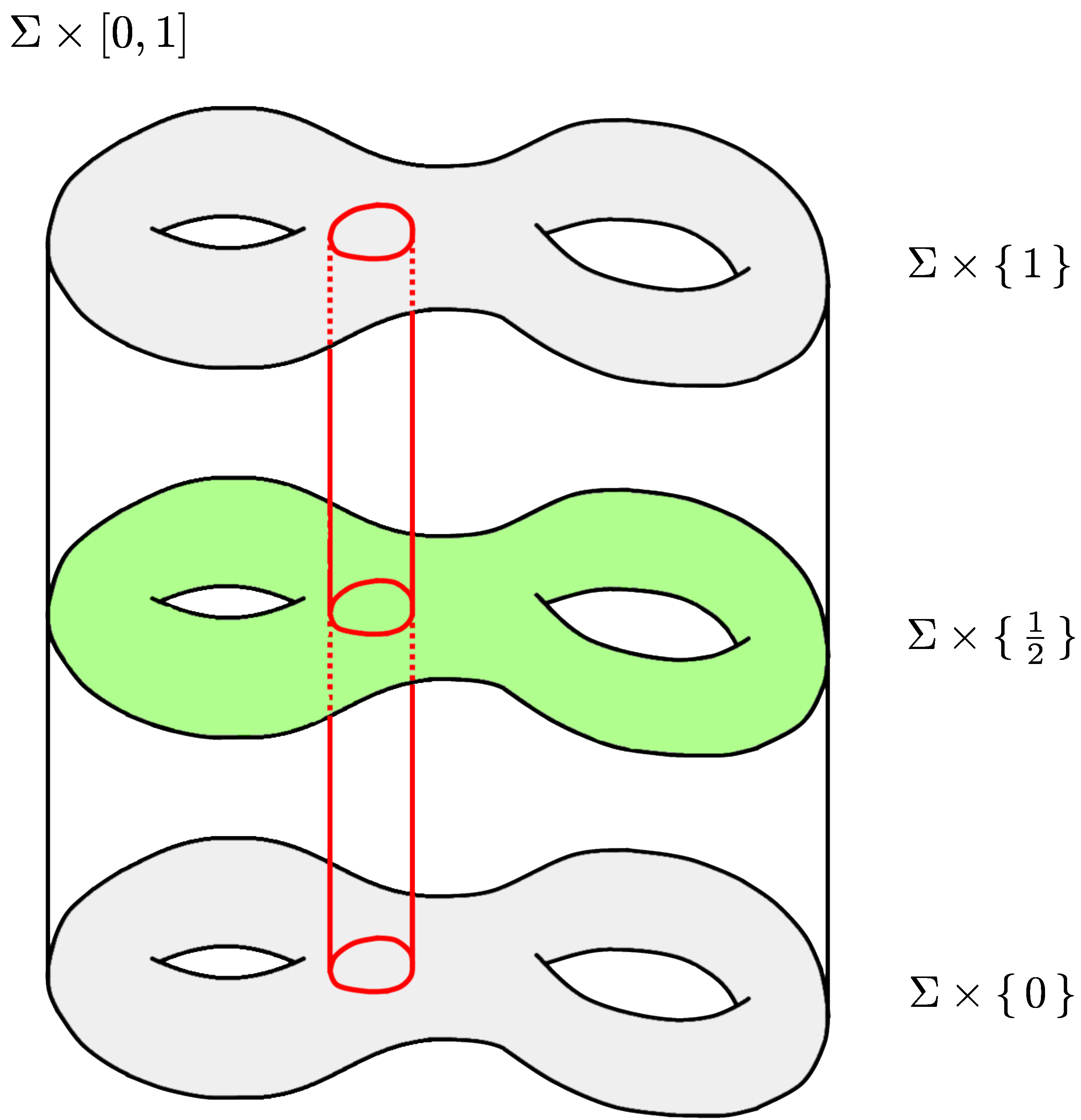}
	\end{center}\caption{The product of a closed Riemann surface $\Sigma$ with the unit interval. The holes in the ribbon graph on $\Sigma \times \{ \frac 12 \}$ have been filled using the special property of the Frobenius algebra.}\label{fig:FRS filled}
\end{figure}

The three-manifold $\Sigma \times [0,1]$ now is decorated by a codimension-one \textit{defect} (`surface operator'). According to \cite{KS2011} this defect only contains the information of the Morita equivalence class of~$Q$, and not of $Q$ itself. Moreover, one can go back to the ribbons and reinterpret them as actual embedded surfaces whose one-dimensional boundaries are labelled by $Q$ and whose two-dimensional interior corresponds to the defect. Upon filling in the holes in the ribbon graph we get rid of the boundary lines, we no longer see $Q$, but only its Morita equivalence class, in the form of a defect.

The partition function of $\Sigma$ is then obtained by evaluating the three-dimensional \textsc{tqft} on this three-manifold with defect.

\subsubsection{Defects between conformal nets}

The last ingredient we need in order to make sense of the three-manifold $\Sigma \times [0,1]$ with embedded surface within the formalism of conformal nets is the notion of a defect, leading to defects between conformal nets~\cite{BDH2009}.
For the purpose of the previous discussion, we would only need defects from a conformal net to itself. But in general, defects behave like bimodules: given two conformal nets $\A$ and $\B$, there is a notion of an $\A$-$\B$-defect ${}_\A D_\B$.

\begin{definitions}
A \textit{bicolored interval} is a contractible one-manifold~$I$ equipped with a decomposition $I = I' \cup I''$ that looks like one of
\smallskip

\begin{center}
\includegraphics[scale=.16]{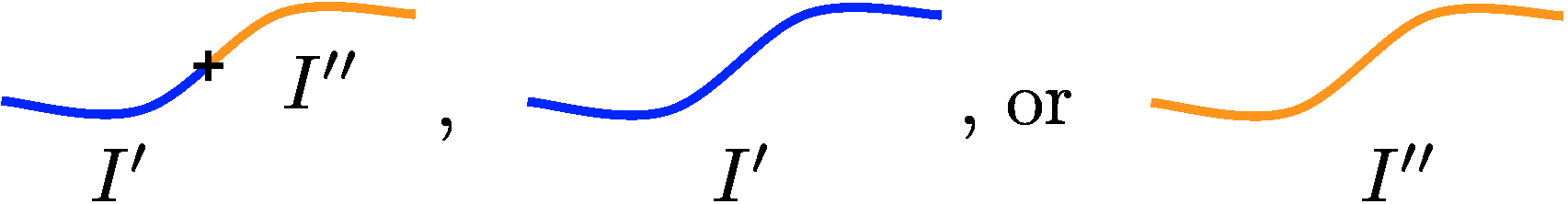}
\end{center}
\smallskip

\noindent along with a local coordinate at the color-changing point.

A \textit{defect} between conformal nets $\A$ and $\B$ is a functor from the category
	\[ \begin{cases} \ \text{objects: bicolored intervals;} \\ \ \text{morphisms: color preserving embeddings that respect the local coordinate;} \end{cases} \]
to the category
	\[ \begin{cases} \ \text{objects: von Neumann algebras;} \\ \ \text{morphisms: homomorphisms and antihomomorphisms;} \end{cases} \]
sending an embedding $I \hookrightarrow J$ to a homomorphism $D(I) \to D(J)$ if it preserves orientation, and to an antihomomorphism $D(I) \to D(J)^\text{op}$ if it reverses orientation. We have that $D(I) = \A(I)$ if $I''$ is empty, and $D(I) = \B(I)$ if $I'$ is empty. Moreover, $D$ satisfies axioms similar to those of conformal nets.
\end{definitions}

\section{Constructing extended conformal field theories}\label{s:construction}

Until this point we have mostly discussed the work of others. It is time to come back to extended \textsc{cft}. In this section we will partially construct an extended \textsc{cft} starting from a $\chi$\textsc{cft} that is given to us in the form of a conformal net~$\A$, and a Frobenius algebra object $Q\in\text{Rep}\,(\A)$.

Recall from Section~\ref{s:loop group nets} that a representation of $\A$ consists of a Hilbert space~$H$ equipped with compatible actions of $\A(I)$ for every $I \subsetneq S^1$. In Section~\ref{s:loop group example} we have seen how the monoidal structure on $\text{Rep}\,(\A)$ is defined: we identify the left half and the right half of $S^1$ with $[0,1]$ and set $A \coloneqq \A\big([0,1]\big)$. This provides a fully faithful embedding of $\text{Rep}\,(\A)$ into the category of $A$-$A$-bimodules, and the tensor product on $\text{Rep}\,(\A)$ is inherited from the monoidal structure on $A$-$A$-bimodules:
	\[ (H,K) \longmapsto H \boxtimes_A K \,. \]
We can therefore view the Hilbert space $Q$ as an $A$-$A$-bimodule.

\subsection{The algebra associated to a point}

We start with dimension zero. The algebra that is associated to a point can be defined in the world of $A$-$A$-bimodules:
\begin{equation}\label{eqn:alg assoc to point}
	B \coloneqq \text{Hom}\big( L^2 A_A,Q_A \big)
\end{equation}
This is the set of bounded linear maps that commute with the right action of~$A$.
The algebra \eqref{eqn:alg assoc to point} also appears in the work of Longo and Rehren~\cite{LR2004}; here we present a different construction of it.
The reason that this works is the following surprising fact.
\begin{lemma}\label{lem:alg assoc to point}
The vector space $B$ is an algebra, and indeed a von Neumann algebra. Moreover, there is an algebra homomorphism $A \longrightarrow B$.
\end{lemma}

\begin{proof}[Proof of Lemma~\ref{lem:alg assoc to point}.]
Let us sketch the proof. For convenience we abbreviate $1 \coloneqq L^2 A$ and write $\boxtimes$ instead of $\boxtimes_A$. Recall that the Frobenius algebra object $Q$ comes equipped with a multiplication $m\colon Q \boxtimes Q \longrightarrow Q$, a unit $\eta\colon 1 \longrightarrow A$ and comultiplication $\Delta = m^*$ and $\varepsilon = \eta^*$. We have to define a product, unit, and an involution on $B$, and show that it is a von Neumann algebra.

Let $f$ and $g$ be elements of $B$. The \textit{product} of $f$ and $g$ is defined as the composition
	\[ f\cdot g \colon 1 \xrightarrow{\ g \ } Q \cong 1 \boxtimes Q \xrightarrow{\ f\times 1\ } Q \boxtimes Q \xrightarrow{\ m\ } Q \ . \]
Figure~\ref{fig:B product} shows how this rule can be represented graphically. Using the diagrams it is clear that the product is associative.

\begin{figure}[h]
	\begin{center}
	\includegraphics[scale=.15]{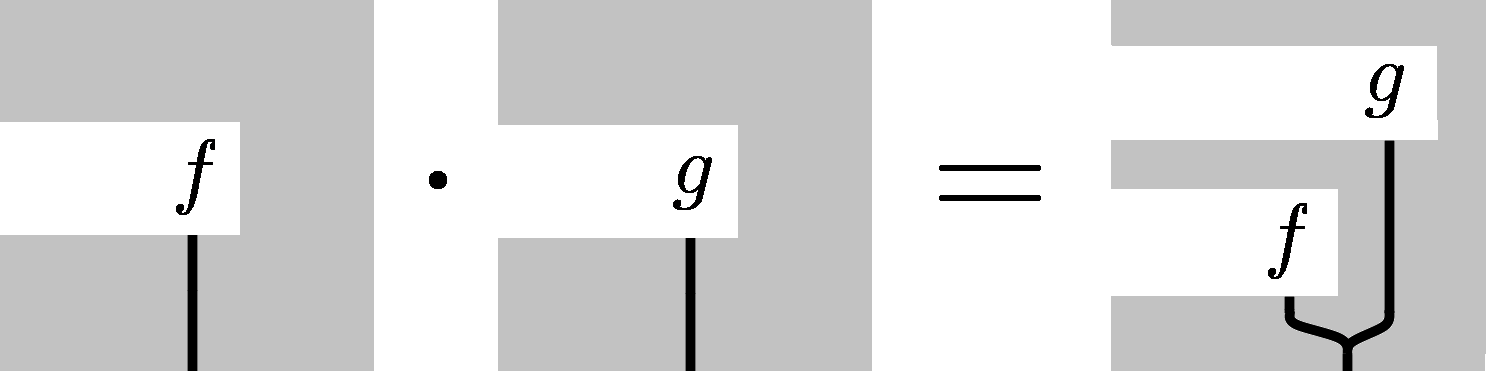}
	\end{center}\caption{Diagrammatic representation of the product of $f,g\in B$. It should be understood as follows. Consider the diagram on the left, representing~$f$. As with the string diagrams in Section~\ref{s:Frobenius} we start at the top, which is empty, corresponding to the unit. Then we apply $f$, which is only linear with respect to the right action of~$A$, so that the left-hand side is `blocked'. In this way, the diagram exactly shows which operations are allowed algebraically. Finally, the line going to the bottom represents a copy of~$Q$.}\label{fig:B product}
\end{figure}

The \textit{unit} of $B$ is just the unit map $\eta \colon 1 \longrightarrow Q$ as shown in Figure~\ref{fig:B unit}. Together with the above product this determines the algebra structure on $B$.

\begin{figure}[h]
	\begin{center}
	\includegraphics[scale=.15]{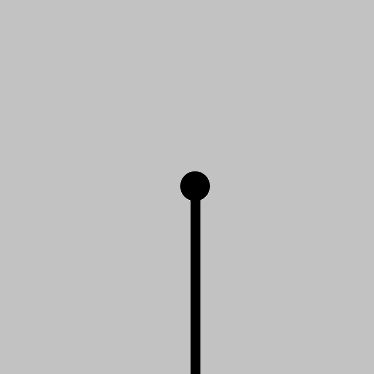}
	\end{center}\caption{Diagrammatic representation of the unit map on~$B$ is the same as in the string diagrams for Frobenius algebra objects (cf.~Figure~\ref{fig:Frobenius alg obj}).}\label{fig:B unit}
\end{figure}

Next, the \textit{involution} is denoted by $^\star$ and is defined as the following composition
	\[ f^\star \colon 1 \xrightarrow{\ \eta\ } Q \xrightarrow{\ \Delta\ } Q \boxtimes Q \xrightarrow{\ f^* \times 1\ } 1 \boxtimes Q \cong Q \ . \]
See Figure~\ref{fig:B involution} for the corresponding diagram.

\begin{figure}[h]
	\begin{center}
	\includegraphics[scale=.15]{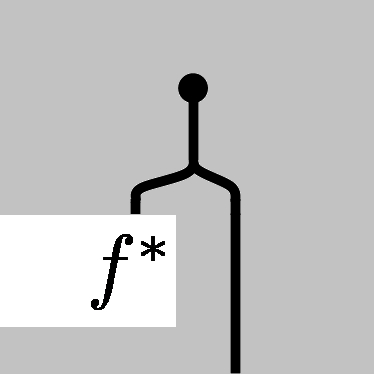}
	\end{center}\caption{Diagrammatic representation of the involution of $f \in B$.}\label{fig:B involution}
\end{figure}

There is also a map from $A$ to $B$, sending an element $a \in A$ to the composition of left multiplication by $a$ with the unit:
	\[ a \colon 1 \xrightarrow{\ a\cdot\ } 1 \xrightarrow{\ \eta\ } Q \ . \]
This can be represented as shown in Figure~\ref{fig:B contains A}.

\begin{figure}[h]
	\begin{center}
	\includegraphics[scale=.15]{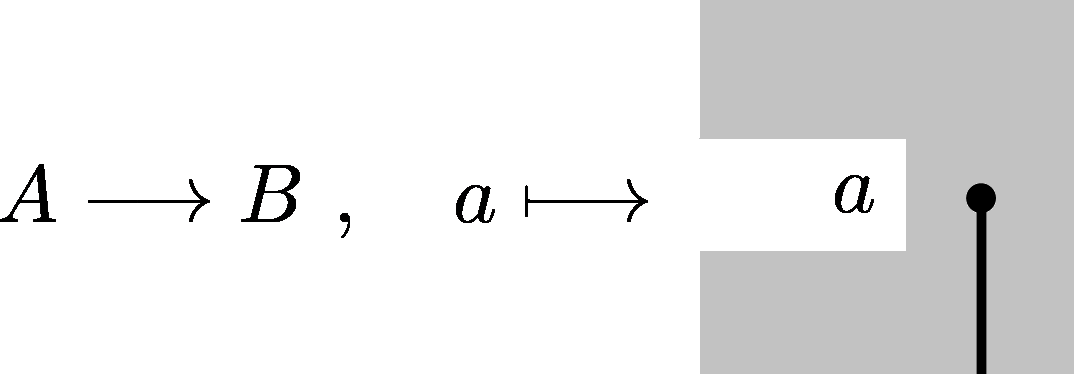}
	\end{center}\caption{Diagrammatic representation of the map $A \longrightarrow B$.}\label{fig:B contains A}
\end{figure}

Let $B(H)$ denote the set of bounded operators on the underlying Hilbert space $H$ of~$Q$. The algebra $B$ acts on $H$ via
\begin{align*}
	B & \longrightarrow B(H) \ , \\
	f \mspace{1mu} & \longmapsto \big[ \ Q \cong 1 \boxtimes Q \xrightarrow{\ f \times 1 \ } Q \boxtimes Q \xrightarrow{\ m \ } Q \ \big] \ .
\end{align*}
The image of $f\in B$ is shown in Figure~\ref{fig:left action on Q}.

\begin{figure}[h]
	\begin{center}
	\includegraphics[scale=.15]{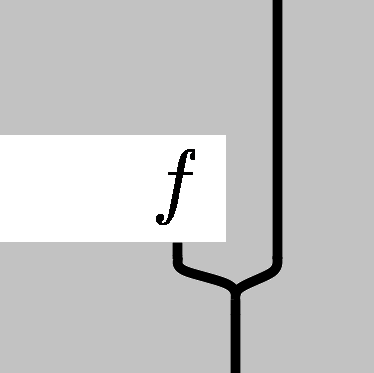}
	\end{center}\caption{Diagrammatic representation of the left action of $B$ on the Frobenius algebra object~$Q$.}\label{fig:left action on Q}
\end{figure}

Actually, $Q$ is a $B$-$B$-bimodule. The right $B$-action is shown in Figure~\ref{fig:right action on Q}. It uses the fact that $Q$ is its own dual (the pairing $\varepsilon\circ m$ is nondegenerate) and that for von Neumann bimodules there is a canonical identification between the dual and the complex conjugate. Therefore we can take the complex conjugate $\bar f$ of $f$ to get a left $A$-linear map.

\begin{figure}[h]
	\begin{center}
	\includegraphics[scale=.15]{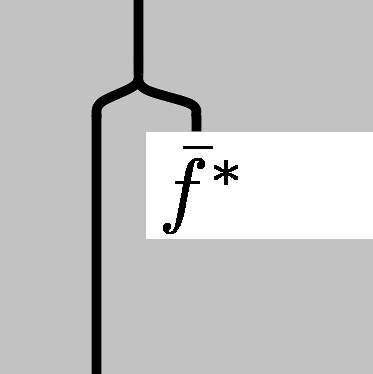}
	\end{center}\caption{Diagrammatic representation of the left action of $B$ on the Frobenius algebra object~$Q$.}\label{fig:right action on Q}
\end{figure}

Finally, one can show that the commutant of the left action of $B$ on $Q$ is the right action of $B$ on $Q$, and vice versa.
The algebra $B$ is its own bicommutant, and therefore a von Neumann algebra.
\end{proof}

One can also check that the Hilbert space $Q$ is canonically isomorphic to $L^2 B$ as a $B$-$B$-bimodule.
In order to show that, one has to construct a positive cone $P\subset Q$ (which corresponds to $L_+^2B$), and
define the modular conjugation $J:Q\longrightarrow Q$.
Those should then satisfy the axioms listed in \cite{Ha1975}.
The cone is defined as $P:=\{b\,\xi\, b^*\,|\,b\in B,\xi\in L^2_+A\}$.
To construct the modular conjugation, one uses the identification ${}_AQ_A\cong {}_AQ^\vee_A$
coming from the pairing $\tikzmath[scale=.8]{\fill[gray!45] (-.5,.65) rectangle (.5,-.1); \fill (0,.1) circle (.055); \draw[line width = 1.1, rounded corners = 1] (0,.1) -- (0,.3) -- (.2,.4) -- (.2,.65) (0,.28) -- (0,.3) -- (-.2,.4) -- (-.2,.65);}$\,,
along with the fact that the dual of a bimodule is always its complex conjugate.
We can then define $J$ to be the composite isomorphism $Q\,\cong\, Q^\vee \,\cong\, \overline Q$.

Recall that the zero-manifolds in the source bicategory of our three-tier \textsc{cft} are generated by two local models: a point with a sign. If $B_+$ is the von Neumann algebra~\eqref{eqn:alg assoc to point} associated to the point with positive orientation, and $B_-$ the von Neumann algebra associated to the point with negative orientation, then $B_+$ is canonically isomorphic to $B_-^\text{op}$.

One can reinterpret the above construction as that of a defect from $\A$ to $\A$. Namely, there exists a defect~$D$, constructed from the Frobenius algebra object~$Q$, such that $D\big([0,1]\big)=B$.
\begin{figure}[h]
	\begin{center}
	\includegraphics[scale=.11]{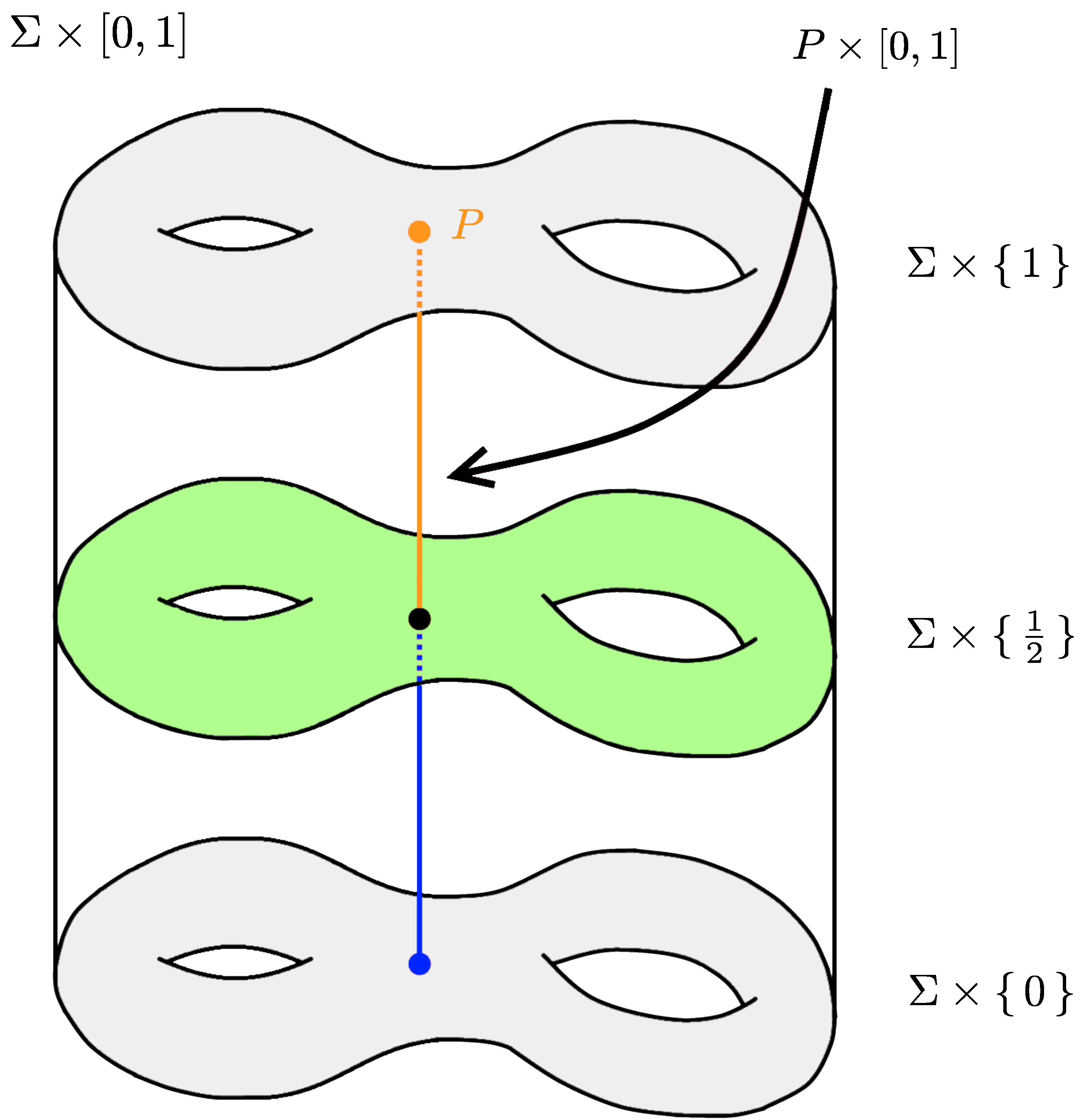}
	\end{center}\caption{The construction of the algebra $B$ associated to a point $P\in\Sigma$ in the context of the FRS/Kapustin-Saulina construction.}\label{fig:algebra associated to point}
\end{figure}
Figure~\ref{fig:algebra associated to point} shows the corresponding defect in the FRS construction. This is a rather special kind of defect, where the precise location in $[0,1]$ where the colors change is actually not important: the only thing that matters is that the interval $[0,1]$ is genuinely bicolored. Such defects are called \textit{topological defects}. The defect that appeared in Section~\ref{s:defects} is also a topological defect: what the \textsc{tqft} assigns to a manifold does not change at all when the location of the defect is moved a bit upwards or downwards.

\subsection{The bimodule associated to an interval}

Points do not have any geometry, and indeed the discussion above was very algebraic. Next we have to decide what to associate to an interval; this will involve some geometry.

We have already seen that, in order to evaluate our extended \textsc{cft} on a point~$P$, we have to form the product $P\times[0,1]$, and evaluate our defect~$D$ on the resulting one-manifold. In the present case we start with an interval $I$. We are again supposed to cross with $[0,1]$ and do something involving the defect, or, equivalently, with the Frobenius algebra object~$Q$ --- see Figure~\ref{fig:algebra associated to interval}.

\begin{figure}[h]
	\begin{center}
	\includegraphics[scale=.11]{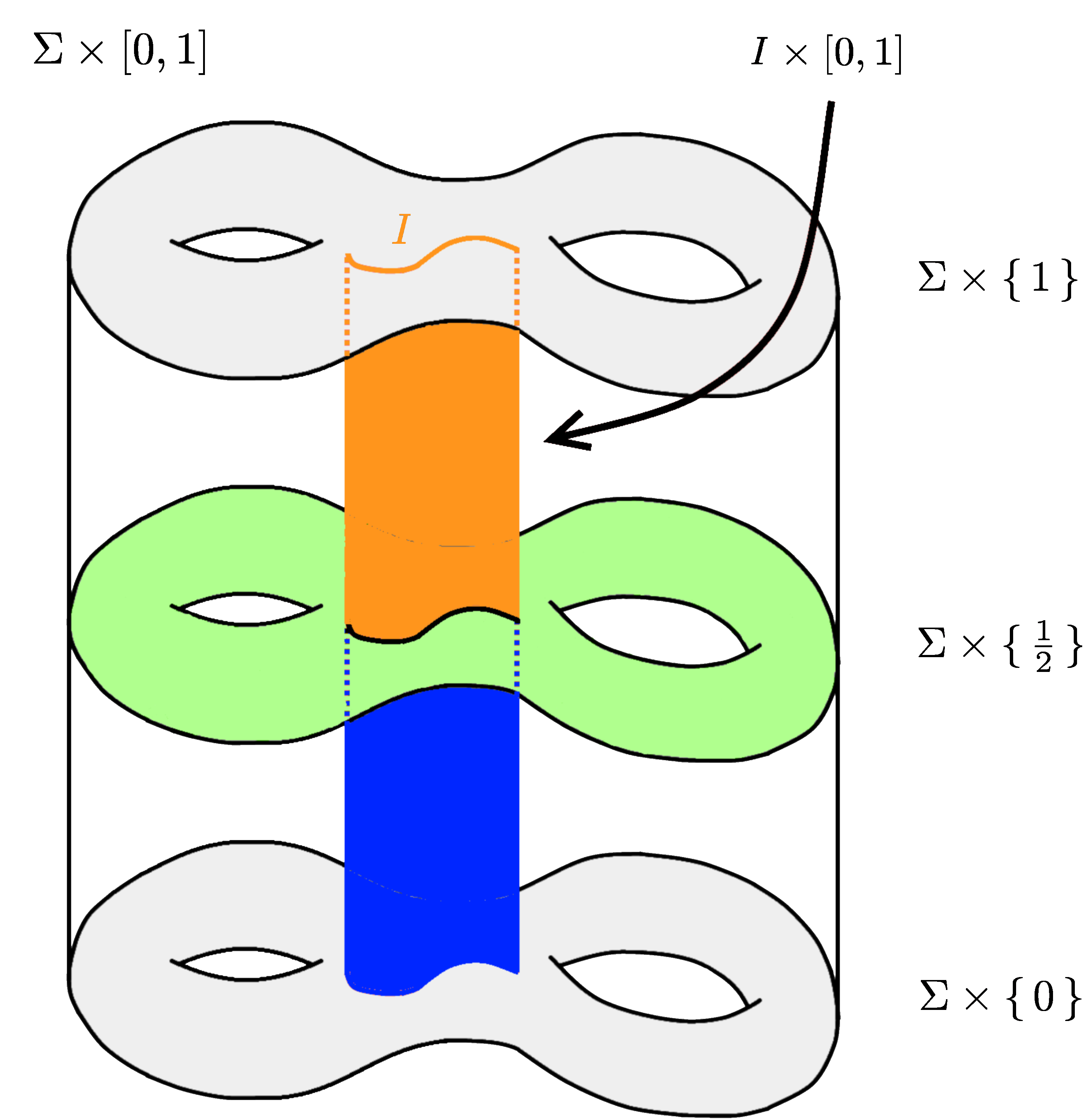}
	\end{center}\caption{The construction of the Hilbert space $Q(I)$ associated to an interval $I\subset\Sigma$ in the context of the FRS/Kapustin-Saulina construction.}\label{fig:algebra associated to interval}
\end{figure}

Since we have collars at the ends of our interval, we can smooth out the rectangle $\partial\big(I\times[0,1]\big)$ to a circle. We will see that the extended \textsc{cft} assigns to $I$ a version of~$Q$ modelled on the boundary $\partial\big(I\times[0,1]\big)$: this is a Hilbert space that looks like $Q$, but which has actions of $\A(J)$ for every $J \subsetneq \partial\big(I\times[0,1]\big)$, as opposed to $J \subsetneq S^1$.

\subsubsection{Intermezzo: representations of conformal nets revisited}\label{s:representations}

Before we proceed it is useful to look at a coordinate-independent approach to the representation theory of conformal nets. Consider a circle~$S$: a manifold that is diffeomorphic to~$S^1$, but \textsl{without} a choice of such a diffeomorphism. Let $ \text{Rep}_S(\A) $ denote the category whose objects are Hilbert spaces equipped with compatible actions of $\A(J)$ for every $J \subsetneq S$. $\text{Rep}\,(\A)$ is the special case in which $S$ is the unit circle.

Clearly, $\text{Rep}_S(\A)$ is equivalent to $\text{Rep}\,(\A)$, but there is no canonical way of picking such an equivalence. Also, unlike with $\text{Rep}\,(\A)$, there is no canonical monoidal structure on $\text{Rep}_S(\A)$. Instead we have an `external product'. Given three circles $S_1$, $S_2$ and $S_3$ with compatible smooth structures\footnote{This is a technical definition that we will not explain here, see (1.29) in \cite{BDH2013}.} as in Figure~\ref{fig:external product}, there is a canonical functor
	\[ \text{Rep}_{S_1}(\A) \times \text{Rep}_{S_2}(\A) \longrightarrow \text{Rep}_{S_3}(\A) \ . \]

\begin{figure}[h]
	\begin{center}
	\includegraphics[scale=.15]{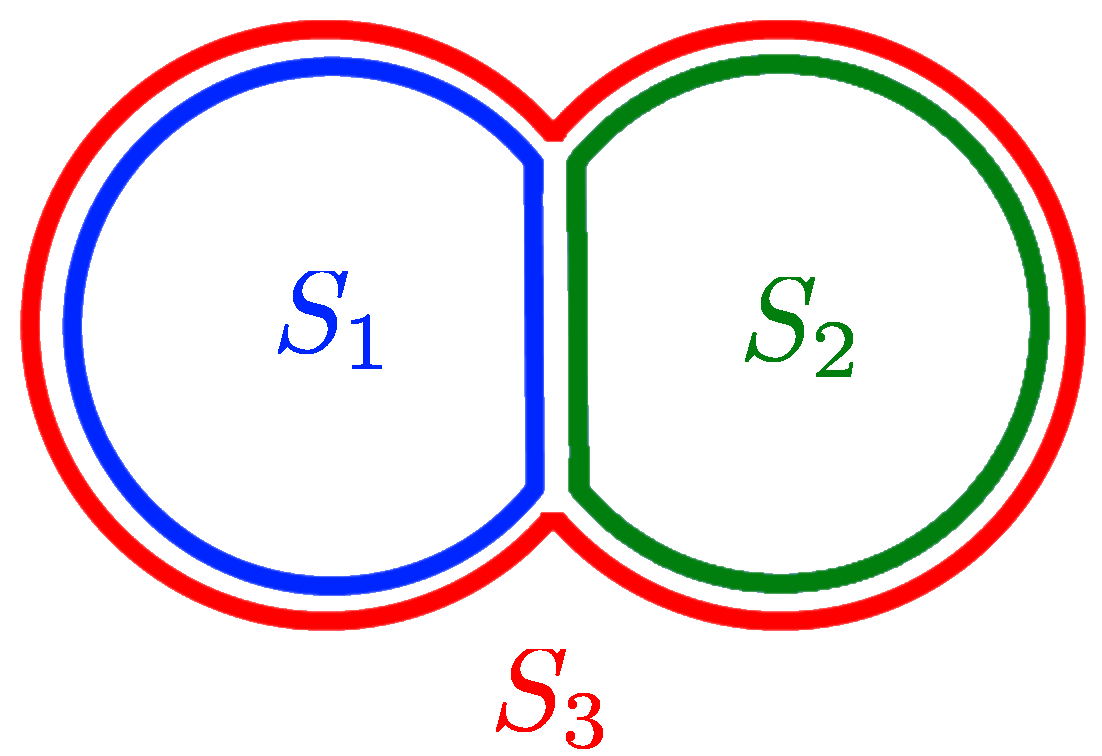}
	\end{center}\caption{The external product in $\text{Rep}_S(\A)$ is defined for each triple of circles with compatible smooth structures.}\label{fig:external product}
\end{figure}

Now, although there is no canonical equivalence between $\text{Rep}_S(\A)$ and $\text{Rep}\,(\A)$, we can nevertheless attempt to construct a functor
	\[ \text{Rep}\,(\A) \longrightarrow\ \text{Rep}_S(\A) \ , \] 
and see where we fail. Of course, we could just pick a diffeomorphism $S \longrightarrow S^1$, but that is clearly non-canonical. Let us try the following
\begin{equation}\label{eqn:attempt}
	H \,\longmapsto\, H \underset{\text{Diff}\,(S^1)}{\times} \text{Diff}\,(S, \, S^1) \ ,
\end{equation}
where $\text{Diff}\,(S, \, S^1)$ is the set of all diffeomorphisms from $S$ to $S^1$, equipped with its natural left action of $\text{Diff}\,(S^1)$.
The reason why there is an action of $\text{Diff}\,(S^1)$ on $H$ is that whenever a diffeomorphism is supported in a small interval, the corresponding automorphism of $\A(J)$ is \textsl{inner}. Thus, there is an element of that algebra associated to the diffeomorphism, which, in turn, acts on~$H$. Those local diffeomorphisms generate $\text{Diff}\,(S^1)$, and so we get an action of $\text{Diff}\,(S^1)$.

The reason that \eqref{eqn:attempt} does not quite work is that the choice of algebra elements implementing the given inner automorphism is not unique. Indeed, it is only defined up to phase, and therefore the action of $\text{Diff}\,(S^1)$ on $H$ is only a projective action.

\subsubsection{Back to business}

We would like to say that the value of the extended full \textsc{cft} on $I$ is the image of $Q \in \text{Rep}\,(\A)$ under the functor
	\[ \text{Rep}\,(\A) \longrightarrow \text{Rep}_{\,\partial(I\times [0,1])}\,(\A) \ . \]
But, as we have seen, at least at first sight, that does not seem to work. The reason that this nevertheless \textsl{does} work is that $\partial\big(I \times[0,1]\big)$ has more structure than an arbitrary circle~$S$: it has an involution $(x,t)\longmapsto(x,-t)$. Therefore it makes sense to talk about \textit{symmetric diffeomorphisms}, i.e., those diffeomorphisms that commute with the involution.

Thus, for $S = \partial\big(I \times[0,1]\big)$, we can replace \eqref{eqn:attempt} by 
\begin{align*}
	\text{Rep}\,(\A) \longrightarrow\ \text{Rep}_S(\A) \ , \quad
	H \longmapsto H \underset{\text{Diff}_\text{sym}(S^1)}{\times} \text{Diff}_\text{sym}(S, \, S^1) \ .
\end{align*}
Now something very nice happens: the universal central extension of $\text{Diff}\,(S^1)$ splits over $\text{Diff}_\text{sym}(S^1)$, and so that group now \textsl{does} act on $H$, and the formula makes sense. Therefore we can define
\begin{equation}\label{eqn:Q(I)}
	\text{Rep}\,(\A) \longrightarrow \text{Rep}_{\,\partial(I\times [0,1])}\,(\A) \ , \quad Q \longmapsto Q(I) \ .
\end{equation}

\begin{figure}[h]
	\begin{center}
	\includegraphics[scale=.15]{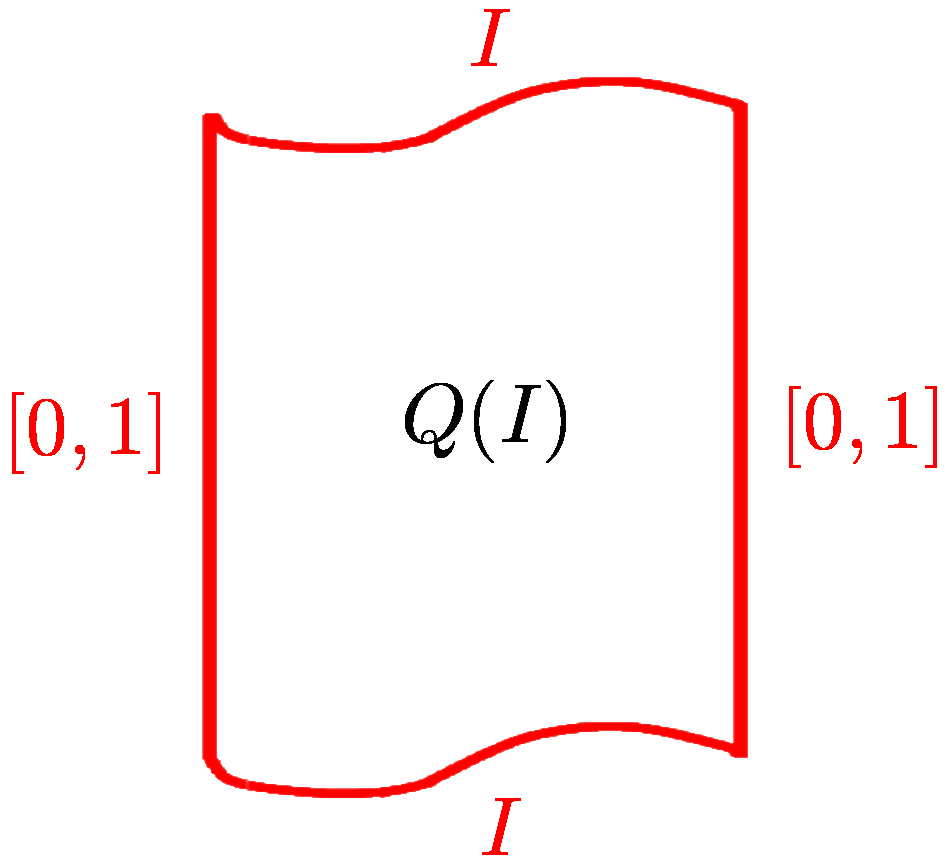}
	\end{center}\caption{The Hilbert space $Q(I)$ has actions of the copies of $B$ associated to the unit intervals $[0,1]$ on the left and on the right.}\label{fig:Q(I)}
\end{figure}

\noindent 
To see that $Q(I)$ is indeed a $B$-$B$-bimodule, notice that $A=\A\big([0,1]\big)$ has two actions on $Q(I)$, corresponding to the two copies of $[0,1]$ in the boundary of $I \times[0,1]$ (cf.~Figure~\ref{fig:Q(I)}). 
If we identify the boundary $\partial\big(I \times[0,1]\big)$ with the unit circle via some symmetric diffeomorphism that sends the corners on the left to the `north' and `south pole' of the circle as illustrated in Figure~\ref{fig:Q(I) left action}, this identifies the left action of~$A$ of $Q(I)$ with the standard left action of $A$ on $Q$. Now recall from the proof of Lemma~\ref{lem:alg assoc to point} that the we have an inclusion $A\subset B$ and that the left action of $A$ on $Q$  extends to an action of~$B$ on $Q$ in a canonical way.
Therefore, the left action of~$A$ of $Q(I)$ extends to an action of $B$.

\begin{figure}[h]
	\begin{center}
	\includegraphics[scale=.10]{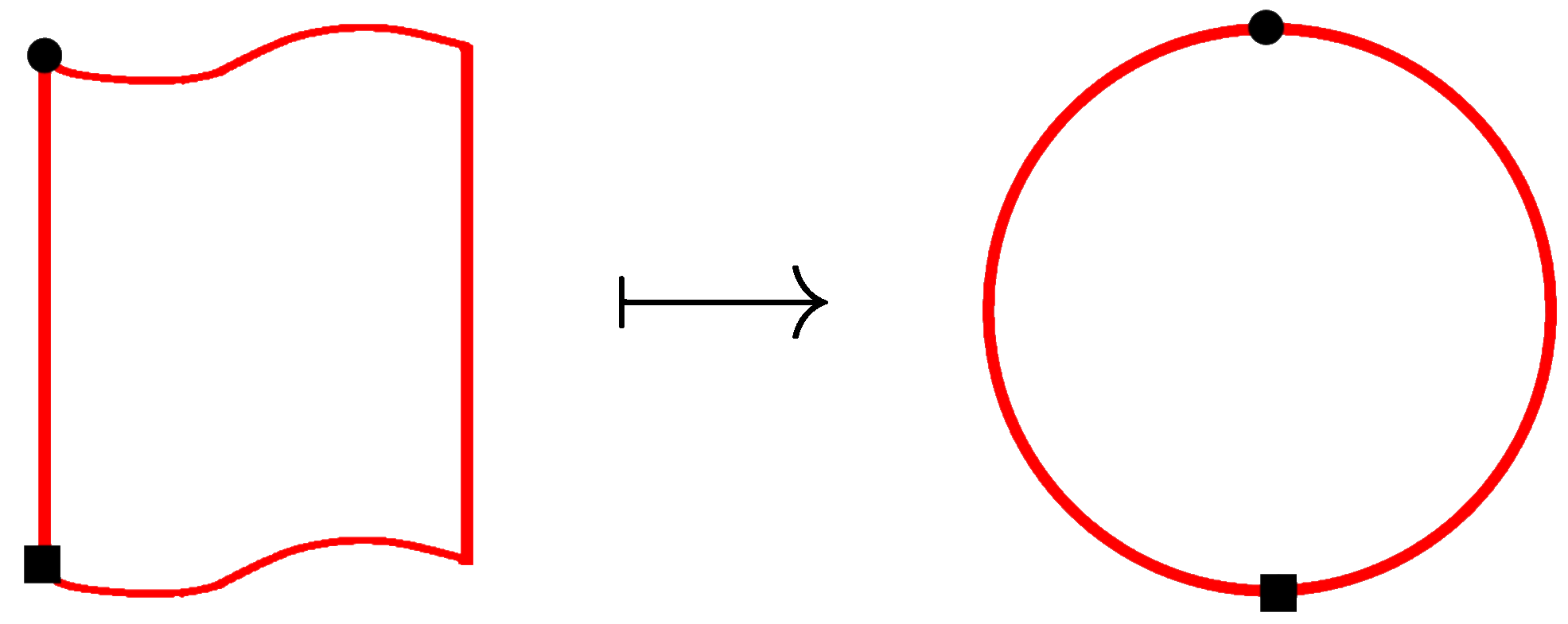}
	\end{center}\caption{An identification of $\partial\big(I \times[0,1]\big)$ and a standard circle via a symmetric diffeomorphism mapping the corners on the left to the north and south poles of the circle.}\label{fig:Q(I) left action}
\end{figure}

\noindent Similarly, with the use of a symmetric diffeomorphism as indicated in Figure~\ref{fig:Q(I) right action}, we can identify the right action of~$\A\big([0,1]\big)$ on~$Q(I)$ 
with the standard right action of $A$ on $Q$, which likewise extends to an action of~$B$.

\begin{figure}[h]
	\begin{center}
	\includegraphics[scale=.10]{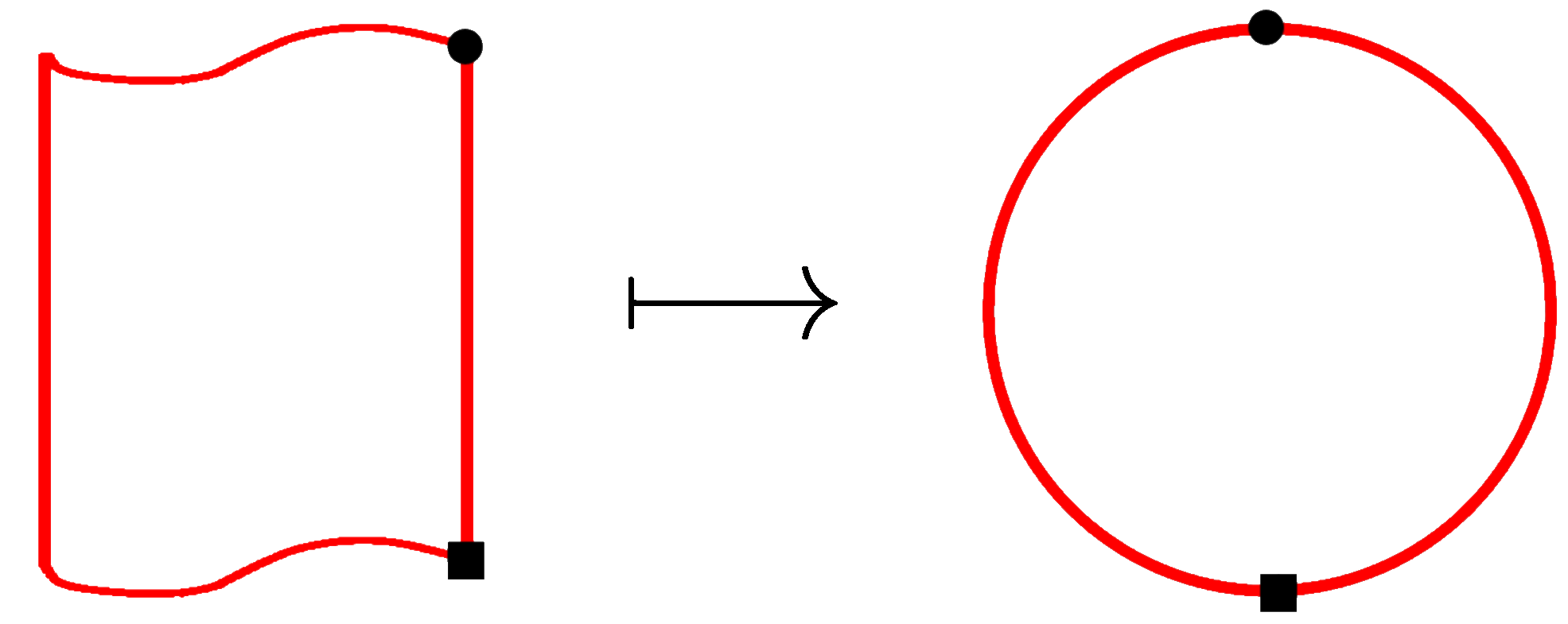}
	\end{center}\caption{Another identification of $\partial\big(I \times[0,1]\big)$ and a standard circle, now via a symmetric diffeomorphism sending the corners on the right to the poles of the circle.}\label{fig:Q(I) right action}
\end{figure}

At this point it is not too difficult to see, using the fact that $Q \cong L^2 B$, that the assignment~\eqref{eqn:Q(I)} is compatible with glueing:
	\[ {}_B Q(I_1) \, \boxtimes_B Q(I_2)_B = {}_B Q(I_1 \cup I_2)_B \ . \]
This is of course necessary for our construction to make sense, but it is not very impressive. Let us turn to something more surprising.

\subsection{Recovering the state space from the FRS construction}

Recall from Section~\ref{s:state space} that the state space~\eqref{eqn:H_full} of the full \textsc{cft} from the FRS construction is given by
	\[ H_\text{full} \coloneqq \bigoplus_{\mu,\lambda} \text{Hom}_{Q,Q}\big(\lambda \boxtimes^+ Q \boxtimes^- \mu^\vee,Q\big) \otimes H_\lambda \otimes \overline{H_\mu} \ . \]

In this section we will show, or at least sketch, how this result can be reproduced with our construction. The idea is to take the unit circle, cut it in half, and fuse the corresponding algebras over $B \, \bar\otimes \, B^\textup{op}$. More precisely, we have the following

\begin{theorem}\label{thm:H_full}
Decompose the unit circle as $S^1 = I\cup J$ such that the intersection $I\cap J$ consists of two points only. Then the fusion of $Q(I)$ with $Q(J)$ over $B \, \bar\otimes \, B^\textup{op}$ (see Figure~\ref{fig:fusion thm}) is canonically isomorphic to $H_\textup{full}$ as a module over the chiral and antichiral algebras.
\end{theorem}

\begin{figure}[h]
	\begin{center}
	\includegraphics[scale=.14]{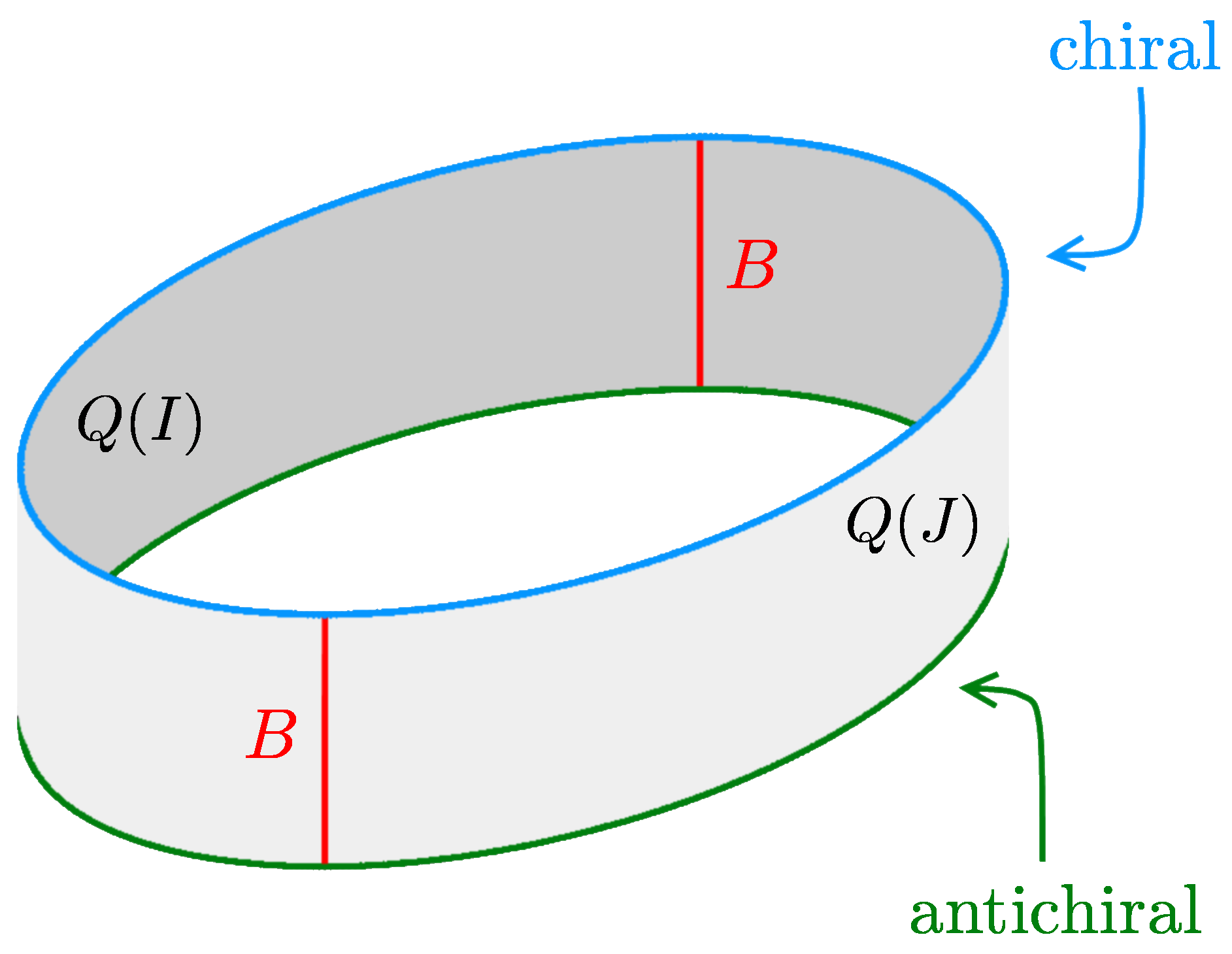}
	\end{center}\caption{The circle on the top corresponds to the representation of the chiral algebra on~$H_\lambda$, whilst the circle on the bottom corresponds to the representation of the antichiral algebra on~$\overline{H_\mu}$. We have used the action of the two copies of the bigger algebra $B$ (recall that $A \subset B$) to fuse $Q(I)$ with $Q(J)$.}\label{fig:fusion thm}
\end{figure}

For the proof of this theorem we need the following lemma:

\begin{lemma}
Let $Q$ and $B$ be as in \eqref{eqn:alg assoc to point}, and let
$H$ be a module over~$A$. Then $H$ is a $Q$-module (i.e.~we have a map $Q\boxtimes_A H \longrightarrow H$ satisfying the obvious axioms), if and only if it is a $B$-module extending the action of~$A$ (i.e.~we have a map $B\otimes H \longrightarrow H$ satisfying the obvious axioms).

Similarly, a homomorphism $H_1\to H_2$ is $Q$-linear if{f} it is $B$-linear.
\end{lemma}

With the help of this lemma, Theorem~\ref{thm:H_full} can be proved as sketched --- literally --- in Figure~\ref{fig:fusion proof}.

\begin{figure}[h]
	\begin{center}
	\includegraphics[scale=.15]{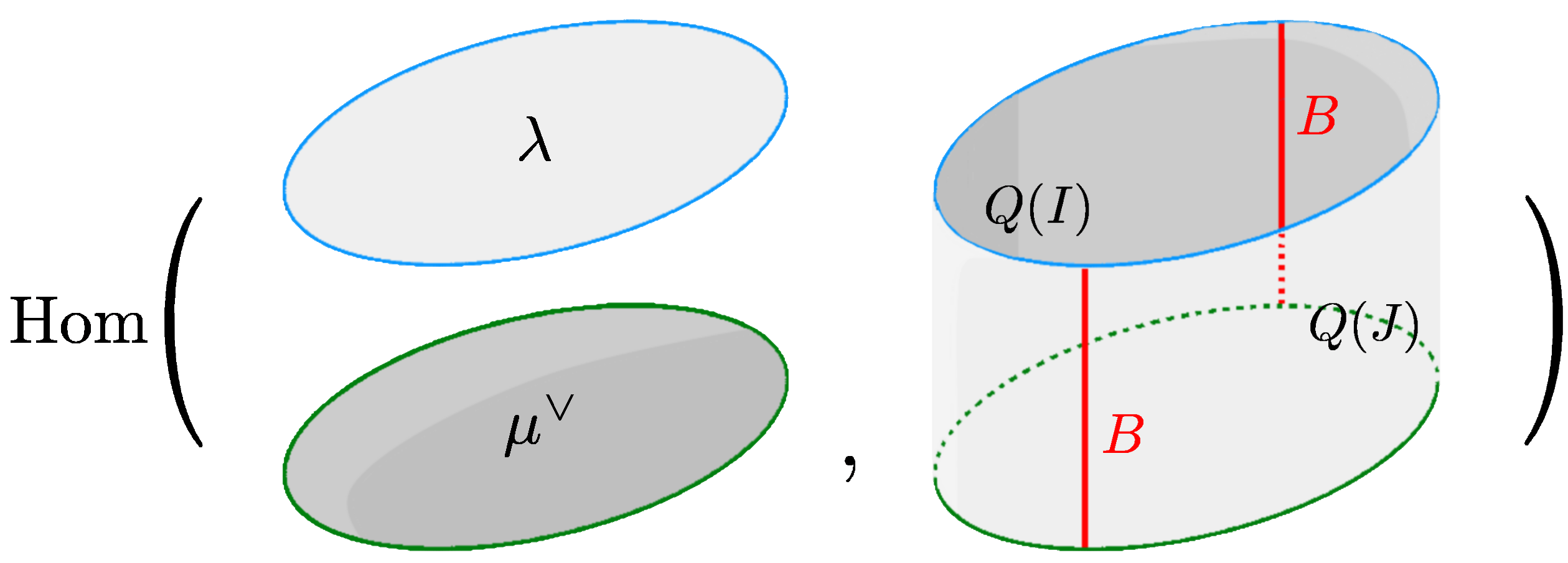}\\
	\includegraphics[scale=.15]{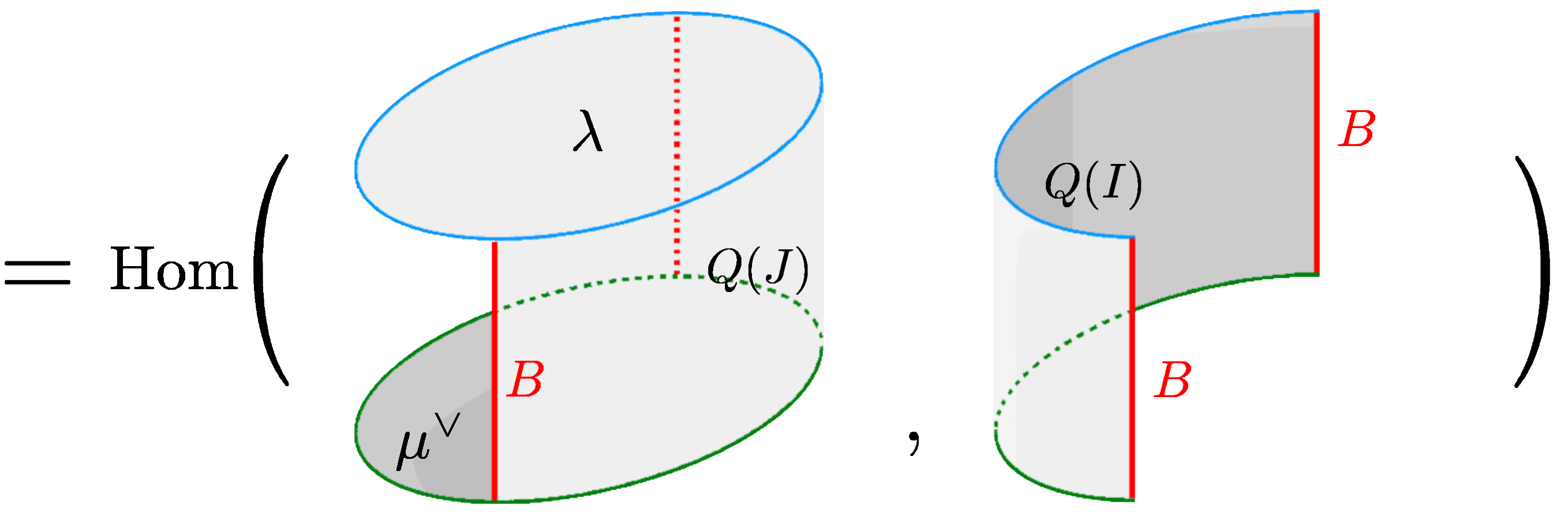}\\
	\includegraphics[scale=.15]{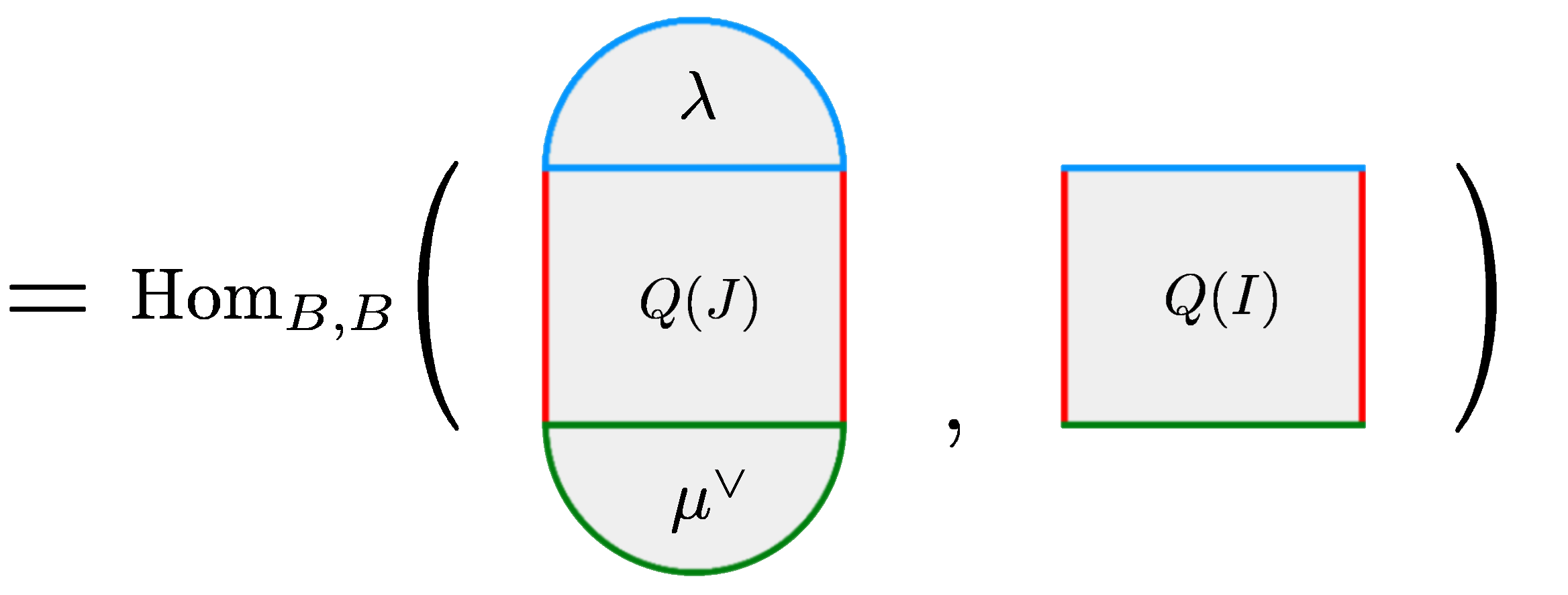}
	\end{center}
	\caption{\textit{Proof of Theorem~\ref{thm:H_full}.} Depict $\text{Hom}_{Q,Q} \big(H_\lambda \otimes \overline{H_\mu} \, , \, Q(I) \boxtimes Q(J)\big)$ as shown on the top. By duality, this is equal to the second line. Now we can flatten the shapes to get to the third line, which corresponds precisely to $\text{Hom}_{Q,Q}\big(\lambda \boxtimes^+ Q \boxtimes^- \mu^\vee,Q\big)$. \hfill $\square$ }\label{fig:fusion proof}
\end{figure}

\subsection{The maps associated to surfaces}

Starting from a $\chi$\textsc{cft} and a Frobenius algebra object we have constructed the extended \textsc{cft} corresponding to zero- and one-dimensional manifolds in the source bicategory. To conclude, we mention what happens to two-dimensional surfaces, and show what the open problem is that has to be solved in order to complete our construction.

\subsubsection{Discs and surfaces with cusps}

It is not too hard to see which bimodule map is associated to a disc with conformal structure.
We can view the disc as a cobordism from the empty one-manifold to the bounding circle. 
Thus, we have to construct a map from $\mathbb{C}$ to the Hilbert space $H_\text{full}$ associated to that circle.
This is the same as a choice of vector in that Hilbert space.
Moreover, this vector should be invariant under the group $PSL_2(\mathbb{R})$ of M\"obius transformations of the circle.
The \textit{vacuum vector} in $H_\text{full}$ is given by $\Omega\otimes\Omega$ in the direct summand $H_0 \otimes H_0$ of $H_\text{full}$. Here, $\Omega$ is (also) called the vacuum vector in $H_0$, and $H_0$ is the \textit{vacuum module} of the conformal net (the unit object in the category $\text{Rep}\,(\A)$). In both cases --- i.e.~in the case $\Omega\otimes \Omega \in H_\text{full}$, and also in the case $\Omega\in H_0$ --- the vacuum vector is the unique $PSL_2(\mathbb{R})$-fixed point up to scalars.

We also have a construction for the bimodule map associated to a surface with two cusps. After a choice of parametrization of the ingoing and outgoing boundaries by the unit interval~$[0,1]$, the semigroup $\text{Bigons}\big([0,1]\big)$ of bigons (as in \eqref{eqn:local model for Sigma}) of the unit interval can be identified with the complexification of the group of those diffeomorphisms of~$[0,1]$ that leave a neighbourhood of the endpoints fixed.

By extending the action of $\text{Diff}\big([0,1]\big)$ in a $\mathbb{C}$-linear fashion to the copy of $\text{Bigons}\big([0,1]\big)$ in the chiral sector, and $\mathbb{C}$-antilinearly to the copy of $\text{Bigons}\big([0,1]\big)$ in the antichiral sector, we get the desired actions of $\text{Bigons}\big([0,1]\big)$.

\subsubsection{Open problem: ninja stars}

The main open problem is the construction of the bimodule map associated to the `ninja star' depicted in  Figure~\vref{fig:ninja star}. We also have to prove a few basic properties of this map, together with one important relation that is shown in Figure~\ref{fig:important relation}.

\begin{figure}[h]
	\begin{center}
	\includegraphics[scale=.20]{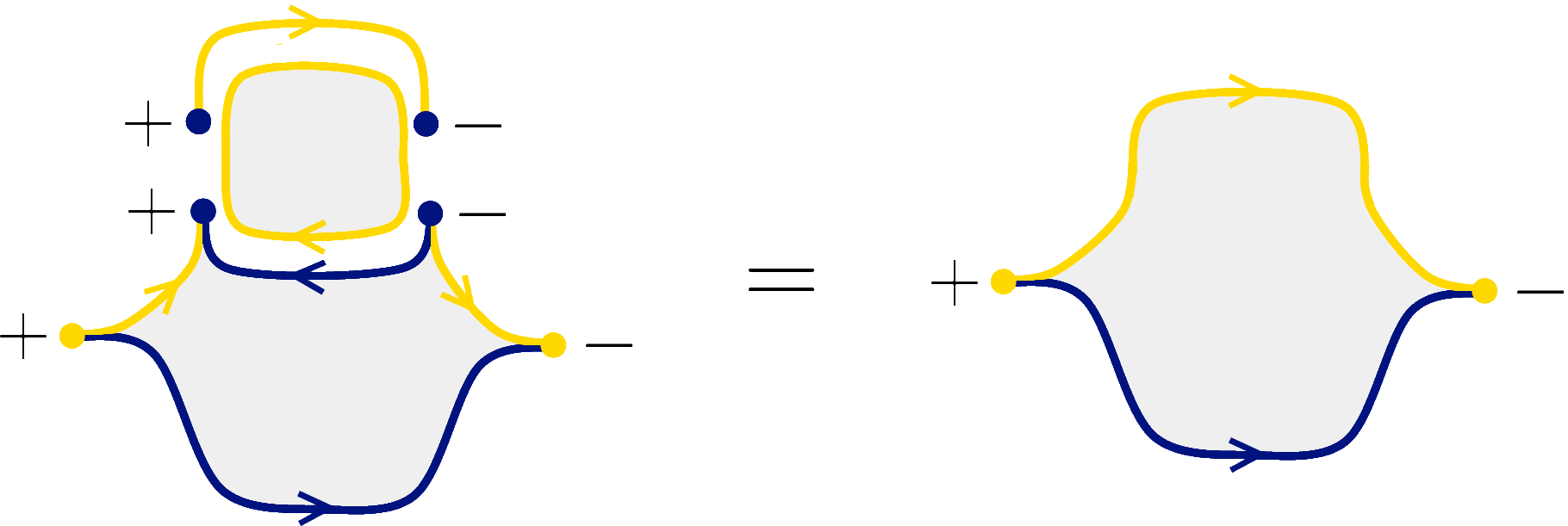}
	\end{center}\caption{An important relation.}\label{fig:important relation}
\end{figure}

This relation ensures the compatibility between the bimodule map associated to the ninja star (that we want to construct) and the parts of the extended \textsc{cft} that we have already constructed. More precisely, it means the following.

Given the 2-morphism from Figure~\ref{fig:important relation star} we can form the horizontal composition with the identity 2-morphism on the 1-morphism in Figure~\ref{fig:important relation cap} to get the result in Figure~\ref{fig:important relation horizontal composition}. 

\begin{figure}[h]
	\begin{center}
	\begin{subfigure}[b]{0.3\textwidth}
		\centering\includegraphics[scale=.15]{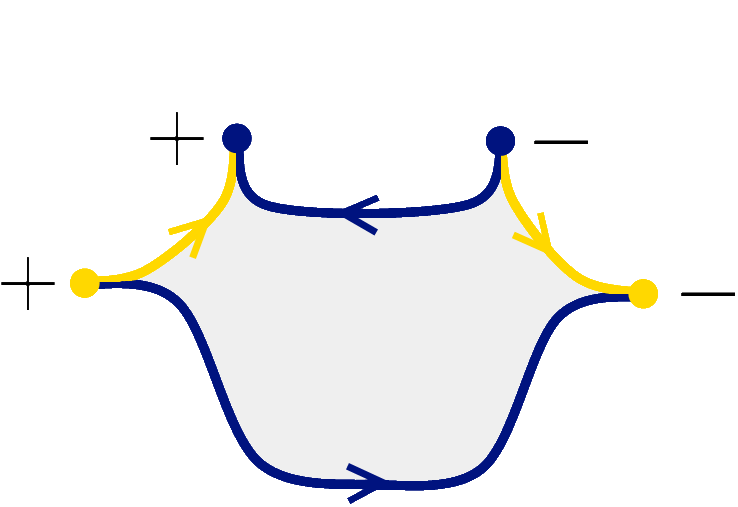}
		\caption{A ninja star.}\label{fig:important relation star}
	\end{subfigure} \quad 
	\begin{subfigure}[b]{0.3\textwidth}
		\centering\includegraphics[scale=.15]{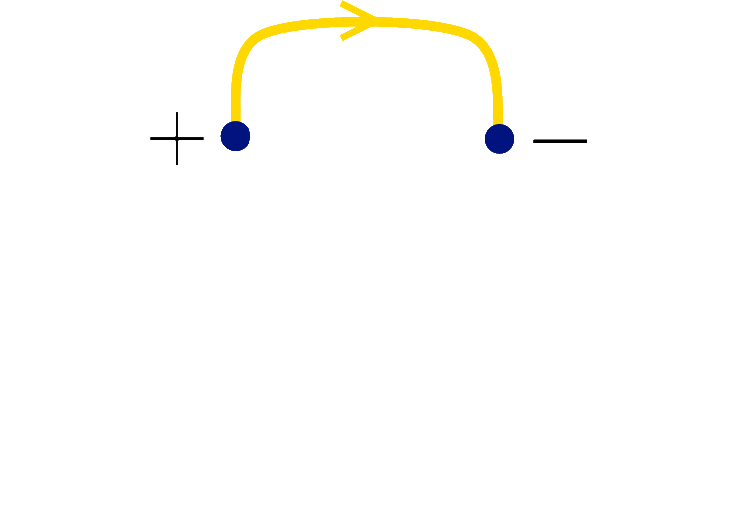}
		\caption{A one-morphism.}\label{fig:important relation cap}
	\end{subfigure} \quad
	\begin{subfigure}[b]{0.3\textwidth}
		\centering\includegraphics[scale=.15]{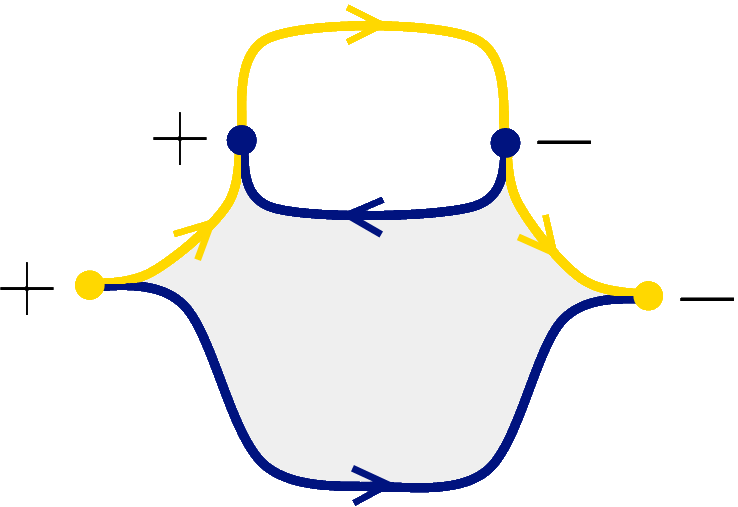}
		\caption{Another 2-morphism.}\label{fig:important relation horizontal composition}
	\end{subfigure}
	\end{center}
	\caption{The 2-morphism on the right is the result of the horizontal compositon of the 2-morphism shown on the left with the identity 2-morphism on the cap in the middle.}
\end{figure}

\noindent The relation drawn in Figure \ref{fig:important relation star} describes what should happen if we fill in the hole by vertical composition with the disc, viewed as a 2-morphism as indicated in Figure~\ref{fig:important relation disc}.
The disc that we fill in corresponds to the lower two-morphism in Figure~\ref{fig:important relation diagram}.

\begin{figure}[h]
	\begin{center}
	\includegraphics[scale=.15]{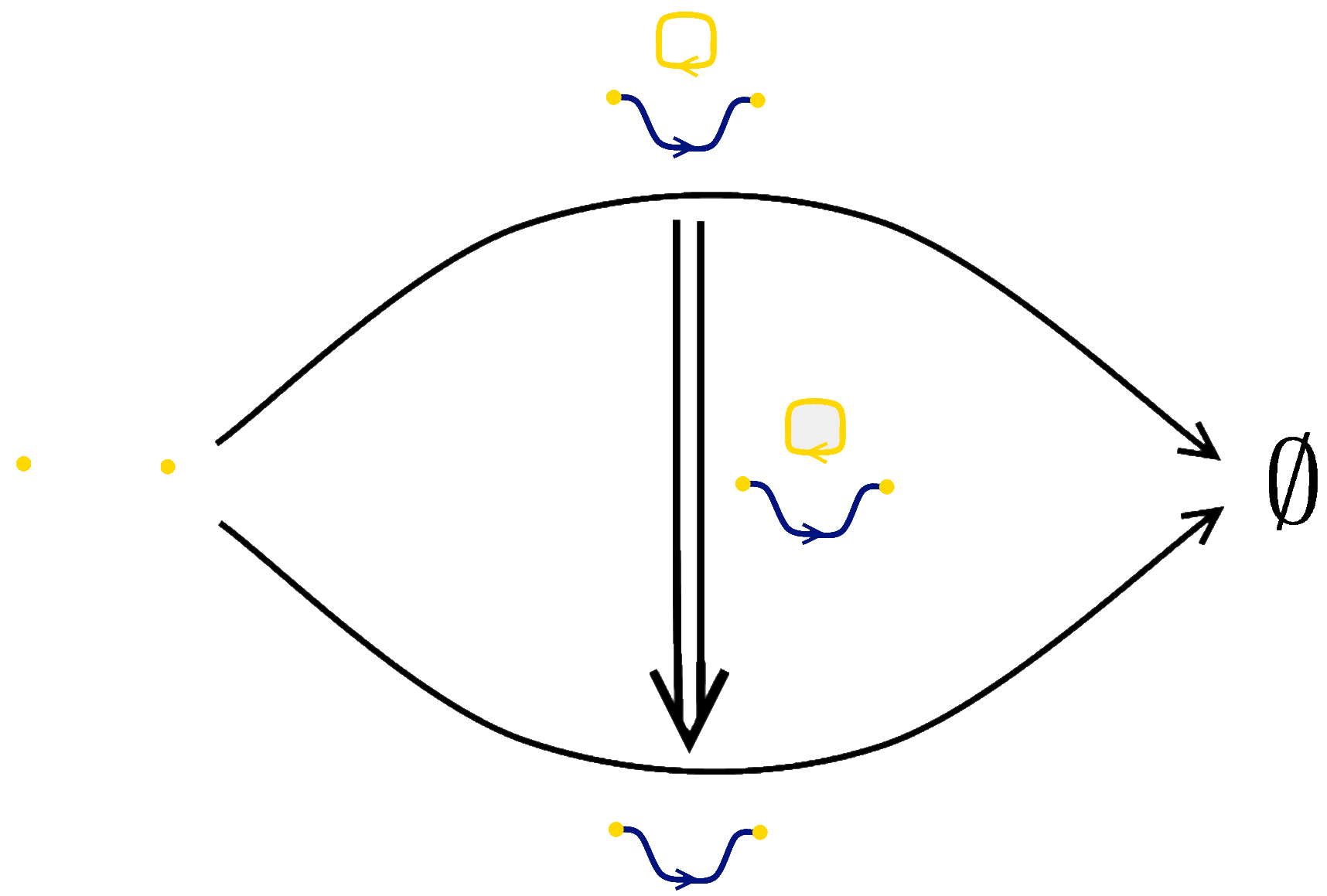}
	\end{center} \caption{The relation from Figure~\ref{fig:important relation} involves a horizontal composition with the tensor product of the 2-morphism corresponding to the disc and the identity 2-morphism on the (blue) 1-morphism.} \label{fig:important relation disc}
\end{figure}

\begin{figure}[h]
	\begin{center}
	\includegraphics[scale=.15]{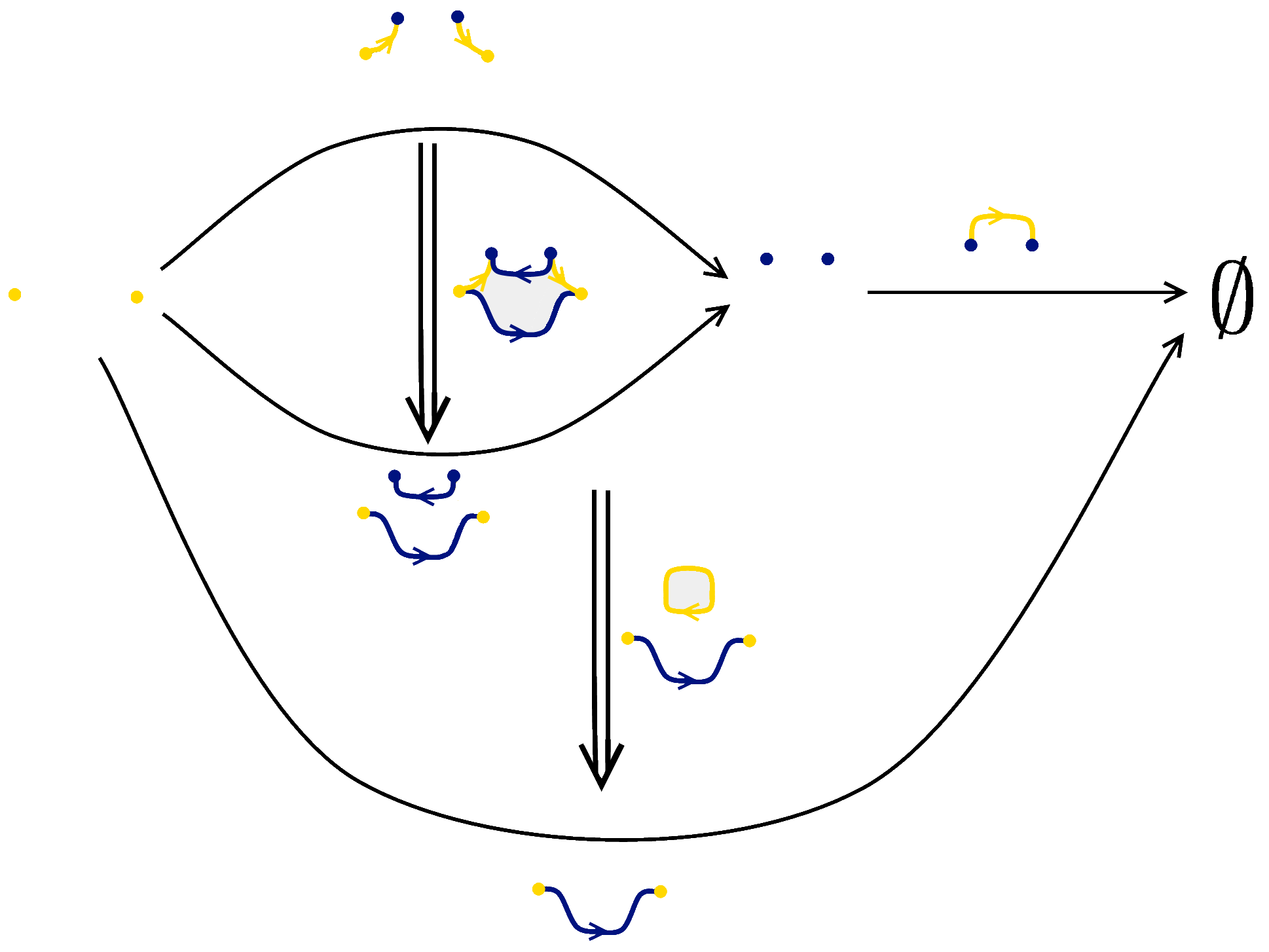}
	\end{center}\caption{The diagram showing the objects, 1-morphisms and 2-morphism featuring in the important relation from Figure~\ref{fig:important relation}.} \label{fig:important relation diagram}
\end{figure}

Now, any surface can be decomposed into discs and ninja-stars via a simple algorithm: draw closed curves with transverse intersections on the surface, and then replace those intersections by ninja stars (see Figure~\ref{fig:surface decomposition}). Given the bimodule map associated to the ninja star and the relation from Figure \ref{fig:important relation}, this decomposition should allow one to construct the full extended \textsc{cft} from it.

\begin{figure}[h]
	\begin{center}
	\includegraphics[scale=.20]{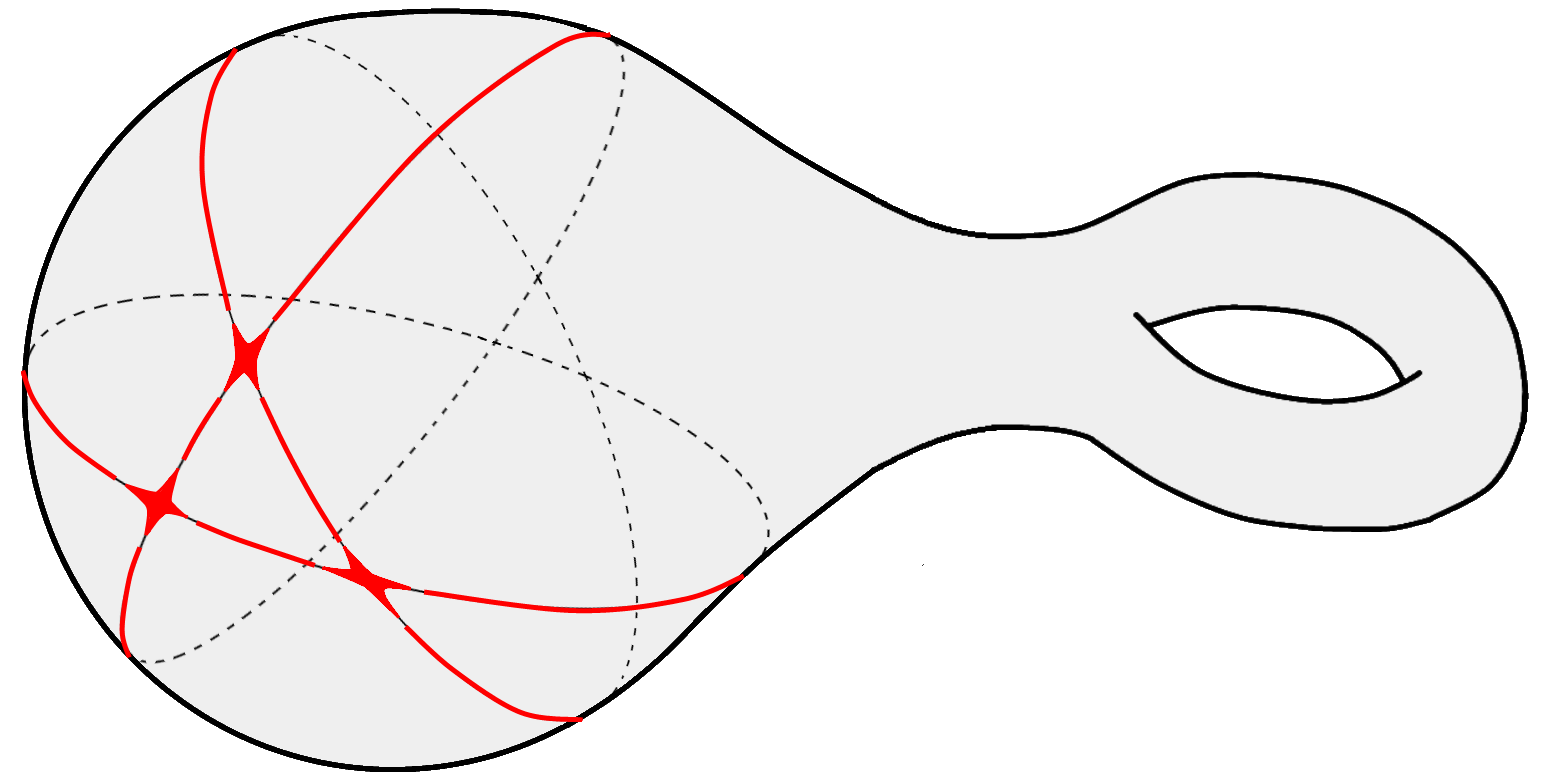}
	\end{center}\caption{By drawing closed curves on a 2-surface and replacing their junctions by ninja stars, every surface can be decomposed into discs, ninja stars, and intervals.}\label{fig:surface decomposition}
\end{figure}

\section*{Acknowledgements}

I am greatly indebted to Stephan Stolz for inviting me to give these lectures,
and thus providing the opportunity for this material to get written.
I am also very grateful to my student Jules Lamers for compiling a first draft of these notes,
and for drawing all the pictures.

	\bibliography{references}
\bibliographystyle{amsplain}

\end{document}